\documentclass[a4paper,onecolumn,11pt,accepted=2024-12-20]{quantumarticle}
\pdfoutput=1

\usepackage[breaklinks=true]{hyperref}
\hypersetup{colorlinks,citecolor=blue,allcolors=blue}

\usepackage{enumitem}
\usepackage{amsfonts}
\usepackage[english]{babel}
\usepackage[utf8]{inputenc}
\usepackage{graphicx}
\usepackage{mathrsfs}
\usepackage[mathscr]{eucal}
\usepackage[numbers,compress]{natbib}

\usepackage{mathtools}
\usepackage{bm}
\usepackage{adjustbox}
\usepackage{subfig}
\usepackage{multirow}

\usepackage{tabularx,colortbl}
\usepackage{colortbl}
\usepackage{listings}
\usepackage[capitalize]{cleveref}
\usepackage{soul}
\usepackage{url}
\usepackage{setspace}
\usepackage{makecell}
\usepackage{algorithm}
\usepackage{algorithmic}
\usepackage{amsthm}

\fancyheadoffset{0cm} 

\def\BibTeX{{\rm B\kern-.05em{\sc i\kern-.025em b}\kern-.08emT\kern-.1667em\lower.7ex\hbox{E}\kern-.125emX}}

\usepackage{thmtools}

\declaretheorem[numberwithin=section]{theorem}
\declaretheorem[sibling=theorem]{lemma}

\declaretheorem[style=definition]{definition}

\declaretheorem[numbered=no,style=definition,name=Fact]{fact*}

\newtheorem*{theorem*}{Theorem}
\newtheorem*{corollary*}{Corollary}
\newtheorem*{proposition*}{Proposition}
\newtheorem*{lemma*}{Lemma}
\newtheorem*{claim*}{Claim}
\newtheorem*{problem*}{Problem}

\newcommand{\R}{\mathbb{R}}

\newcommand{\tsr}[1]{\boldsymbol{\mathscr{#1}}}
\newcommand{\Set}[1]{\mathcal{#1}}

\newcommand{\mat}[1]{\mathbf{#1}}

\newcommand{\ktdist}[1]{d_{\text{KT}}\left(#1\right)}
\newcommand{\cut}{{\mathsf{cut}}}
\newcommand{\mincut}{{\mathsf{mincut}}}
\newcommand{\pth}{{\mathsf{path}}}
\newcommand{\tn}{{\mathsf{tn}}}

\newcommand{\name}[1]{{\color{blue}[name] }}

\newcommand {\comment}[1]{{\small{\textit{//   #1}}}}
\usepackage[normalem]{ulem}

\definecolor{mygreen}{rgb}{0,0.2,0}
\definecolor{mygray}{rgb}{0.5,0.5,0.5}
\definecolor{mymauve}{rgb}{0.58,0,0.82}
\definecolor{mypurple}{rgb}{0.38,0,0.32}
\definecolor{myblue}{rgb}{0.1,0,0.32}

\begin{document}

\title{Approximate contraction of arbitrary tensor networks with a flexible and efficient density matrix algorithm
}

\author{Linjian Ma}
\affiliation{Department of Computer Science, University of Illinois Urbana-Champaign, Urbana, IL 61801, USA}
\email{lma16@illinois.edu}
\author{Matthew T. Fishman}
\affiliation{Center for Computational Quantum Physics, Flatiron Institute, New York, New York 10010, USA}
\email{mfishman@flatironinstitute.org}
\author{E. M. Stoudenmire}
\affiliation{Center for Computational Quantum Physics, Flatiron Institute, New York, New York 10010, USA}
\email{mstoudenmire@flatironinstitute.org}
\author{Edgar Solomonik}
\affiliation{Department of Computer Science, University of Illinois Urbana-Champaign, Urbana, IL 61801, USA}
\email{solomon2@illinois.edu}

\maketitle

\begin{abstract}

Tensor network contractions are widely used in statistical physics, quantum computing, and computer science.
We introduce a method to efficiently approximate tensor network contractions using low-rank approximations, where each intermediate tensor
generated during the contractions is approximated as a low-rank
binary tree tensor network. The proposed algorithm has the flexibility to incorporate a large portion
of the environment when performing low-rank approximations, which can lead to high accuracy for a given rank. Here, the environment
refers to the remaining set of tensors in the network, and low-rank approximations with larger environments
can generally provide higher accuracy.
For contracting tensor networks defined on lattices, the proposed algorithm can be viewed as a generalization of the standard boundary-based algorithms.
In addition, the algorithm includes a cost-efficient density matrix algorithm for approximating a tensor network with a general graph structure into a tree structure, whose computational cost is asymptotically upper-bounded by that
of the standard algorithm that uses canonicalization.
Experimental results indicate that the proposed technique outperforms previously proposed approximate tensor network contraction algorithms for multiple problems
in terms of both accuracy and efficiency.

\end{abstract}

\section{Introduction}
\label{sec:intro}

A tensor network~\cite{orus2014practical,vidal2003efficient} uses a set of (small) tensors, where some or all of their modes are contracted according to some pattern, to implicitly represent the structure of high-dimensional tensors that are expensive to form explicitly.
Tensor network techniques have been widely used in computational quantum physics~\cite{vidal2003efficient,verstraete2008matrix,white1992density,verstraete2004renormalization,shi2006classical,schollwock2005density},
where low-rank tensor networks can be used to both represent Hamiltonians and quantum states.
These techniques are also applied in
multiple other applications, including quantum circuit simulation~\cite{pan2020contracting,guo2019general,pang2020efficient,ma2022low,zhou2020limits},  data mining via tensor methods~\cite{kolda2009tensor,cichocki2014tensor}, machine learning~\cite{stoudenmire2016supervised,reyes2021multi,li2020understanding}, and so on.

The \textit{tensor network contraction} operation explicitly evaluates the single tensor represented by a given tensor network, and it has multiple applications. In quantum computing, each quantum
circuit execution
can be viewed as a tensor network contraction, making this method a useful tool for simulating quantum computers~\cite{ma2022low,zhou2020limits,pang2020efficient,pan2020contracting}.
In statistical physics, tensor network contraction has been used to evaluate the classical partition function of physical models defined on specific graphs~\cite{levin2007tensor}.
Tensor network contraction has also been used for counting satisfying assignments
of constraint satisfaction problems (\#CSPs)~\cite{kourtis2019fast}.  In this approach, an arbitrary \#CSP formula is transformed into a tensor network, where its full contraction yields the number of satisfying assignments of that formula.
Tensor network contraction is typically achieved through a sequence of pairwise tensor contractions. This sequence, known as the \textit{contraction path}, is determined by a topological sort of the underlying \textit{contraction tree}. The contraction tree is a rooted binary tree that depicts the complete contraction of the tensor network. In this tree, the leaves correspond to the tensors in the network, and each internal vertex represents the tensor contraction of its two children.

In the general case, contracting tensor networks with arbitrary structure is \#P-hard because of the potential production of intermediate tensors with high orders or large modes, leading to significant computational costs for accurate contraction \cite{damm2002complexity,o2019parameterization,biamonte2015tensor}.
Nonetheless, in some applications such as many-body physics, it has been observed that tensor networks built on top of specific models can often be approximately contracted with satisfactory accuracy, without incurring large computational costs~\cite{orus2019tensor}.

 \begin{figure}[htb]
\centering

\subfloat[MPS]{
\hspace{10mm}
\includegraphics[width=.2\textwidth, keepaspectratio]
{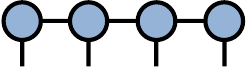}\hspace{10mm}
\label{subfig:mps_vis}}
\subfloat[Binary tree tensor network]{\hspace{10mm}\includegraphics[width=.2\textwidth, keepaspectratio]{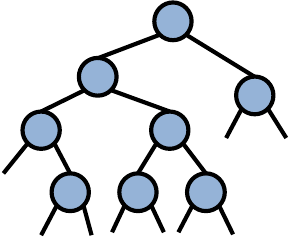}
\hspace{10mm}
\label{subfig:binary_tree_tn}}
\subfloat[TTNS]{\hspace{10mm}\includegraphics[width=.2\textwidth, keepaspectratio]{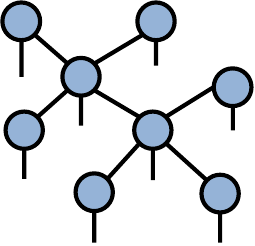}
\hspace{10mm}
\label{subfig:ttn}}

\caption{
Illustration of the matrix product state (MPS), the (full) binary tree tensor network, and the tree tensor network state (TTNS). MPS is a maximally-unbalanced binary tree tensor network if contracting the tensor at one end with its neighbor. Both MPS and the binary tree tensor network are special cases of TTNS, where each tensor has an order of at most 3.}
\label{fig:trees}
\end{figure}

A common approach to approximately contract a tensor network is to approximate large intermediate tensors as (low-rank) tensor networks, which reduces the memory usage and computational overhead for subsequent contractions. Widely used tensor networks for approximation including the matrix product state (MPS~\cite{verstraete2008matrix}, also called tensor train~\cite{oseledets2011tensor}), the binary tree tensor network~\cite{shi2006classical}, and the  tree tensor network state (TTNS)~\cite{nakatani2013efficient,murg2015tree,felser2021efficient}, which are visualized in \cref{fig:trees}.
For tensor network contractions defined on regular structures, such as projected entangled pair states (PEPS) with 2D lattice structures~\cite{verstraete2004renormalization,verstraete2008matrix}, many efficient approximate contraction algorithms based on MPS approximations~\cite{lubasch2014unifying,lubasch2014algorithms} have been proposed.
However, many of these methods have not been extended to other general tensor network structures.

 \begin{figure}[htb]
\centering
\includegraphics[width=.9\textwidth, keepaspectratio]{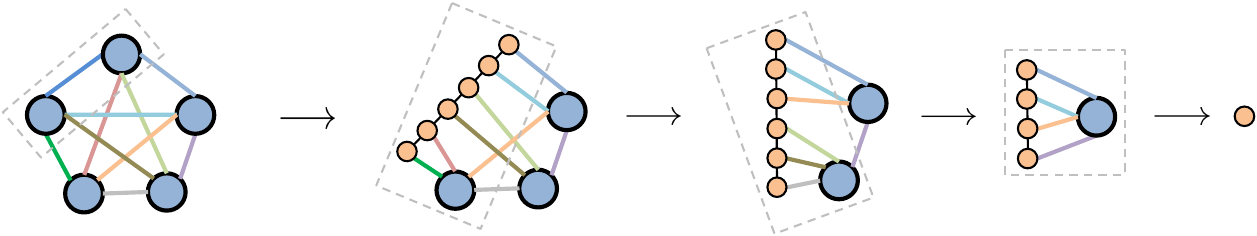}

\caption{
Illustration of the approximate contraction technique used in \cite{jermyn2020automatic,pan2020contracting,chubb2021general}. Each intermediate is approximated as an MPS, which has an unbalanced binary tree structure. The left diagram is the tensor diagram of the input tensor network.
Each dashed box denotes the part of the tensor network that is approximated as an MPS.
}
\label{fig:contract_illustrate}
\end{figure}

Recent works have proposed automated approximation algorithms for contracting tensor networks with more general graph structures~\cite{jermyn2020automatic,pan2020contracting,chubb2021general,gray2022hyper,sahu2022efficient,alkabetz2021tensor}, and many  of these methods employ low-rank approximation/truncation techniques.
In \cite{jermyn2020automatic,pan2020contracting,chubb2021general}, each intermediate tensor produced during the contraction is approximated as a binary tree tensor network, and we illustrate this approach in \cref{fig:contract_illustrate}.
In particular, \cite{jermyn2020automatic} approximates each intermediate tensor as a general binary tree tensor network,
while the algorithm proposed in \cite{pan2020contracting} called ``contracting arbitrary tensor network'' (CATN) approximates each intermediate tensor as an MPS.
When contracting two MPSs, CATN swaps/permutes the modes that connect both MPSs to the boundaries. Then, it contracts these modes to obtain the output MPS.
The adjacent mode swaps are the bottleneck for complexity in CATN.
In another algorithm proposed in \cite{chubb2021general} called ``SweepContractor", each intermediate tensor is also approximated as an MPS,
and the algorithm leverages an embedding of the tensor network graph into 2D space to find an effective contraction path.

Several factors can significantly impact the efficiency and accuracy of the approximate tensor network contraction process.
To begin with, the choice of contraction path plays a crucial role. Ref.~\cite{gray2022hyper} demonstrates that selecting different contraction paths using various heuristics can lead to substantial variations in both runtime and accuracy for different problems.
Additionally, for both CATN~\cite{pan2020contracting} and SweepContractor~\cite{chubb2021general},
it is essential to carefully select the binary tree/MPS structures and permutations (i.e., a mapping from tensor modes onto binary tree vertices)~\cite{li2022permutation}. These choices should yield accurate low-rank approximations while enabling efficient subsequent contractions. However, previous works such as \cite{jermyn2020automatic,pan2020contracting,chubb2021general} have not systematically explored these parts of the design space.

The low-rank truncation algorithm used to reduce the tensor size in approximate contraction is another important factor.
Let $\mat{M}$ represent the part of the network that requires approximation, and let $\mat{E}$ denote the remaining set of tensors in the network, which is commonly referred to as the \textit{environment}.
The optimal way to truncate is to minimize the global error by solving $\min_{\mat{X}}\|\mat{E}\mat{X} - \mat{E}\mat{M}\|_F$ with the constraint that $\mat{X}$ has a specific low-rank tensor network structure, where $\|\cdot\|_F$ denotes the Frobenius norm.
Two standard algorithms for solving the low-rank approximation problem are the canonicalization-based algorithm and the density matrix algorithm. In the canonicalization-based algorithm, one first performs a QR decomposition on $\mat{E}$, $\mat{Q},\mat{R}\leftarrow \texttt{QR}(\mat{E})$, then updates $\mat{X}$ based on the low-rank approximation of $\mat{Q}^T\mat{M}$.
In the density matrix algorithm, the leading eigenvectors of the density matrix (also called the Gram matrix/normal equations), $\mat{M}^T\mat{E}^T\mat{E}\mat{M}$, is computed, and $\mat{X}$ is computed by projecting $\mat{M}$ to the subspace spanned by the leading eigenvectors. Both algorithms have the same output but can have different computational costs.

If the environment tensor network $\mat{E}$ contains a large number of tensors, minimizing the global error could be computationally expensive. In such cases, one typically resorts to minimizing the local error by solving $\min_{\mat{X}}\|\mat{X} - \mat{M}\|_F$, or by replacing $\mat{E}$ with a smaller environment $\hat{\mat{E}}$ so the optimization problem is easier to solve.

Achieving a balance between accuracy and efficiency requires favoring different structures and sizes of the environment $\hat{\mat{E}}$ for different problems. Hence, it becomes crucial to provide an automated tensor network contraction algorithm with the necessary flexibility to accommodate different environments. This flexibility enables the algorithm to adapt and optimize the contraction process according to the specific requirements of each problem.

In previous studies \cite{pan2020contracting, chubb2021general}, the selection of environments was implicitly determined by the algorithm.
For instance, in the CATN algorithm \cite{pan2020contracting}, truncation takes place during adjacent swaps of MPS modes, with the environment consisting of all tensors in the target MPS. Similarly, the SweepContractor algorithm \cite{chubb2021general} performs truncation while contracting an input MPS with a single tensor, incorporating both the MPS and the tensor into the environment.
The method proposed in \cite{gray2022hyper}  introduces user-specified environment sizes, and utilizes tree-structured environments $\hat{\mat{E}}$ that are constructed by including a
spanning tree of tensors around the pair of tensors to be truncated. Ref. \cite{gray2022hyper} demonstrates that including a larger environment leads to more accurate contraction results for multiple problems.
In this work, we generalize the strategies presented in the previous works and propose a tensor network contraction algorithm that allows more flexible environment incorporation.
For contracting tensor networks defined on lattices, the proposed strategy can be viewed as a generalization of the standard boundary-based algorithms~\cite{verstraete2004renormalization}.

\subsection{Our contributions}

We
propose a new approach, \texttt{partitoned\_contract}, for performing approximate contractions of arbitrary tensor networks.
We illustrate the approach in \cref{fig:ovw}. This approach follows the technique used in \cite{jermyn2020automatic,pan2020contracting,chubb2021general}, where each intermediate tensor produced during the contraction is approximated as a binary tree tensor network.
Moreover, our approach is composed of the following two novel components.

 \begin{figure}[htb]
\centering

\subfloat[Complete contraction tree]{
\hspace{18mm}
\includegraphics[width=.2\textwidth, keepaspectratio]
{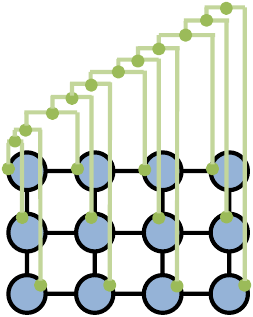}\hspace{18mm}\label{subfig:peps1}}
\subfloat[Contraction tree on the partitioned network]{\hspace{18mm}\includegraphics[width=.2\textwidth, keepaspectratio]{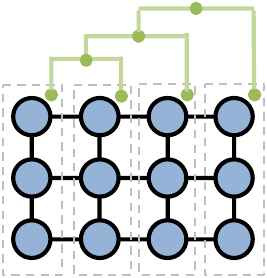}
\hspace{18mm}
\label{subfig:peps2}}

\caption{
Illustration of different contraction trees. Each blue vertex denotes a tensor, and the green lines and dots denote the binary contraction tree.
The contraction tree visualization has been adapted from \cite{gray2022hyper}.
In (b), each dotted box denotes a partition of the tensor network. The partial contraction sequence shown in (b) corresponds to a standard left-to-right boundary MPS contraction \cite{verstraete2004renormalization}.}
\label{fig:contract_peps}
\end{figure}

First, unlike prior works~\cite{jermyn2020automatic,pan2020contracting,chubb2021general,gray2022hyper} that contract the tensor network based on a \textit{complete} contraction tree with each  leaf corresponding to a tensor in the network, our technique relies on a contraction tree of parts
of the tensor network, which is a
    \textit{partial} contraction tree and each leaf vertex corresponds to a partition.
    We illustrate complete and partial contraction trees in \cref{fig:contract_peps}.
    In the algorithm, each low-rank approximation considers all tensors in the input partitions as the environment, thus utilizing a larger partition means using a larger environment and can potentially lower the truncation error.
    In practical applications, one has the option of either utilizing automated graph partitioning libraries like KaHyPar~\cite{schlag2023high} and Metis~\cite{karypis1997metis} for partitioning the tensor network,
    or manually selecting suitable partitions for specific problems.
    In \cref{subsec:2d}, we will demonstrate how the utilization of the partial contraction tree abstraction enables the straightforward extension of various contraction algorithms designed for 2D grids with different environments, including those that have not been automated in the prior work \cite{jermyn2020automatic, pan2020contracting, chubb2021general,gray2022hyper}.

 \begin{figure}[htb]
\centering

\subfloat [Illustration of the \texttt{partitoned\_contract} algorithm]{\includegraphics[width=.98\textwidth, keepaspectratio]{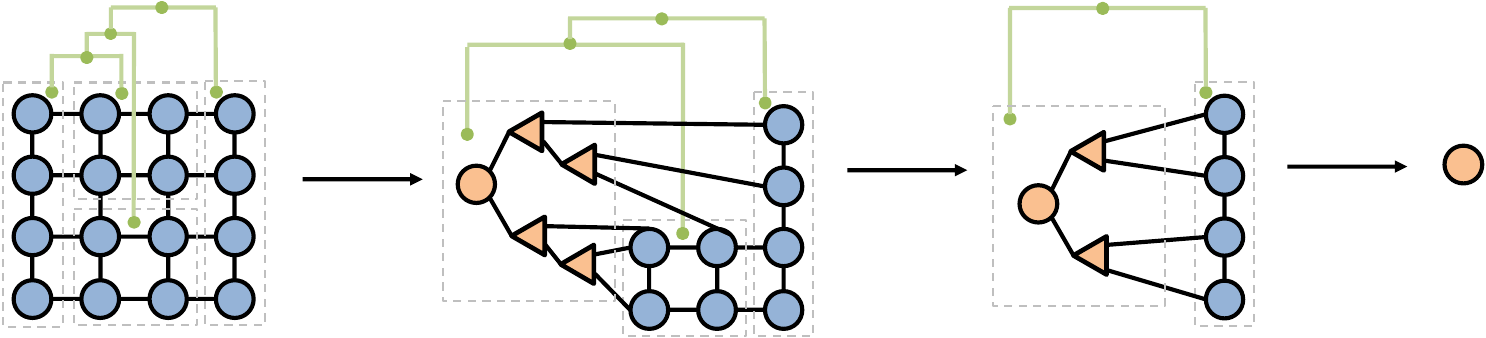}
\label{fig:pipeline1}}

\vspace{3mm}

\subfloat [Illustration of the process to approximate the contraction of two partitions into a binary tree tensor network]{\includegraphics[width=.98\textwidth, keepaspectratio]{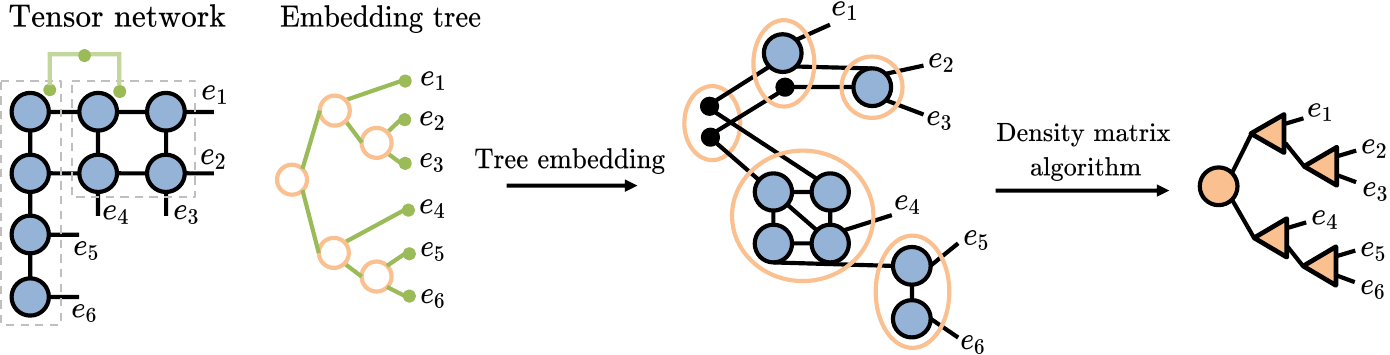}
\label{fig:pipeline}}

\caption{(a) Illustration of the \texttt{partitoned\_contract} algorithm.
The algorithm takes as inputs a tensor network, a partitioning of that tensor network, and a partial contraction tree. The algorithm proceeds by traversing the partial contraction tree and approximately contracting a pair of tensor network partitions into a binary tree tensor network.
(b) Illustration of the process to approximate the input tensor network (left diagram) into a binary tree tensor network (right diagram). The embedding tree is a rooted binary tree that represents the output tree structure. The tree embedding step maps a partition of the input tensor network to each non-leaf (orange) vertex in the embedding tree.
Finally, the density matrix algorithm (or the canonicalization-base algorithm)
approximates the embedded tensor network into a binary tree tensor network. Each black dot in the diagrams represents an identity matrix.
}
\label{fig:ovw}
\end{figure}

    Second, we provide a new approach to approximate a given tensor network into a binary tree structure, as depicted in \cref{fig:pipeline}.
    This approach is composed of the following three novel components.
    \begin{itemize}
        \item It encompasses a new heuristic for generating binary tree structures and permutations (i.e., a mapping from tensor modes onto binary tree vertices \cite{li2022permutation}) of intermediate tensor networks.
        The binary tree structure is also called the embedding tree in \cref{fig:pipeline} and throughout the paper.
Unlike previous studies that relied on arbitrary choices for such structures and permutations, our approach takes into consideration the efficiency of subsequent contractions. This is achieved by ensuring that the embedding tree aligns with a contraction path-generated tree, which imposes constraints on the adjacency relations of binary tree modes.
Moreover, we ensure that the selected structure is similar to the given sub-tensor network by solving a graph embedding problem that minimizes the congestion~\cite{hruska2008tree,bienstock1990embedding,matsubayashi2015separator,bezrukov2000congestion,manuel2009exact},
allowing for an accurate approximation with low ranks in the resulting tree tensor network.
 The details of the algorithm can be found in \cref{sec:binary_tree}.

\item It includes a density matrix algorithm to approximate a given tensor network into the target embedding tree. The algorithm uses a sequence of density matrix algorithms for low-rank approximation to output the embedding tree tensor network, and
includes all tensors in the input tensor network as the environment.
 When compared to the canonicalization-based algorithm that employs the same environment,
 the density matrix algorithm exhibits the same or lower asymptotic cost, making it more efficient.
In particular, the density matrix algorithm exhibits the potential to significantly reduce the asymptotic cost when dealing with large environment sizes. The detail of the algorithm can be found in \cref{sec:dm}.

The tensor network contraction framework proposed in \cite{gray2022hyper} also offers the capability to handle large environments. However, the framework in \cite{gray2022hyper} approximates the environments as trees from the outset by cutting certain bonds in the environment, which ignores certain loop correlations in the environment. In contrast, our density matrix algorithm directly works with the full environment including loops and then approximates the result of the contraction as a tree, which for a given environment should be more accurate but potentially more computationally expensive.

\item
In scenarios where the mode ordering of the selected tree structure intended for efficient later contractions does not align with the input structure, our approach employs a hybrid algorithm that integrates the density matrix algorithm and a swap-based algorithm to perform the tree approximation.
Swap-based algorithms, extensively utilized in MPS-based tensor network contraction algorithms such as when applying long-range gates \cite{stoudenmire2010minimally} and in other general approximate contraction algorithms like CATN and SweepContractor, use a sequence of adjacent swaps of MPS  modes to permute the ordering of the MPS tensors.
Within our algorithm, a  sequence of local swap operations are performed using the density matrix algorithm, each time progressively modifying the structure by a small amount to ensure that the overall cost remains manageable.
The detail of the algorithm can be found in \cref{sec:swap}.

    \end{itemize}

In \cref{sec:exp}, we assess the performance of the proposed algorithm. Regarding the sub-problem of approximating a general tensor network into a tree tensor network, our experimental results show the superior efficiency of the density matrix algorithm compared to the canonicalization-based algorithm when applied to multiple input tensor network structures. These empirical findings consistently align with our theoretical analysis.

To evaluate the efficacy of our contraction algorithm, we conduct experiments on various tensor network structures. The results demonstrate that by leveraging environments and employing the density matrix algorithm, we achieve significant reductions in overall execution time and improvements in accuracy when dealing with tensor networks defined on lattices and random regular graphs. Notably, our algorithm outperforms both the CATN algorithm proposed in \cite{pan2020contracting} and the SweepContractor proposed in \cite{chubb2021general} when considering tensor networks defined on lattices representing the classical Ising model. Specifically, our approach achieves an order of magnitude
speed-up in execution time while maintaining the same level of accuracy. This
improvement in speed demonstrates the efficiency of our
approach.

\subsection{Organization}

This paper is organized as follows. In \cref{sec:defn,sec:background}, we introduce the definitions, the computational cost model, and background for the proposed algorithm.
\cref{sec:method} provides an overview of the proposed approximate tensor network contraction algorithm, with detailed components discussed in \cref{sec:binary_tree,sec:dm,sec:swap}.
In \cref{sec:exp}, we present the results of a series of experiments to evaluate the performance of the proposed algorithm.

\section{Definitions and the computational cost model}
\label{sec:defn}

\subsection{Tensor network definitions}

We introduce the tensor network notation here.
The structure of a tensor network can be described by an undirected graph $G=(V,E)$, where each tensor of the tensor network is associated with a vertex in $V$ and each mode of the tensors is associated with an edge in $E$.
We refer to edges with a dangling end (one end not adjacent to any vertex) as uncontracted edges, and those without dangling ends as contracted edges.
We use $w$ to denote an edge weight function such that for each edge $e\in E$, $w(e) = \log(s)$ is the natural logarithm of the mode size $s$ associated with an edge $e$.
For an edge set $E$, we use $w(E)=\sum_{e\in E}w(e)$ to denote the weighted sum of the edge set. The weight $w(E)$ is related to the cost of contracting neighboring tensors along the modes associated with the edge set $E$, which will be discussed in more detail in the next section.

\subsection{The computational cost model}\label{subsec:cost_model}

We summarize the computational cost model used throughout the paper. We assume that all tensors in the tensor network are dense.
The contraction of two general dense tensors $\tsr{A}$ and $\tsr{B}$, represented as vertices $v_a$ and $v_b$ in $G=(V,E)$, can be cast as a matrix multiplication, and the overall asymptotic cost is
\begin{equation}
\Theta\left(\exp\left(w(E(v_a)) + w(E(v_b)) - w(E(v_a,v_b))\right)\right),
\end{equation}
where $E(v_a),E(v_b)$ denotes the edges adjacent to $v_a$, $v_b$, respectively, and $E(v_a,v_b)$ denotes the edge connecting $v_a$ and $v_b$.
Above we assume the classical matrix multiplication algorithm is used rather than fast algorithms such as Strassen's algorithm~\cite{strassen1969gaussian}.

To canonicalize the tree tensor network, a series of QR factorizations is employed. Given a matrix $\mat{A}\in \R^{m\times n}$, performing the QR factorization incurs an asymptotic cost of $\Theta(mn\cdot \min(m, n))$. In order to reduce the bond size or rank within the tensor network, we utilize low-rank factorization. Given a matrix $\mat{A}\in \R^{m\times n}$, low-rank factorization aims to find two matrices, $\mat{B}\in \R^{m\times r}$ and $\mat{C}\in \R^{r\times n}$, with $r$ being less than the minimum of $m$ and $n$, while minimizing the Frobenius norm $\|\mat{A}-\mat{BC}\|_F$. In our cost analysis, we assume the use of the standard low-rank factorization algorithm that employs a rank-revealing QR factorization~\cite{gu1996efficient}. The asymptotic cost of this algorithm is $O(mnr)$.

\section{Background}
\label{sec:background}

This section offers background for the proposed approach. In \cref{subsec:survey_tn_structure}, we provide a short survey of several common tensor networks discussed in the paper.
In \cref{subsec:can_dm},
we review both the canonicalization-based algorithm and the density matrix algorithm for low-rank approximation of tensor networks.
This review serves as motivation for the density matrix algorithm explained in detail in \cref{sec:dm}.

We cover additional backgrounds in the appendix.
\cref{subsec:swap} covers the standard swap-based algorithm used to permute MPS modes, which serves as a motivation for our  algorithm that combines the density matrix algorithm and the swap-based algorithm, as outlined in \cref{sec:swap}. Furthermore, in  \cref{subsec:recur_bisection}, we delve into the definition and heuristics of the graph embedding problem, which is utilized in \cref{sec:binary_tree} to select an efficient binary tree structure.

\subsection{A survey of common tensor network structures}\label{subsec:survey_tn_structure}

We survey both tree tensor networks and tensor networks defined on lattices. The matrix product state (MPS)~\cite{verstraete2008matrix,oseledets2011tensor}, a binary tree tensor network~\cite{shi2006classical}, and a general tree tensor network state (TTNS)~\cite{nakatani2013efficient,murg2015tree,felser2021efficient} are illustrated in \cref{fig:trees}.
An MPS is a tensor network with a linear structure, with each tensor having one uncontracted mode.
The binary tree tensor network has a rooted binary tree structure, and all non-root vertices have an order of three. In a general TTNS, each tensor can have uncontracted modes, and the network has a general tree structure.

In this work, we focus on discussing both MPS and the binary tree tensor network. These networks are considered as special cases of TTNS, where each tensor has a maximum order of three. This characteristic makes them more memory-efficient compared to more general TTNS, especially when considering a fixed rank $r$. When the uncontracted mode size $s$ is much smaller than $r$, each MPS tensor has a size of $O(sr^2)$. This memory requirement is more efficient than that of the general binary tree tensor network, whose tensor size is $O(r^3)$.

 \begin{figure}[htb]
\centering

\subfloat[MPO]{
\hspace{10mm}
\includegraphics[width=.2\textwidth, keepaspectratio]
{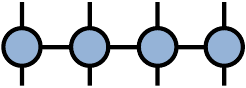}\hspace{10mm}
\label{subfig:mpo}}
\subfloat[PEPS]{\hspace{10mm}\includegraphics[width=.2\textwidth, keepaspectratio]{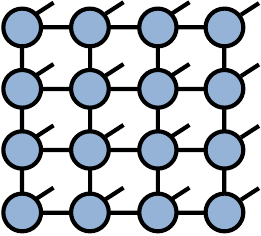}
\hspace{10mm}
\label{subfig:peps}}
\subfloat[3D lattice tensor network]{\hspace{10mm}\includegraphics[width=.2\textwidth, keepaspectratio]{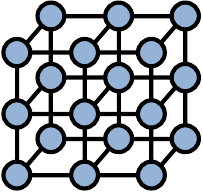}
\hspace{10mm}
\label{subfig:2d_lattice}}

\caption{
Illustration of the matrix product operator (MPO), the projected entangled pair states (PEPS), and the $3 \times 3 \times 2$ 3D lattice tensor network. }
\label{fig:tns}
\end{figure}

\cref{fig:tns} and \cref{fig:contract_peps} provide visual representations of other tensor networks, including the matrix product operator (MPO), the projected entangled pair states (PEPS)~\cite{verstraete2004renormalization,verstraete2008matrix}, and a closed tensor network defined on a 3D lattice.
In the 2D lattice tensor network, each row is either an MPS or an MPO. In the 3D lattice, each slice is either a PEPS or a PEPO. The PEPO has a similar structure to PEPS, but with the distinction that each tensor has two uncontracted edges.

\subsection{The canonicalization-based algorithm and the density matrix algorithm}\label{subsec:can_dm}

Let $\mat{A}\in \R^{b\times R},\mat{B}\in \R^{R\times c}$ denote two tensors in a tensor network, and let $\mat{E}\in \R^{a\times b}$ denote the environment tensor network. The low-rank approximation problem that is widely used in this work can be stated as
\begin{equation}
\min_{\hat{\mat{A}}\in \R^{b\times r},\mat{V}\in \R^{c\times r}} \left\|\mat{E}\mat{A}\mat{B} - \mat{E}\hat{\mat{A}}\mat{V}^T\right\|_F,
\quad \text{s.t. } \mat{V}^T\mat{V} = \mat{I},
\end{equation}
where $r<R$.
For the canonicalization-based algorithm, one first performs a QR decomposition on $\mat{EA}$ and gets $\mat{Q}\in \R^{a\times R}, \mat{R}\in \R^{R\times R}$, and then computes the right $r$ leading singular vectors of $\mat{R}\mat{B}$ to obtain $\mat{V}$.
For the density matrix algorithm, one first computes the Gram matrix (normal equations)  $\mat{L}=(\mat{EAB})^T\mat{EAB}$,
commonly known as the density matrix in the physics literature (and is at the heart of the original formulation of the density matrix renormalization group (DMRG) algorithm \cite{white1992density}),
and then computes the right $r$ leading singular vectors/eigenvectors of $\mat{L}$ to obtain $\mat{V}$.

For the case where $\mat{E}$ is a single matrix, both algorithms yield the same asymptotic cost with the computational cost introduced in \cref{subsec:cost_model}.
However, when $\mat{E}$ takes the form of a tensor network containing a large number of  tensors, the density matrix algorithm is more advantageous in terms of simplicity and efficiency. In particular, the density matrix $\mat{L}=(\mat{EAB})^T\mat{EAB}$ can be easily computed using the existing exact tensor network contraction algorithms, while orthogonalizing $\mat{EA}$ is usually hard when $\mat{E}$ does not have a tree structure. One potential approach for orthogonalizing $\mat{EA}$ involves directly performing orthogonalization on the matrix resulting from the contraction of $\mat{EA}$, but this method is inefficient when $\mat{E}$ is not path-like.

In \cref{subsec:cano_tree}, we review the canonicalization-based algorithm to reduce the mode sizes of tree tensor networks.
We will show in \cref{subsec:cost_analysis} that the cost of the density matrix algorithm is upper-bounded by the canonicalization-based algorithm.
In \cref{subsec:2d}, we provide a review of existing algorithms employed in truncating the MPO-MPS contraction, a common tensor network contraction and a special case of our more general algorithm.

\subsubsection{The canonicalization-based algorithm for truncating tree tensor networks}
\label{subsec:cano_tree}

We review the canonicalization-based algorithm to truncate a tree tensor network~\cite{zhang2020stability}.
We first introduce the canonical form in \cref{def:canonical_form}. For a given matrix $\mat{M}$ that is implicitly represented by a tree tensor network, its canonical form makes the whole tree orthogonal and uses another matrix to store the non-orthogonal part.

 \begin{figure}[!ht]
\centering
\includegraphics[width=.9\textwidth, keepaspectratio]{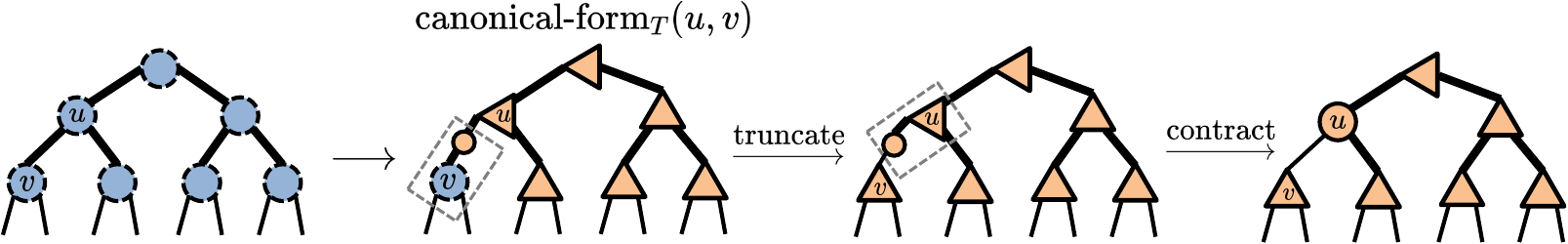}

\caption{
Illustration of truncating the mode represented by the edge $(u,v)$ through canonicalization.
}
\label{fig:canonical}
\end{figure}

\begin{definition}[Canonical form]\label{def:canonical_form}
    Consider a tensor network with a tree structure $T=(V,E)$.
    For a given vertex $u\in V$ and an edge $(u,v)$,
    let $S\subseteq V$ denote the vertices connected to $u$ when the edge $(u,v)$ is removed from $T$.
    $\texttt{canonical\_form}_T(u, v)$ means that all tensors represented by vertices in $S$ are orthogonalized towards the edge $(u,v)$, and a new vertex is added between $u$ and $v$ whose tensor contains the non-orthogonal part. An illustration of $\texttt{canonical\_form}_T(u, v)$ is in \cref{fig:canonical}.
\end{definition}

The canonicalization-based algorithm is shown in \cref{alg:canonize}.
It proceeds by computing the truncated network through a post-order depth-first search (DFS) traversal of the tree structure. At each vertex $v$, the algorithm constructs the canonical form around $v$ while truncating the edge connected to $v$. The resulting orthogonal tensor $\mat{U}_v$ is then computed. This iterative process continues until only the root vertex remains, which contains the comprehensive non-orthogonal information of the entire network.

\begin{algorithm}[!ht]
\setstretch{1.1}
\caption{The canonicalization-based algorithm for truncating a tree tensor network}
\begin{algorithmic}[1]
\STATE{\textbf{Input:} The tree tensor network $T= (V, E)$, the maximum mode size $\chi$, and the root vertex $r$
}
\STATE{$T_r\leftarrow$ a directed tree of $T$ with a root vertex $r$}
\FOR{each $v\in V\setminus \{r\}$ based on a post-order DFS traversal of $T_r$}
\STATE{$u\leftarrow \texttt{parent}(T_r,v)$}
\STATE{Change the tree tensor network to $\texttt{canonical\_form}_T(u,v)$ with the non-orthogonal matrix $\mat{R}_{u}$}
\STATE{$\mat{M}_v\leftarrow$ matricization of the tensor at $v$ with the mode connecting $u,v$ combined into the column}
\STATE{$\mat{U}_v\hat{\mat{R}}_{u}\leftarrow$ rank-$\chi$ approximation of $\mat{M}_v\mat{R}_u$ with $\mat{U}_v$ being orthogonal}
\STATE{Update the tensor at $u$ as $\hat{\mat{R}}_u\mat{M}_u$
}
\ENDFOR
\RETURN the tree tensor network that contains all $\mat{U}_v$ and the root tensor $\mat{M}_r$
\end{algorithmic}
\label{alg:canonize}
\end{algorithm}

\subsubsection{Existing algorithms for
truncating MPO-MPS multiplication}
\label{subsec:2d}

We provide a review of a set of algorithms to truncate the output of MPO-MPS multiplication. These algorithms are widely used in the boundary-based algorithm to approximately contract 2D lattice tensor networks like those surveyed in \cref{subsec:survey_tn_structure}. The boundary-based contraction algorithm initiates the process with a boundary MPS of the 2D network (e.g., the leftmost MPS in \cref{subfig:peps2}). At each step, the adjacent MPO is applied to the MPS and the result is approximated as a low-rank MPS.
The boundary-based contraction algorithm serves as the basis for motivating the proposed  partial contraction tree abstraction and the generalized density matrix algorithm for contracting arbitrary tensor networks.

Previous studies~\cite{paeckel2019time,itensornetwork} have explored various algorithms for MPO-MPS multiplication. These algorithms include approaches based on canonicalization~\cite{stoudenmire2010minimally,paeckel2019time}, the density matrix algorithm~\cite{itensornetwork,fishman2022itensor}, and the iterative fitting algorithm~\cite{verstraete2004renormalization,stoudenmire2010minimally}.
In this work,
we specifically concentrate on the first two types of algorithms. This choice is driven by the fact that both are one-pass algorithms and theoretical error bounds can be derived for the resulting output of both algorithms. The iterative fitting algorithm could have better scaling and lead to better performance in some cases but the use of that approach within our new algorithm is left for future work.

\paragraph{Algorithms that use canonicalization}

We review two different canonicalization-based algorithms: the zip-up algorithm~\cite{stoudenmire2010minimally} and the canonicalization algorithm with full environment~\cite{paeckel2019time}.
The zip-up algorithm uses a smaller environment compared to the other algorithm, which consider all tensors in the input MPO and MPS when performing truncations.
Throughout the analysis we use $r$ to denote the MPS rank, use $a$ to denote the MPO rank, and use $s$ to denote the size of all the other modes. The computational cost comparison between the algorithms is summarized in \cref{tab:compare}.

\begin{table}[!ht]
  \begin{center}
    \renewcommand{\arraystretch}{1.2}
    {
    \begin{tabular}{l|l|l|l}
\hline
       Algorithm  & Asymptotic cost & $s\ll a=\Theta(r)$ & $s= \Theta(a)\ll r$ 
       \\ \hline
      Zip-up  & $\Theta(N(s^2a^2r^2 + sar^3))$ & $\Theta(N(s^2r^4))$ &
      $\Theta(N(s^2r^3))$
      \\ \hline
      Canonicalization w/ full env &
      $\Theta(N(s^2a^2r^2+sa^3r^3))$
      &
      $\Theta(N(sr^6))$
      & $\Theta(N(s^4r^3))$
      \\ \hline
      Density matrix &  $\Theta(N(sa^2r^3+s^2a^3r^2+s^2ar^3))$  &
$\Theta(N(s^2r^5))$  & $\Theta(N(s^3r^3))$
\\ \hline
    \end{tabular}
    }
    \renewcommand{\arraystretch}{1}
  \end{center}
\caption{Comparison of asymptotic algorithmic complexity between the zip-up algorithm, the canonicalization-based algorithm that uses the full environment, and the density matrix algorithm. $s=\Theta(a)$ means $s$ is asymptotically bounded by $a$ both above and below.
}
\label{tab:compare}
\end{table}

 \begin{figure}[htb]
\centering
\includegraphics[width=.85\textwidth, keepaspectratio]{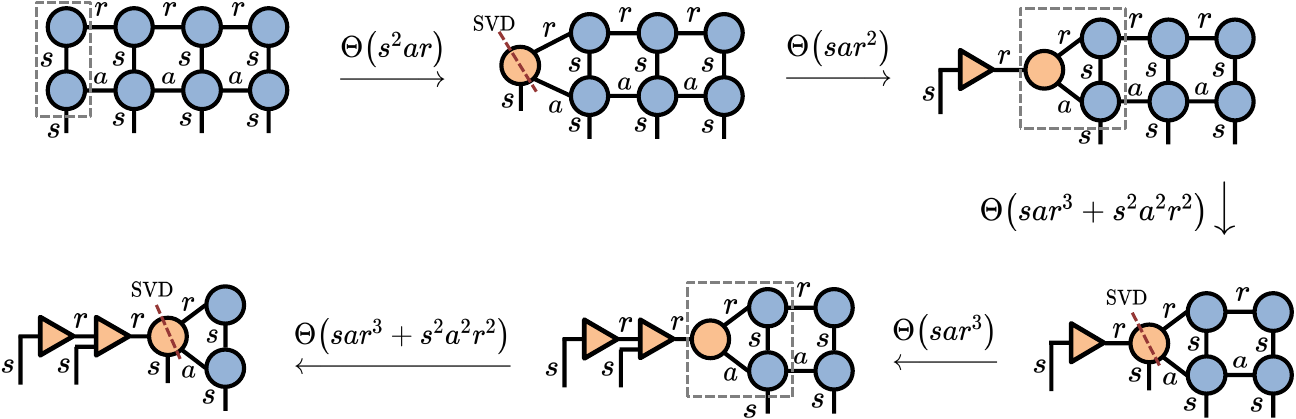}

\caption{
Illustration of the zip-up algorithm. Each dashed block includes the tensors to be contracted at a given step. Each tensor represented by a triangular vertex denotes a tensor with an orthogonality property.
}
\label{fig:zipup}
\end{figure}

The zip-up algorithm~\cite{stoudenmire2010minimally,paeckel2019time} is illustrated in \cref{fig:zipup}.
We also let the output truncated MPS have rank $r$.
The algorithm begins by contracting the leftmost pair of tensors. A truncated singular value decomposition (SVD) is then performed to obtain the left leading singular vectors $\mat{U}_1$ and the remaining non-orthogonal component $\mat{V}_1$. Next, $\mat{V}_1$ is combined with the second leftmost pair of tensors, and another truncated SVD is performed. This process continues until it reaches the right boundary of both the MPO and MPS.
When the resulting MPS has an order of $N$, the algorithm's asymptotic computational cost is $\Theta(N(s^2a^2r^2 + sar^3))$.
It should be noted, as depicted in \cref{fig:zipup}, that the truncation at the $i$th step employs an environment including all $i$ left MPO and MPS tensors, but not the full environment (all tensors in the MPS and MPO).

 \begin{figure}[htb]
\centering
\includegraphics[width=.99\textwidth, keepaspectratio]{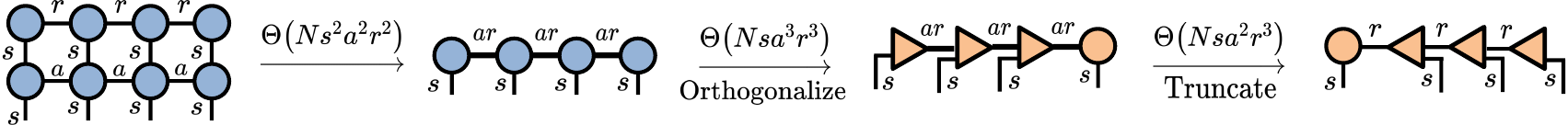}

\caption{
Illustration of the application and truncation algorithm.
}
\label{fig:apply_truncate}
\end{figure}

The canonicalization-based algorithm that uses the full environment is illustrated in \cref{fig:apply_truncate}.
The algorithm first multiplies the MPS and MPO,
resulting in an MPS with a rank of $ar$.
Subsequently, the MPS is truncated via the canonicalization-based algorithm reviewed in \cref{subsec:cano_tree}.
When the output MPS has an order $N$, the algorithm has an asymptotic cost of $\Theta(N(s^2a^2r^2+sa^3r^3))$, which is $O(a^2)$ times the cost of the zip-up algorithm.
However, this algorithm offers better accuracy since each truncation utilizes the full environment. Furthermore, the algorithm maintains a theoretical upper
bound on the truncation error~\cite{oseledets2011tensor}.

 \begin{figure}[htb]
\centering

\subfloat[The first step]{
\includegraphics[width=.85\textwidth, keepaspectratio]
{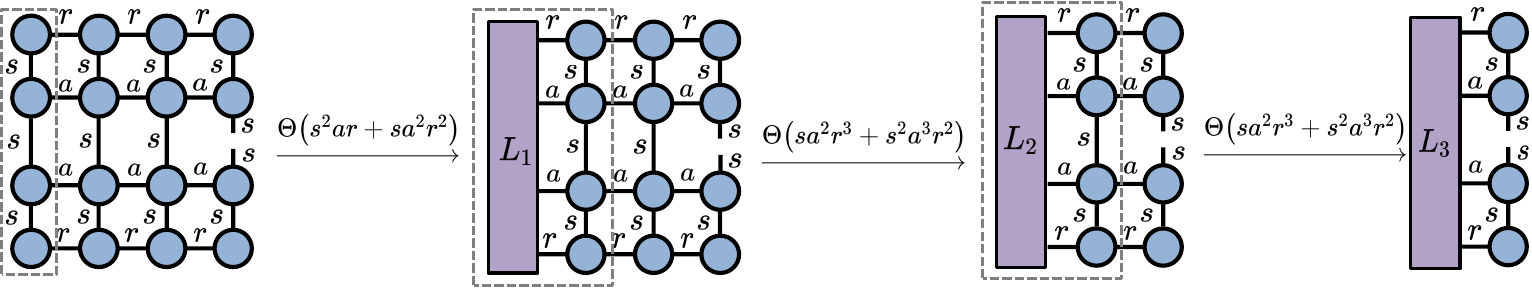}
\label{subfig:dm_step1}}

\subfloat[The second step]{
\includegraphics[width=.85\textwidth, keepaspectratio]
{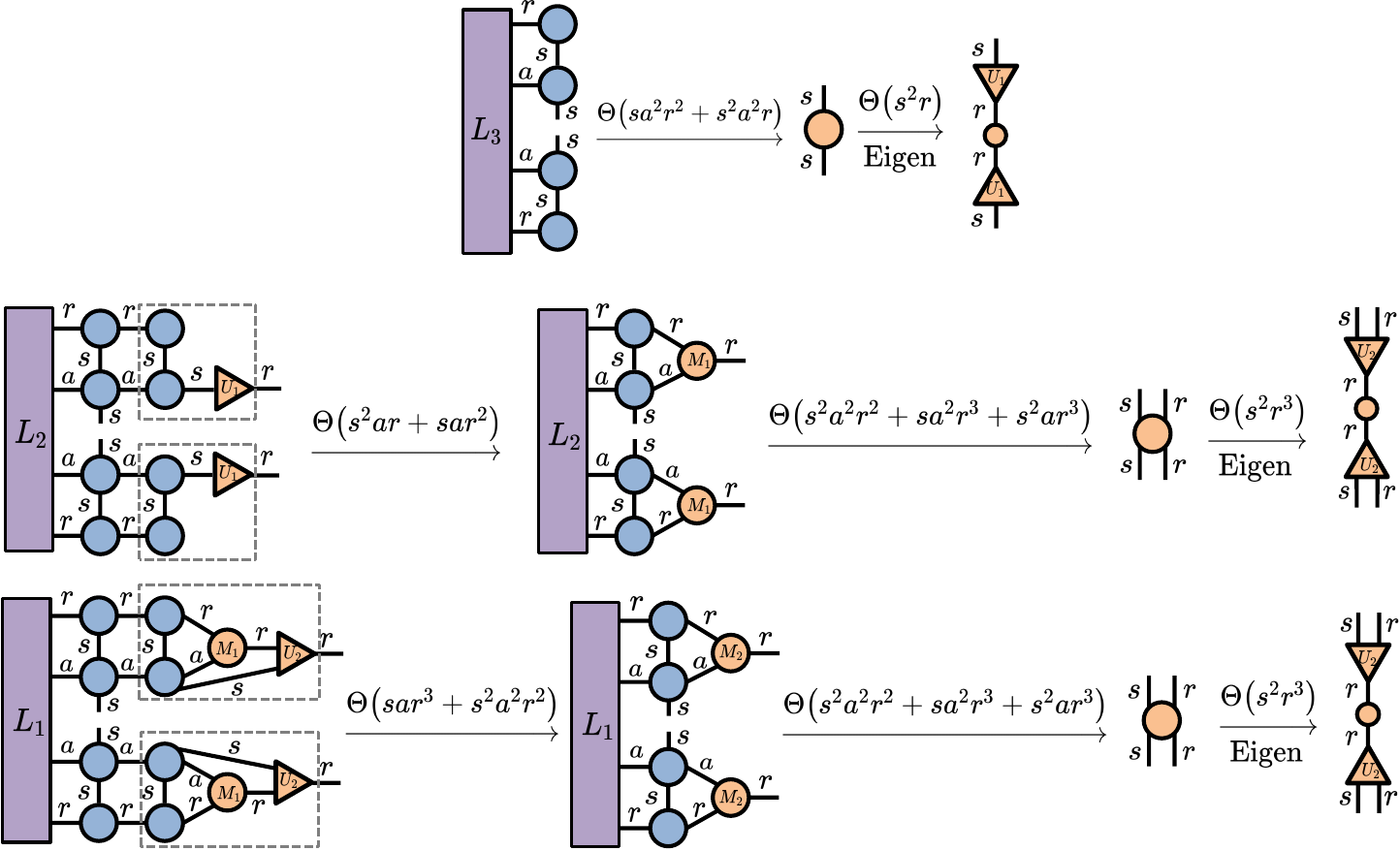}
\label{subfig:dm_step2}}

\subfloat[The third step]{
\includegraphics[width=.85\textwidth, keepaspectratio]
{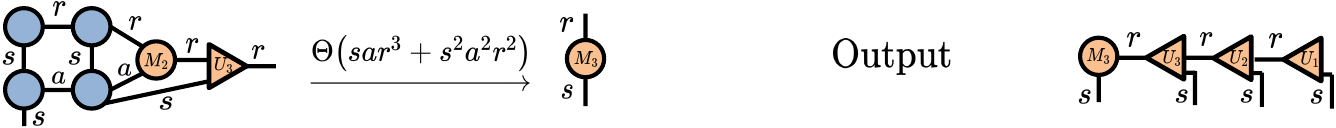}
\label{subfig:dm_step3}}

\caption{
Illustration of the density matrix algorithm. Each triangular vertex represents a tensor with an orthogonal property.
}
\label{fig:density_matrix}
\end{figure}

\paragraph{The density matrix algorithm}

The density matrix algorithm produces an equivalent truncated MPS as the application and truncation algorithm, and we illustrate the algorithm in \cref{fig:density_matrix}. The algorithm contains three steps,
\begin{enumerate}
    \item Computing matrices $\mat{L}_i$, as is shown in \cref{subfig:dm_step1}. These matrices are
computed by sequentially contracting the network from left to right, and intermediates $\mat{L}_i$ are saved during the contractions.
    \item Performing a sweep of contractions from right to left and use $\mat{L}_i$ to compute all the leading singular vectors $\mat{U}_i$ for $i\in \{1,\ldots, N-1\}$.
    Specifically, $\mat{L}_N$ is firstly used to compute the density matrix with the last pair of uncontracted modes left open, and truncated
eigendecomposition is performed on the density matrix to yield the leading singular vectors $\mat{U}_1$.
Next, the intermediates $\mat{L}_{N-1}$ is utilized to compute the density matrix with the right two uncontracted modes left open. Additionally, the basis of this density matrix is transformed by applying $\mat{U}_1$, as shown in \cref{subfig:dm_step2}. This process is repeated until $N-1$ tensors $\mat{U}_i$ are obtained.
    \item Getting the leftmost matrix $\mat{M}_N$ that encompasses all the non-orthogonal information through the contraction depicted in  \cref{subfig:dm_step3}, and form the output MPS by combining all $\mat{U}_i$ and $\mat{M}_N$.
\end{enumerate}

When the output MPS has an order $N$, the density matrix algorithm has an asymptotic cost of $\Theta(N(sa^2r^3+s^2a^3r^2+s^2ar^3))$. In applications arising in statistical physics and quantum computing, the size $s$ is commonly the smallest. As is shown in \cref{tab:compare}, for the case where $s\ll a=\Theta(r)$, the cost of density matrix algorithm is $\Theta(s^2r^5)$, which is $\Theta(r/s)$ better than the canonicalization with full environment algorithm.
For the other case, where $s=\Theta(a)\ll r$, the cost of the density matrix algorithm is $\Theta(s^3r^3)$, which is $\Theta(s)$ better than the canonicalization with full environment algorithm.

\paragraph{Automation and generalization of the MPO-MPS multiplication algorithms}

There is an opportunity to generalize the MPO-MPS multiplication algorithms to arbitrary graphs. In particular, SweepContractor~\cite{chubb2021general}
generalizes the MPO-MPS zip-up algorithm, and uses a subroutine that contracts a single tensor with an MPS into a
new MPS to contract arbitrary tensor networks.
In contrast, our proposed algorithm includes a subroutine that contracts a general tensor network (such as an MPO) with a binary tree tensor network into a binary tree network, allowing the generalizing of all three MPO-MPS multiplication algorithms.

The analysis and observations above suggest that the density matrix algorithm has
greater efficiency compared to the canonicalization-based algorithm. As a result, we generalize the density matrix algorithm for the MPO-MPS multiplication and implement one that is able to approximate a general tensor network into a  tree tensor network.
Generalization of the
density matrix algorithm to trees presents two challenges. Firstly, determining how to efficiently perform memoization (accelerating the contraction by caching partial contraction results) to reduce costs becomes less straightforward. In order to address this issue, we have introduced a strategy that utilizes graph partitioning in \cref{sec:dm}. Secondly, selecting an appropriate output tree structure that enhances the efficiency of the approximation poses a challenge. For the MPO-MPS multiplication, it is evident that the MPS ordering consistent with the input MPS and MPO would yield favorable results.  In \cref{sec:binary_tree}, we propose algorithms to select efficient tree structures for general graphs.

\section{The proposed tensor network contraction algorithm}
\label{sec:method}

In this section, we present the proposed approximate tensor network contraction algorithm.

\subsection{Definitions}\label{subsec:Def}

We use $G[S]=(S,E_S)$ to denote a sub tensor network defined on $S\subseteq V$, where $E_S$ contains all edges in $E$ adjacent to any $v\in S$.
For two disjoint subsets of $V$ denoted as $X,Y$, we let $E(X,Y)$ denote the set of edges connecting $X,Y$.
We let $E(X)$ denote the set of uncontracted edges of $G[X]$.

For the tensor network represented by $G = (V, E)$, we use $\mathcal{V} = \{V_1,\ldots,V_N\}$ to denote a  graph partitioning that partitions $V$ into $V_1,\ldots,V_N$.
    A \textit{contraction tree} of the partitioned network is a directed binary tree showing how vertex subsets in $\mathcal{V}$ are contracted, and it is denoted $T^{(\mathcal{V})}$. Each leaf of $T^{(\mathcal{V})}$ is a vertex subset in $\mathcal{V}$, and each non-leaf vertex in $T^{(\mathcal{V})}$ can be represented by a subset of the vertices, $W_1\cup W_2$, where its two children are represented by $W_1$ and $W_2$,  respectively. 

    For a given vertex in the contraction tree $T^{(\Set{V})}$ that is represented by $V'\subset V$, $\pth(T^{(\Set{V})}, V')$ denotes a sub-contraction path of $T^{(\Set{V})}$.
This sub-contraction path is a subgraph of $T^{(\Set{V})}$ that contains all vertices in $T^{(\mathcal{V})}$ that are ancestors of $V'$  as well as the children of these ancestors.
To illustrate, we provide an example of the sub-contraction path of $V_4$ in \cref{fig:path}.

\begin{figure}[htb]
\centering
\includegraphics[width=.55\textwidth, keepaspectratio]{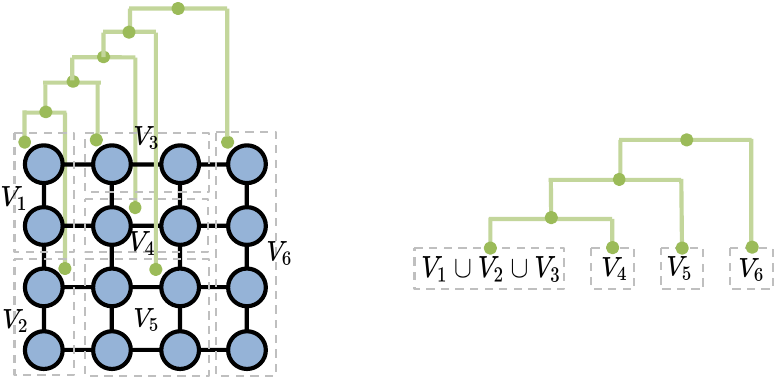}

\caption{
Illustration of the sub-contraction path. The left diagram denotes the graph partitioning and the contraction tree $T^{(\Set{V})}$, and the right diagram denotes the sub-contraction path
$\pth(T^{(\Set{V})}, V_4)$.
}
\label{fig:path}
\end{figure}

\subsection{An overview of the algorithm}\label{subsec:ovw}

In this section we present an overview of the algorithm. The algorithm takes as inputs a tensor network, a partitioning of that tensor network, a partial contraction tree of the partitioned tensor network, an ansatz for the structure of intermediate network contractions (for example, an MPS or a comb tree), and parameters for performing intermediate approximate tensor network contractions which tune the level of accuracy of the method. The algorithm proceeds by traversing the partial contraction tree and approximately contracting pairs of tensor network partitions specified by the contraction tree. The pair of tensor network partitions are approximately contracted using a specified algorithm such as the density matrix algorithm, resulting in a tensor network with the specified structure, such as an MPS or comb tree structure. The algorithm proceeds until all partitions are contracted. An example of the algorithm is shown in \cref{fig:ovw}.

\begin{algorithm}[!ht]
\setstretch{1.1}
\caption{\texttt{partitioned\_contract}: approximate tensor network contraction based on its given partition}
\begin{algorithmic}[1]
\STATE{\textbf{Input:} The tensor network $\tsr{T}$ with graph $G= (V, E)$, its partition $\Set{V}=\{V_1,\ldots, V_N\}$ and its contraction path $T^{(\Set{V})}$,
 ansatz $A$, maximum bond size $\chi$, and swap batch size $r$ \comment{The ansatz $A$ can be either ``MPS" or ``Comb"}
}
\STATE{$\tn\leftarrow$ a mapping that maps each vertex set to its approximated tensor network}
\STATE{$\Set{E}\leftarrow\left\{E(V_i,V_j): i,j\in \{1,\ldots,N\}\right\}$ \comment{The set where each element is an edge subset connecting two different partitions}
\label{line:edge_partitions}}
\STATE{\comment{Lines \ref{line:forloop_ansatz}-\ref{line:endforloop_ansatz}: construct edge/edgeset linear orderings that define the embedding tree}}
\STATE{$\left\{\sigma^{(E')}: E'\in \Set{E}\right\}\leftarrow $  selecting an ordering for each edgeset in $\Set{E}$ via recursive bisection\label{line:forloop_ansatz}}
\FOR{each contraction $(U_s,W_s)\in T^{(\Set{V})}$}
\STATE{$\Set{E}_s\leftarrow$ The subset of $\mathcal{E}$ that is adjacent to the sub tensor network with vertices $U_s\cup W_s$}
\STATE{
$ \sigma^{(\Set{E}_s)}
\leftarrow \texttt{embedding\_tree\_ordering}\left(G[U_s\cup W_s], \pth\left(T^{\Set{V}}, U_s\cup W_s\right), \Set{E}_s\right)$
  \comment{Select an ordering for edgesets in $\Set{E}_s$}}
\ENDFOR\label{line:endforloop_ansatz}
\FOR{each contraction $(U_s,W_s)\in T^{(\Set{V})}$\label{line:forloop_approx}}
\STATE{$\tn(U_s\cup W_s) \leftarrow \texttt{approx\_tensor\_network}\left(\tn(U_s)\cup \tn(V_s), \sigma^{(\Set{E}_s)}, \{\sigma^{(E')}: E'\in \Set{E}_s\}, \chi, r, A\right)$
\comment{Approximate the input tensor network $\tn(U_s)\cup \tn(V_s)$ as a binary tree tensor network $\tsr{X}_s$, \cref{alg:approx_tensor_network}}
}
\ENDFOR\label{line:endforloop_approx}
\RETURN the final approximated tensor network, $\tn(V)$
\end{algorithmic}
\label{alg:ovw}
\end{algorithm}

Pseudocode providing more details of steps of the algorithm is presented in \cref{alg:ovw}. Definitions of all notations are summarized in \cref{tab:notations}.
The algorithm takes as input the tensor network partition $\mathcal{V}$ and its contraction path $T^{(\mathcal{V})}$. During each contraction step along the path, all tensors within the input partitions are treated as the environment. Consequently, larger partitions typically lead to higher approximation accuracy, but at the cost of increased computational complexity.

When two tensor network partitions are contracted, an embedding tree is first constructed which specifies the structure of the network that will result from the contraction. The embedding tree is a full binary tree where each leaf vertex is associated with a dangling edge/mode of the subnetwork made from composing the two partitions that are being contracted. Furthermore, each non-leaf vertex in the embedding tree corresponds to a tensor within the resulting binary tree tensor network. All tensors within this network have an order of three, except for the tensor located at the root vertex.
An example of such an embedding tree is illustrated in the second left diagram of \cref{fig:pipeline}.

The selection of the embedding tree is guided by an analysis of the structure of the input tensor network graph $G$, its partitioning, and the contraction path. This analysis aims to identify a tree structure that optimizes the efficiency of both the current contraction and any subsequent contractions involving the contracted output. The determination of each embedding tree structure occurs in lines \ref{line:forloop_ansatz}-\ref{line:endforloop_ansatz}.
Note that the generation of the embedding tree only depends on the tensor network graph structure, rather than the actual tensor data.
The relationship between the embedding tree and the orderings of the edges is further explained in \cref{subsec:embed_tree}.

While our method introduces a novel perspective on the construction of the embedding tree, it is worth noting that alternative approaches have been proposed for determining tree structures based on various other heuristics. For instance, Seitz et al. \cite{seitz2023simulating} propose a method specifically for determining a tree structure based on a quantum circuit. Other studies, such as those by Nakatani and Chan \cite{nakatani2013efficient}, Murg et al. \cite{murg2015tree}, Szalay et al. \cite{szalay2015tensor}, and Ferrari et al. \cite{ferrari2022adaptive}, explore different heuristics for tree construction, leveraging factors like entanglement or interaction strength.

After selecting an embedding tree, we proceed to embed the tensor network comprising two partitions into the embedding tree and truncate it to ensure that the maximum bond size remains below $\chi$. This process is performed in lines \ref{line:forloop_approx}-\ref{line:endforloop_approx}.
In-depth explanations of the hybrid algorithm, which combines the density matrix algorithm and the swap-based algorithm to obtain the approximated binary tree tensor network, can be found in  \cref{sec:dm} and \cref{sec:swap}. This hybrid algorithm involves multiple iterations of the density matrix algorithm, each progressively modifying the structure of the tensor network to a degree controlled by the swap batch size $r$. The choice of $r$ allows the user to find a balance between accuracy and computational cost for specific problem instances.

\subsection{Determination of the embedding tree}\label{subsec:embed_tree}

We explain the embedding tree structure used in \cref{alg:ovw}. As is defined in \cref{subsec:ovw}, an embedding tree is a rooted full binary tree, with each leaf vertex representing an uncontracted edge in the tensor network.

Let
$\mathcal{E}=\left\{E(V_i,V_j): i,j\in \{1,\ldots,N\}\right\}$, so that each element in $\mathcal{E}$ is an edge subset connecting two different partitions. For a specific contraction $(U_s,W_s)$ defined in \cref{tab:notations}, we let $\mathcal{E}_s$ be the subset of $\mathcal{E}$ that is adjacent to the tensor network represented by $U_s\cup W_s$.
We design the embedding tree structure so  that the leaves that represent each $E_i\in \mathcal{E}_s$ are in close proximity to one another. This arrangement is advantageous because all edges within each $E_i$ are always contracted together in the same contraction. Placing them close to each other simplifies the contraction process and eliminates the need for unnecessary permutation of modes.

Two structures we use for the embedding tree are the MPS (maximally-unbalanced full binary tree) and the comb~\cite{bauernfeind2017fork,chepiga2019comb}.
The comb tensor network is a tree tensor network arranged in a linear chain with branches.
Both structures are based on a linear orderings $\sigma^{(\mathcal{E}_s)}$ for $\mathcal{E}_s$
and linear orderings $\sigma^{(E')}$ for $E'\in \mathcal{E}_s$, and they are generated in lines~\ref{line:forloop_ansatz}-\ref{line:endforloop_ansatz} of \cref{alg:ovw}.
We formally define the embedding tree with an MPS and a comb structure in \cref{sec:appendix_tree_def}.
We visualize both the embedding tree with an MPS structure and with a comb structure in \cref{fig:ansatz}.

 \begin{figure}[!ht]
\centering

\subfloat[MPS tree]{
\hspace{3mm}
\includegraphics[width=.17\textwidth, keepaspectratio]
{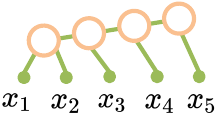}\hspace{1mm}\label{contract_subfig:mps_tree}}

\subfloat[Embedding tree with an MPS structure]{
\hspace{6mm}
\includegraphics[width=.36\textwidth, keepaspectratio]
{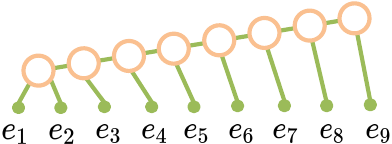}\hspace{6mm}\label{contract_subfig:mps}}
\subfloat[Embedding tree with a comb structure]{\hspace{8mm}\includegraphics[width=.33\textwidth, keepaspectratio]{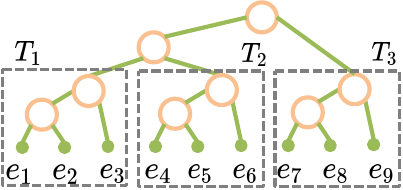}
\hspace{6mm}
\label{subfig:comb}}

\caption{(a) Visualization of the MPS tree defined on $\sigma^S$ with $\sigma^S(x_i) = i$. (b)(c) Visualization of the embedding tree with an MPS structure and a comb structure.
The input orderings are $\sigma^{(\mathcal{E}_i)} = (E_1, E_2, E_3)$ with
$E_1 = \{e_1, e_2, e_3\}$,
$E_2 = \{e_4, e_5, e_6\}$, and
$E_3 = \{e_7, e_8, e_9\}$. $\sigma^{(E_1)},\sigma^{(E_2)},\sigma^{(E_3)}$ are defined so that in $\hat{\sigma} = \sigma^{(E_1)}\oplus\sigma^{(E_2)}\oplus\sigma^{(E_3)}$, $\hat{\sigma}(e_i) = i$.
}
\label{fig:ansatz}
\end{figure}

In comparison to the MPS structure, the comb structure has a smaller diameter, representing the maximum distance between any two vertices. However, the comb structure also has a larger maximum tensor size of $\Theta(\chi^3)$, where $\chi$ is the maximum bond size. This is larger than the maximum tensor size of the MPS structure, which is $\Theta(s\chi^2)$, where $s$ represents the uncontracted mode size and is typically much smaller than $\chi$. In \cref{sec:exp}, we conduct experimental comparisons between the performance of the MPS structure and the comb structure.

Various heuristics can be used to obtain the linear ordering $\sigma^{(E')}$ for each $E' \in \mathcal{E}$. In this work, we utilize the recursive bisection algorithm described in \cref{subsec:recur_bisection}, on a partition of the input graph $G$ that is connected to $E'$. The recursive bisection algorithm is a heuristic that aims to minimize congestion in the linear ordering. By applying this algorithm, we obtain an ordering that results in the embedding tree tensor network having low ranks.
The algorithm for selecting the ordering $\sigma^{(\mathcal{E}_s)}$ is explained in detail in \cref{sec:binary_tree}.

\section{The algorithm to select the edge subset ordering of the embedding tree}\label{sec:binary_tree}

For a given contraction $(U_s,W_s)$,
we detail the algorithm to select the linear ordering $\sigma^{(\mathcal{E}_s)}$ for the intermediate tensor network $G_s = (V_s, E_s)$, where $V_s=U_s\cup W_s$. $\sigma^{(\mathcal{E}_s)}$ is generated based on both $G_s$ and the sub-contraction path $T=\pth\left(T^{(\mathcal{V})}, V_s\right)$ defined at \cref{subsec:Def}, where $T^{(\mathcal{V})}$ is the contraction tree over the partition $\mathcal{V}$.

The  ordering $\sigma^{(\mathcal{E}_s)}$ is chosen with two objectives.
Firstly, it is designed to satisfy a specific adjacency relation that greatly facilitates efficient subsequent contractions. This adjacency relation ensures that for each of the subsequent contractions $(U_k,W_k)$, the contracted edges between $U_k$ and $W_k$ are adjacent in both input tensor networks $\texttt{tn}(U_k)$ and $\texttt{tn}(V_k)$. The adjacency of these contracted edges results in a lower cost for the contraction, compared to the scenario where the contracted edges are not adjacent.
This adjacency relation is described by the \textit{constraint tree} for $\mathcal{E}_s$, $T^{(\mathcal{E}_s)}$. Each leaf vertex in the constraint tree represents an edge set in $\mathcal{E}_s$, and each non-leaf vertex has at least 2 children and indicates the edge subsets represented by the
children are adjacent. 
Each non-leaf vertex also denotes whether the children's vertices are ordered or not.
We show an example of the constraint tree in the bottom right diagram of \cref{fig:adjacency_tree_example}.
In \cref{subsec:adjacency}, a detailed explanation is provided on how to select the constraint tree.

Secondly, the resulting binary tree structure should be similar to the tensor network $G_s$ in order to keep the ranks of the resulting tree tensor network low.
In \cref{subsec:cost-efficient_order}, we detail the algorithm to find the ordering not only consistent with the constraint tree, but also to minimize the cost of permutation (Kendall-Tau distance between the chosen ordering and another reference ordering whose corresponding line structure is similar to $G_s$).

\subsection{Determination of the edge set ordering based on the constraint tree}\label{subsec:cost-efficient_order}

We provide an explanation of the algorithm that determines the ordering for the set of elements $\mathcal{E}_s$, denoted as $\sigma^{(\mathcal{E}_s)}$. This ordering is not only constrained by the constraint tree  $T^{(\mathcal{E}_s)}$ but also aims to reflect the structure of the input graph $G_s$. The algorithm is presented in \cref{alg:min_kdt}.
To begin with, in Line~\ref{line:reference_order}, we generate a reference ordering denoted as $\tau$ for the set of elements $\mathcal{E}_s$. This reference ordering is generated using recursive bisection  and represents a linear structure that  is close to the structure of $G_s$.
Subsequently, the algorithm proceeds to construct the output ordering by employing a post-order DFS traversal of the constraint tree $T^{(\mathcal{E}_s)}$. This traversal strategy ensures that the ordering takes into account the constraints imposed by the tree structure. In \cref{sec:order_optimality}, we prove that the output ordering of \cref{alg:kdt} minimizes the Kendall-Tau distance with the reference ordering under the adjacency constraint.

\begin{algorithm}[!ht]
\caption{\texttt{linear\_ordering\_under\_constraint\_tree}: Algorithm to get the edge set ordering that minimizes the Kendall-Tau distance with the reference ordering under the adjacency constraint}
\begin{algorithmic}[1]\label{alg:kdt}
\STATE{\textbf{Input:} the edge set $\Set{E}_s$, the constraint tree $T^{(\Set{E}_s)}$, the tensor network graph $G_s=(V_s,E_s)$} 
\STATE{$\tau \leftarrow \texttt{linear\_ordering}\left(\Set{E}_s, G_s\right)$ \comment{Ordering generated via recursive bisection}\label{line:reference_order}
}
\STATE{$f\leftarrow$ a mapping that maps each vertex in $T^{(\Set{E}_s)}$ to its edge set ordering}
\FOR{each leaf vertex $v$ that represents $E_i$ in $T^{(\Set{E}_s)}$}
\STATE{$f(v)\leftarrow$ the ordering that contains the single edge set $E_i$}
\ENDFOR
\FOR{each non-leaf vertex $v$ that represents $\hat{\Set{E}}_i$ based on a post-order DFS traversal of  $T^{(\Set{E}_s)}$}
\STATE{$u_1,\ldots,u_{n_v}\leftarrow $ children of $v$}
\IF{$v$ is labeled as \textit{ordered}}
\STATE{$\sigma_1\leftarrow f(u_1)\oplus f(u_2)\oplus \cdots \oplus f(u_{n_v})$ \comment{Concatenate all $f(u_i)$ in order}}
\STATE{$\Set{S}\leftarrow \left\{\sigma_1,\texttt{reverse}(\sigma_1)\right\}$}
\ELSE
\STATE{$\Set{S}\leftarrow$ a set of all permutations of $\left\{f(u_1), f(u_2), \ldots, f(u_{n_v})\right\}$}
\ENDIF
\STATE{$\tau_v\leftarrow$ a partial ordering of $\tau$ over the subset $\hat{\Set{E}}_i$\label{line:tauv}}
\STATE{$f(v)\leftarrow \arg \min_{{\sigma}\in \Set{S}}\ktdist{{\sigma}, \tau_v}$ \label{line:min_kdt} \comment{Minimize the Kendall-Tau distance defiend at \cref{defn:ktdist}}}
\ENDFOR
\RETURN $f\left(\texttt{root}\left(T^{(\Set{E}_s)}\right)\right)$
\end{algorithmic}
\label{alg:min_kdt}
\end{algorithm}

\section{The density matrix algorithm for tree approximations}\label{sec:dm}

We present a density matrix algorithm to approximate an arbitrary tensor network into a tree tensor network.
The standard approach involves embedding the input tensor network into an embedding tree and explicitly forming the untruncated tree tensor network,
then truncating the resulting tree tensor network using the canonicalization-based algorithm. However, this can lead to tree tensor networks with large ranks, resulting in expensive canonicalization and low-rank approximation processes.

Our proposed density matrix algorithm builds upon the density matrix algorithm originally designed for MPO-MPS multiplication, which is discussed in~\cref{subsec:2d}. Given a tree embedding of the input tensor network, our algorithm eliminates the need to explicitly construct the untruncated tree tensor network. It offers the advantage of forming a low-rank tree tensor network without requiring the generation of large intermediate tensors.
Specifically, we show in \cref{subsec:cost_analysis} that the asymptotic computational cost of the algorithm is upper-bounded by the cost of the canonicalization-based algorithm, and we show in \cref{sec:exp} that for many input tensor networks, the proposed algorithm substantially reduces the overall execution time.

\subsection{Definitions}

Within the algorithm, we use $\texttt{density\_matrix}_{T}\left(v\right)$  and $\texttt{density\_matrix}_{T}\left(v, z\right)$ introduced in \cref{def:dm}.
For a given embedding tree $T=(V_T,E_T)$ with each vertex in $T$ representing a partition of the tensor network embedded to that vertex,
we use the notation
 $\texttt{density\_matrix}_{T}\left(v\right)$
to  calculate the density matrix of vertex $v$ on top of the embedding tree $T$, with the open edges of the matrix being the uncontracted edges incident to $v$.
Moreover, $\texttt{density\_matrix}_{T}\left(v,z\right)$
 calculates the density matrix of vertex $v$ with the open edge of the matrix being $E_T(v,z)$. We show an illustration in \cref{fig:dm_cache}.

 \begin{figure}[!ht]
\centering
\subfloat[$\texttt{density\_matrix}_{T}\left(v, z\right)$]{
\includegraphics[width=.45\textwidth, keepaspectratio]
{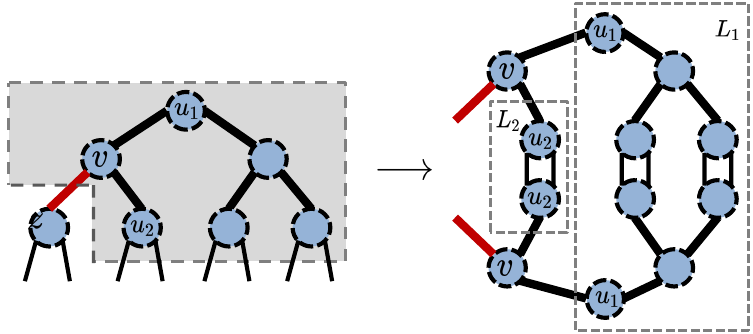}\hspace{5mm}\label{subfig:dm_1}}
\subfloat[$\texttt{density\_matrix}_{T}\left(z\right)$]{\hspace{5mm}\includegraphics[width=.45\textwidth, keepaspectratio]{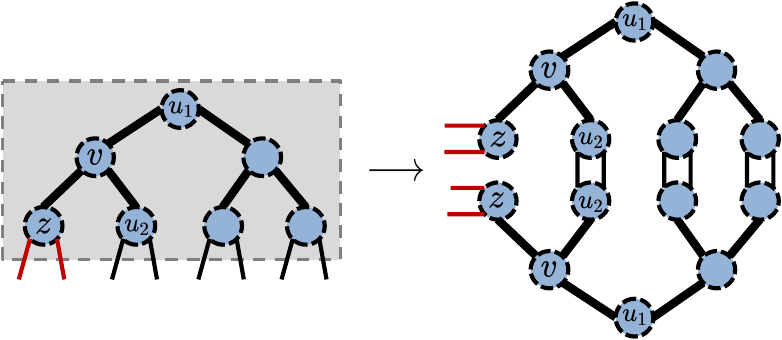}
\label{subfig:dm_2}}

\caption{
Visualization of $\texttt{density\_matrix}_{T}\left(v, z\right)$ and $\texttt{density\_matrix}_{T}\left(z\right)$.
In the left diagrams of (a)(b), the tree structure is the embedding tree $T$, and each vertex represents a partition of the network embedded in that vertex. The open edge of the density matrix is marked in red. The dashed boxes denote the tensor networks squared in the density matrices ($T(S)$ in \cref{def:dm}). The right diagrams visualizes the density matrices. In (a), $L_1 = \texttt{density\_matrix}_{T}\left(u_1, v\right)$ and $L_2 = \texttt{density\_matrix}_{T}\left(u_2, v\right)$ can be cached and reused when computing the density matrix.
}
\label{fig:dm_cache}
\end{figure}

\begin{definition}[Density matrix]\label{def:dm}
    Consider a given embedding tree $T=(V_T,E_T)$ with each $T(v)$ for $v\in V_T$ representing a sub tensor network, and let $T(S) = \cup_{v\in S}T(v)$.
    For a given vertex $v\in V_T$, and a set of edges $\tilde{E}_v\subset E_T$ that is adjacent to $v$, let $S\subseteq V_T$ denote the vertices connected to $v$ when $\tilde{E}_v$ is removed from $T$.
    Let $\mat{T}_{(\tilde{E}_v)}$ denote the matricization of the tensor network $T(S)$ with all modes defined by $\tilde{E}_v$ are combined into the matrix row.
    Then the density matrix defined on $T, v, \tilde{E}_v$,
    denoted as $\texttt{density\_matrix}_{T}\left(v, \tilde{E}_v\right)$,
    equals $\mat{T}_{(\tilde{E}_v)}\mat{T}^{T}_{(\tilde{E}_v)}$.
    For simplicity, we let $\texttt{density\_matrix}_{T}\left(v\right)$ denote the density matrix of $v$ when $\tilde{E}_v=E_T(v,*)$ is the uncontracted edge set incident on $v$, and we let $\texttt{density\_matrix}_{T}\left(v, u\right)$ denote the density matrix of $v$ when $\tilde{E}_v = E_T(u, v)$.
\end{definition}

For a tensor network $G=(V,E)$,
we also define
$\cut_G(X,Y) = \sum_{e\in E(X,Y)}w(e)$, where $E(X,Y)$ denotes the set of edges connecting two disjoint vertex subsets $X,Y$.
For two vertices $u,v\in V$, we define the minimum cut between $u,v$ in $G$ as
\[\mincut_G(u,v) = \min_{\substack{A,B\subset V \\ u\in A, v\in B}}\cut_G(A,B).\]
Let $E_1,E_2$ be the two different subsets of the uncontracted edges of $G$, we define $\mincut_G(E_1,E_2)$ as the mincut between two new vertices $a,b$ on the graph that contains both $G$ and $a,b$, where $a,b$ are adjacent to $E_1,E_2$, respectively.

\subsection{The density matrix algorithm}

The density matrix algorithm is summarized in \cref{alg:dm}.
The algorithm  involves computing the output network by performing a post-order DFS traversal of the embedding tree.
During the traversal, at each vertex $v$, the corresponding tensor $\mat{U}_v$ is computed. Subsequently, vertex $v$ is removed from the embedding tree. This process continues iteratively until only the root vertex remains, whose tensor encapsulates all the non-orthogonal information of the network.
A visualization of the algorithm is shown in \cref{fig:dm}.

In \cref{alg:dm}, we initially construct an embedding $\phi$ utilizing the recursive bisection technique outlined in \cref{alg:embed_tree}.
This embedding assigns a tensor network partition to each vertex in the embedding tree and serves as a guide for the memoization strategy.
As is reviewed in \cref{subsec:recur_bisection}, recursive bisection is a standard heuristic to find embeddings with low congestion.
It is worth noting that \cref{alg:embed_tree} may produce an embedding in which there exists a vertex in the embedding tree whose corresponding tensor network partition is empty. In such cases, we can address this problem by introducing identity matrices into the input graph. This adjustment ensures that the resulting tensor network remains equivalent while guaranteeing the non-emptiness of each partition.

For computing $\mat{U}_v$ at each vertex $v \in V_T$,
\cref{alg:dm} incorporates two subroutines that handle two distinct cases efficiently.
In the algorithm, we let
$\mat{M}_v$ denote the matricized contraction output of the partition at $v$, $T(v)$, that  combines all uncontracted modes  into the matrix row. In addition, let
$\mat{L}_v= \texttt{density\_matrix}_{T'}( v)$ and $\mat{L}_u= \texttt{density\_matrix}_{T'}(u, v)$.

Since $\mat{L}_v = \mat{M}_v\mat{L}_u\mat{M}_v^T$,
if the number of rows in $\mat{L}_v$ is smaller than the number of rows in $\mat{L}_u$,  in Lines~\ref{line:begin_dm}-\ref{line:end_dm} we compute $\mat{L}_v$ then obtain its singular vectors, which is
the most efficient approach.
Conversely, if the number of rows in $\mat{L}_v$ exceeds the number of rows in $\mat{L}_u$, it implies that $\mat{L}_v$ is not full rank.
In such cases, we use an subroutine called QR-SVD~\cite{pang2020efficient} instead in Lines~\ref{line:begin_qr_svd}-\ref{line:end_qr_svd}.
we first use QR factorization to orthogonalize $\mat{M}_v$ and yield $\mat{Q}_v\mat{R}_v$, and subsequently calculate the leading singular vectors of $\mat{R}_v\mat{L}_u\mat{R}_v^T$, which yields an implicit representation of the singular vectors of $\mat{L}_v$.
QR-SVD avoids the generation of the large density matrix $\mat{L}_v$, thus having a better asymptotic cost.
In \cref{subsec:cost_analysis}, we demonstrate that \cref{alg:dm} provides a guarantee that its asymptotic computational cost remains upper-bounded by that of the canonicalization-based algorithm.

\begin{algorithm}[!ht]
\setstretch{1.1}
\caption{\texttt{density\_matrix\_alg}: The density matrix algorithm for tree approximation}
\begin{algorithmic}[1]
\STATE{\textbf{Input:} The tensor network $G= (V, E)$,  its embedding tree $T=(V_T,E_T)$,
and maximum bond size $\chi$
}
\STATE{$\phi\leftarrow \texttt{tree\_embedding}(G,T)$ \comment{Constructed based on \cref{alg:embed_tree}}}
\STATE{$r\leftarrow$ root vertex in $T$}
\STATE{$T'\leftarrow$ a tree with the same structure as $T$ and $T'(v)$ for $v\in V_T$ denotes all tensors embedded to $v$ in $\phi$}
\FOR{each $v\in V_T\setminus \{r\}$ based on a post-order DFS traversal of $T$}
\STATE{$A_v\leftarrow \texttt{uncontracted\_edges}(T',v)$}
\STATE{$B_v\leftarrow \texttt{contracted\_edges}(T',v)$}
\STATE{$u\leftarrow \texttt{parent}(T',v)$}
\IF{$w(A_v)=O(w(B_v))$}
    \STATE{$\mat{L}_v\leftarrow \texttt{density\_matrix}_{T'}( v)$ \comment{Defined in \cref{def:dm}}\label{line:begin_dm}}
    \STATE{$\mat{U}_v\leftarrow \texttt{leading\_eigenvectors}(\mat{L}_v, \chi)$\label{line:end_dm}}
\ELSE
    \STATE{\comment{Perform QR-SVD~\cite{pang2020efficient} to reduce the asymptotic cost}\label{line:begin_qr_svd}}
    \STATE{$\mat{L}_u\leftarrow \texttt{density\_matrix}_{T'}(u, v)$ \label{line:dm_2}}
    \STATE{$\mat{M}_v\leftarrow$ the matricized contraction output of $T(v)$ with $A_v$ combined into row}
    \STATE{$\mat{Q}_v,\mat{R}_v\leftarrow \texttt{QR}(\mat{M}_v)$}
    \STATE{$\hat{\mat{U}}_v\leftarrow \texttt{leading\_singular\_vectors}(\mat{R}_v\mat{L}_u\mat{R}_v^T, \chi)$ }
    \STATE{$\mat{U}_v\leftarrow \mat{Q}_v\hat{\mat{U}}_v$\label{line:end_qr_svd}}
\ENDIF
\STATE{Add both $T'(v)$ and a vertex that represents $\mat{U}^T_v$ to $T'(u)$, and remove $v$ from  $T'$}
\ENDFOR
\STATE{$\mat{M}_r\leftarrow $ contraction output of $T(r)$}
\RETURN the tree tensor network that contains all $\mat{U}_v$ and the root tensor $\mat{M}_r$
\end{algorithmic}
\label{alg:dm}
\end{algorithm}

 \begin{figure}[!ht]
\centering
\includegraphics[width=1.\textwidth, keepaspectratio]{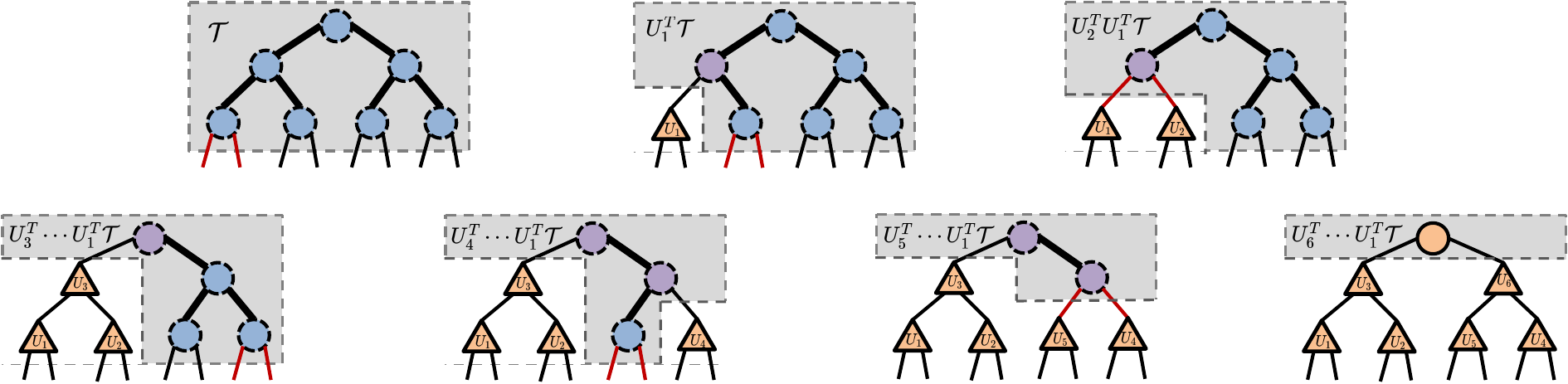}

\caption{
Visualization of the density matrix algorithm. The tree structure in each diagram is the embedding tree.
Each dashed circle represents a partition of the tensor network, and each solid circle/rectangle represents a tensor. Blue, purple, and orange vertices represent the input tensor network, intermediate tensors generated during the algorithm, and the output tensors, respectively.
The input tensor network is represented by the top left diagram, and the output one is represented by the bottom right diagram. In each diagram, the network included in the dashed box has a structure of $T'$ in \cref{alg:dm} and is used to compute the density matrix, and red edges denote the open edges of the density matrix.
}
\label{fig:dm}
\end{figure}

\begin{algorithm}[!ht]
\setstretch{1.1}
\caption{$\texttt{tree\_embedding}$: embedding a graph into the embedding tree via recursive bisection}
\begin{algorithmic}[1]
\STATE{\textbf{Input:} The source graph $G= (V, E)$, the embedding tree $T=(V_T,E_T)$
}
\IF{$|V_T| = 1$}
\RETURN an embedding that mapping all $v\in V$ to the vertex in $V_T$
\ENDIF
\STATE{$\phi\leftarrow$ an empty embedding function}
\STATE{$r\leftarrow$ root vertex in $T$}
\STATE{$E_L,E_R\leftarrow$ open edges represented by the left leaves and right leaves $r$, respectively}
\STATE{$S_L, S_R\leftarrow$ bipartition of $V$ such that $\cut_G(S_L,S_R)=\mincut_G(E_L,E_R)$}
\STATE{$\phi_L\leftarrow \texttt{tree\_embedding}\left(G[S_L],\texttt{left\_child\_tree}(T)\right)$}
\STATE{$E'_L\leftarrow E(S_L,S_R)$}
\STATE{$S'_L, S'_R\leftarrow$ bipartition of $S_R$ such that $\cut_G(S'_L,S'_R)=\mincut_G(E'_L,E_R)$}
\STATE{For each $v\in S_L'$, let $\phi(v)=r$}
\STATE{$\phi_R\leftarrow \texttt{tree\_embedding}\left(G[S_R'],\texttt{right\_child\_tree}(T)\right)$}
\RETURN the combination of $\phi, \phi_L,\phi_R$
\end{algorithmic}
\label{alg:embed_tree}
\end{algorithm}

\subsection{The density matrix algorithm with memoization}

As can be seen from \cref{fig:dm}, there are many shared tensor network parts across density matrices.
 We present a memoization strategy that generalizes the memoization strategy for the density matrix algorithm of the MPO-MPS multiplication to reduce the computational cost.
The strategy is used in Lines~\ref{line:begin_dm}, \ref{line:dm_2} of \cref{alg:dm}.

The memoization strategy uses  the following recursive relation for $\texttt{density\_matrix}_{T}\left(v\right)$  and $\texttt{density\_matrix}_{T}\left(v, z\right)$,
\begin{equation}\label{eq:dm_cache}
\begin{split}
    \texttt{density\_matrix}_{T}\left(v\right) &=
\mat{M}^{(v)}_{E_T(v,*)}
\left(\bigotimes_{
\substack{u\in N(v)}
}
\texttt{density\_matrix}_{T}(u, v)\right)\mat{M}^{(v)T}_{E_T(v,*)}, \\
    \texttt{density\_matrix}_{T}\left(v, z\right) &=
\mat{M}^{(v)}_{E_T(v,z)}
\left(\bigotimes_{
\substack{u\in N(v)\setminus \{z\}}
}
\texttt{density\_matrix}_{T}(u, v)\right)\mat{M}^{(v)T}_{E_T(v,z)},
\end{split}
\end{equation}
where $N(v)$ denotes the set of vertices adjacent to $v$,
 $\otimes$ denotes a Kronecker product, and $\mat{M}^{(v)}_{E_T(v,*)}$ denote a matricization of the tensor network represented by $v$, $T(v)$. In this matricization, all uncontracted modes incident on $v$ are combined into the row. $\mat{M}^{(v)}_{E_T(v,z)}$ denote a matricization of $T(v)$ where the mode represented by the edge $E_T(v,z)$ is the matrix row.

To compute the density matrix $\texttt{density\_matrix}_{T}\left(v,z\right)$, we first compute the density matrices for its neighboring vertices $u\in N(v)\setminus \{z\}$, then contract the target network that contains the density matrices as well as the tensor network $T(v)$ following  \eqref{eq:dm_cache}.
The contraction cost of the above target tensor network is dependent on the selected contraction path, and in practice one can either choose the optimal contraction path that minimizes the contraction cost or select it based on multiple heuristics~\cite{gray2021hyper}. Note that one way to contract the target network is to contract $T(v)$ into a tensor first and then contract it with the density matrices, but it may not yield the optimal cost.
If the terms $\texttt{density\_matrix}_{T}(u, v)$ have already been computed when generating other density matrices, we will cache and reuse them here. We illustrate such strategy in \cref{subfig:dm_1}.
The same strategy is  used to compute $\texttt{density\_matrix}_{T}\left(v\right)$.

In \cref{alg:dm}, the computation of each density matrix occurs only once. Considering that the embedding tree $T$ is limited to being a rooted binary tree, there are at most three density matrices to be calculated for each vertex $v$ in the embedding tree. Below we bound the asymptotic computational cost of the density matrix algorithm using memoization,  and we show that for a given embedding $\phi$, the cost will be upper-bounded by the algorithm that uses canonicalization, justifying the efficiency of the algorithm.

\subsection{Computational cost analysis}\label{subsec:cost_analysis}

We compare the asymptotic computational costs of the density matrix algorithm and the baseline algorithm that utilizes canonicalization, as discussed in \cref{subsec:cano_tree}.  In \cref{thm:upper_bound}, we establish that the density matrix algorithm can be more efficient in approximating a general tensor network as an embedding tree. The cost of the density matrix algorithm is upper-bounded by that of the canonicalization-based algorithm. This efficiency arises from the fact that the density matrix algorithm does not need to explicitly contract the partition embedded in each tree vertex into a tensor.

\begin{theorem}\label{thm:upper_bound}
Consider a given tensor network $G=(V,E)$, an embedding tree $T=(V_T,E_T)$, and an embedding $\phi$ that embeds $G$ into $T$.
Let $\sigma : V_T\to \{1,\ldots, |V_T|\}$ be a post-order DFS traversal of $T$ that shows the the tensor update ordering.
Assuming that changing a tree tensor network into its canonical form will not change any bond size of the network, the asymptotic cost of the density matrix algorithm (\cref{alg:dm}) is upper-bounded by that of the canonicalization-based algorithm (\cref{alg:canonize})
if both algorithms use the same embedding $\phi$, the same update ordering $\sigma$, and the same maximum bond size $\chi$.
\end{theorem}
\begin{proof}

For the contraction of each density matrix at vertex $v$ in the density matrix algorithm, a valid contraction path can be obtained by contracting the partition at $v$ into
a tensor first, then contracting it with other density matrices based on \eqref{eq:dm_cache}. The cost of this contraction path is an upper bound of the contraction cost of this density matrix, assuming the optimal contraction path is selected.

Therefore, the overall cost of the density matrix algorithm, assuming the optimal contraction path is used during the contraction of each density matrix, is upper-bounded by
the case where each partition embedded into every vertex $v \in V_T$ is contracted into a tensor $\mat{M}_v$ prior to conducting the depth-first search (DFS) traversal. This transforms the tensor network into an untruncated tree tensor network. According to \cref{lem:same_cost_tree_approximation} in \cref{sec:cost_analysis_appendix}, both the density matrix algorithm and the canonicalization-based algorithm exhibit the same asymptotic cost when truncating a tree tensor network.
By examining this particular case, we establish that the upper bound of the density matrix algorithm matches the asymptotic cost described in \cref{alg:canonize}. This finishes the proof.

\end{proof}

\section{The algorithm to approximate an input tensor network into an embedding tree}
\label{sec:swap}

\begin{algorithm}[!ht]
\setstretch{1.1}
\caption{\texttt{approx\_tensor\_network}: approximate a tensor network into an embedding tree}
\begin{algorithmic}[1]
\STATE{\textbf{Input:} The tensor network $\tsr{T}$ with graph $G_s= (V_s, E_s)$, the edge set ordering $\sigma^{(\Set{E}_s)}$, the edge orderings $\{\sigma^{(E')}: E'\in \Set{E}_s\}$,
the maximum bond size $\chi$, the swap batch size $r$, and ansatz $A$ \comment{The ansatz $A$ can be either ``MPS" or ``Comb"}
}
\STATE{$\tau^{(\Set{E}_s)} \leftarrow \texttt{linear\_ordering}\left(\Set{E}_s, G_s\right)$ \comment{Ordering generated via recursive bisection}
}
\STATE{$d\leftarrow \ktdist{\tau^{(\Set{E}_s)}, \sigma^{(\Set{E}_s)}}$ \comment{Number of adjacent edge set swaps needed to change $\tau^{(\Set{E}_s)}$ to $\sigma^{(\Set{E}_s)}$}}
\STATE{$n \leftarrow \lceil d/ r\rceil$ \comment{The number of density matrix algorithms to be performed}}
\STATE{$\hat{\sigma}_1\ldots, \hat{\sigma}_n\leftarrow$ $n$ equally-spaced inverval orderings that separate $\tau^{(\Set{E}_s)}$ and $\sigma^{(\Set{E}_s)}$}
\STATE $\tsr{X}_0\leftarrow \tsr{T}$
\FOR{$i\in \{1,\ldots, n\}$}
\STATE{
$T\leftarrow \texttt{embedding\_tree}\left(\hat{\sigma}_i, \{\sigma^{(E')}: E'\in \Set{E}\}, A\right)$ \comment{construct the embedding tree based on \cref{def:embed_mps} and \cref{def:embed_comb}}
}
\STATE{$\tsr{X}_i\leftarrow \texttt{density\_matrix\_alg}(\tsr{X}_{i-1}, T, \chi)$ \comment{\cref{alg:dm}}}
\ENDFOR
\RETURN the output tensor network $\tsr{X}_n$
\end{algorithmic}
\label{alg:approx_tensor_network}
\end{algorithm}

 \begin{figure}[!ht]
\centering
\includegraphics[width=.94\textwidth, keepaspectratio]{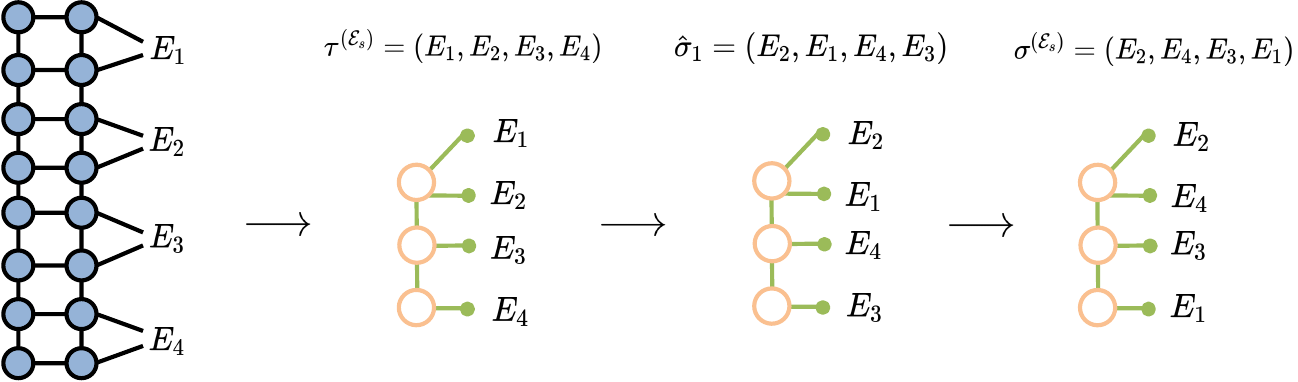}

\caption{
Illustration of \cref{alg:approx_tensor_network} with the swap batch size being $r=2$. The left diagram denotes the input tensor network $G_s$ as well as the set of edge subsets $\Set{E}_s = \{E_1,E_2,E_3,E_4\}$. The second leftmost diagram shows the ordering $\tau^{(\Set{E}_s)}$ that is generated based on analyzing the graph structure of $G_s$. Since 4 swaps are needed to change $\tau^{(\Set{E}_s)}$ to $\sigma^{(\Set{E}_s)}$, two density matrix algorithms are performed, one with the embedding tree generated by the ordering $\hat{\sigma}_1$ and the other with the embedding tree generated by the ordering $\sigma^{(\Set{E}_s)}$.
}
\label{fig:vis_swap}
\end{figure}
We introduce a hybrid algorithm that combines the density matrix algorithm with the swap-based algorithm to approximate an input tensor network $G_s=(V_s,E_s)$ into an embedding tree. This hybrid algorithm offers a compromise between accuracy and computational cost by performing multiple iterations of the density matrix algorithm. Each iteration incrementally modifies the structure of the tensor network by a small degree, ensuring that the overall computational cost remains manageable. While this approach may sacrifice a certain degree of approximation accuracy, it provides a balanced solution that achieves a reasonable trade-off between accuracy and computational efficiency compared to the pure density matrix algorithm.

We present the algorithm in \cref{alg:approx_tensor_network}, and an illustration is shown in \cref{fig:vis_swap}.
In this algorithm, we denote the edge set ordering in the embedding tree as $\sigma^{(\mathcal{E}_s)}$, and the reference edge set ordering of $G_s$ as $\tau^{(\mathcal{E}_s)}$. We measure the structural difference between $G_s$ and the embedding tree using the Kendall-Tau distance in \cref{defn:ktdist}, defined as $d = \ktdist{\tau^{(\mathcal{E}_s)}, \sigma^{(\mathcal{E}_s)}}$.
The algorithm utilizes a parameter
$r$ to control the extent of structural modifications made by each density matrix algorithm iteration. The number of density matrix algorithms performed is determined by $\lceil d/r \rceil$. Users can choose different values of $r$ depending on the specific problem.
By selecting a larger value of $r$, the behavior of the algorithm closely resembles that of the pure density matrix algorithm. On the other hand, a smaller value of $r$ generally leads to improved computational efficiency while sacrificing some approximation accuracy.

In summary, the hybrid approach within the proposed \texttt{partitioned\_contract} algorithm effectively balances accuracy and computational cost for specific problem instances.

\section{Experimental results}
\label{sec:exp}
In this section, we present the results of a series of experiments to evaluate the performance of the proposed approach. All experiments were executed on an Intel Core i7 2.9 GHz Quad-Core machine.

In \cref{subsec:imple}, we introduce our implementations, the tensor networks and models tested in our experiments.
In \cref{subsec:exp_dm}, we conducted a detailed comparison between the proposed density matrix algorithm for tree approximation and the canonicalization-based algorithm. Across all experiments, the density matrix algorithm consistently demonstrated either lower or the same asymptotic cost. In particular, we achieved a remarkable 4.9X speedup with the density matrix algorithm compared to the canonicalization-based algorithm when approximating an MPO-MPS multiplication into an MPS.

In \cref{subsec:exp_par_contract}, we justify the \texttt{partitioned\_contract} algorithm presented in \cref{alg:ovw}. We justify our embedding tree selection algorithm and explore the impact of the environment size on accuracy and efficiency across multiple problems. Additionally, we conduct a comprehensive comparison between the MPS and the comb ansatz.
Furthermore, we evaluate \texttt{partitioned\_contract}, the CATN algorithm~\cite{pan2020contracting}\footnote{We use the CATN implementation at \url{https://github.com/panzhang83/catn}.}, SweepContractor~\cite{chubb2021general}\footnote{We use the SweepContractor implementation at \url{https://github.com/chubbc/SweepContractor.jl}.}, and hyperoptimized approximate contraction \cite{gray2022hyper}\footnote{We use the hyperoptimized approximate contraction implementation at quimb~\cite{gray2018quimb} (\url{https://github.com/jcmgray/quimb}).},  in contracting tensor networks defined on lattices and random regular graphs as well as tensor networks from random quantum circuit simulation.
We demonstrate a 9.2X speed-up while maintaining the same level of accuracy
when contracting tensor networks defined on 3D lattices using
the Ising model.

\subsection{Implementations, tested tensor networks, and the evaluation}\label{subsec:imple}

The proposed algorithms in the paper are being developed at \url{https://github.com/ITensor/ITensorNetworks.jl}.
ITensorNetworks.jl is a publicly available Julia~\cite{bezanson2017julia} package built for manipulating tensor networks of arbitrary geometry, and is built on top of ITensors.jl~\cite{fishman2022itensor}. The library also provides an interface to OMEinsumContractionOrders.jl\footnote{The library is implemented at \url{https://github.com/TensorBFS/OMEinsumContractionOrders.jl}.}, which implements multiple heuristics introduced in \cite{gray2021hyper,liu2023computing} to generate efficient contraction paths for exact tensor network contractions.
For all the results presented in this work, we use the Simulated Annealing bipartition + Greedy algorithm (SABipartite)~\cite{liu2023computing} to generate contraction paths for exact tensor network contractions.

Our experiments consider three types of tensor networks: those generated from random tensors,
the classical Ising model, and random quantum circuits.

For those generated from random tensors, each element within the tensors is an i.i.d. variable uniformly distributed in the range of $[\alpha, 1]$, where $\alpha \in [-1, 0]$. These particular tensor networks have been utilized in previous research~\cite{gray2022hyper} as benchmarks for evaluating contraction algorithms.
For specific structures like random regular graphs and 3D lattices, the approximate contraction of the tensor network becomes more challenging as $\alpha$ approaches the value of $-1$~\cite{chen2024sign}.

For a tensor network defined on a graph $G=(V,E)$ using the ferromagnetic Ising model, the contraction output, denoted as $Z$ and referred to as the partition function, can be expressed as follows,
\[
Z = \sum_{\sigma_i,\sigma_j\in\{-1,1\}}\prod_{(i,j)\in E}\exp(\beta \sigma_i\sigma_j).
\]
In the tensor network, the tensor $\tsr{T}^{(v)}$ defined at each
$v\in V$ has an elementwise expression of
\[
t^{(v)}_{E(v)} = \sum_i \prod_{e\in E(v)}W_{i,e},
\]
where
\[
W =  \frac{1}{\sqrt{2}}\begin{bmatrix}
\sqrt{\cosh(\beta)} + \sqrt{\sinh(\beta)} & \sqrt{\cosh(\beta)} - \sqrt{\sinh(\beta)} \\
\sqrt{\cosh(\beta)} - \sqrt{\sinh(\beta)} & \sqrt{\cosh(\beta)} + \sqrt{\sinh(\beta)}
\end{bmatrix}
\]
and $\beta$ is an input parameter to the model.
We show the relation between the relative error of $\ln Z$ and the running time of \texttt{partitioned\_contract} and the baselines in  \cref{subsec:exp_par_contract}.
The quantity $\ln Z$ is an important measure that is proportional to the free energy of the system.

We also explore the simulation of a 2D random quantum circuit, denoted as $|{\psi}\rangle$, as detailed in~\cite{pang2020efficient,boixo2018characterizing,arute2019quantum}. The initial quantum state $|{0,\ldots,0}\rangle$ is organized into a $6 \times 6$ grid, and we apply six layers of random circuit gates to this initial state. Each layer contains random one-qubit rotations on top of each qubit and a sequence of two-qubit controlled-X gates, and the two-qubit gates are structured in a brick-layer pattern. The specific configuration of each 2-qubit layer is detailed in \cref{fig:layers_rqc}. In \cref{subsec:exp_par_contract}, we approximately contract the tensor network $\langle\psi|\psi\rangle$, and measure the absolute error of the quantity.

 \begin{figure}[!ht]
\centering
\includegraphics[width=.8\textwidth, keepaspectratio]{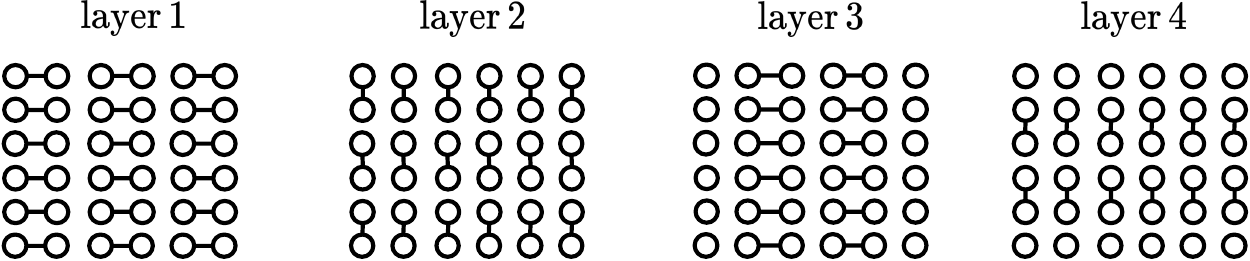}

\caption{The arrangement of the initial four layers of quantum gates. Each circle denotes a qubit, and each line denotes a two-qubit gate. The subsequent layers follow the same pattern.
}
\label{fig:layers_rqc}
\end{figure}

To evaluate and compare the efficiencies of various algorithms, we measure both the execution time and the required number of GFlops (giga floating-point operations).
The GFlops calculations encompass tensor contractions, QR factorization, and low-rank approximations, as outlined in the model detailed in \cref{subsec:cost_model}.
It is worth noting that in our reported results, the execution time excludes the graph analysis part, which involves graph embedding and computing the contraction sequence of given tensor networks. This part remains independent of the tensor network ranks and is negligible when the ranks are high.

\subsection{Comparison between the density matrix algorithm and the canonicalization-based algorithm}\label{subsec:exp_dm}

 \begin{figure}[!ht]
\centering
\subfloat[MPS, $\chi=100$]{
\includegraphics[width=.242\textwidth, keepaspectratio]
{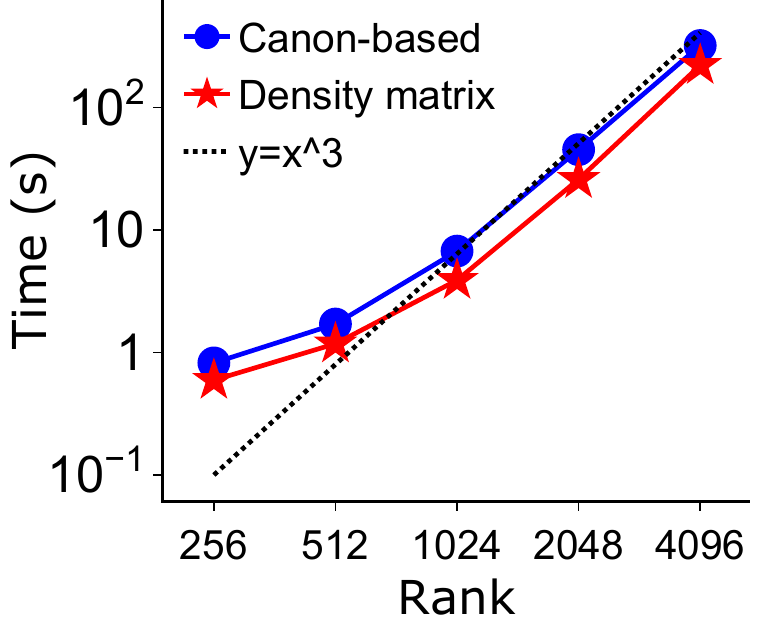}\label{subfig:bench_mps1}}
\subfloat[MPS, $\chi=100$]{\includegraphics[width=.242\textwidth, keepaspectratio]{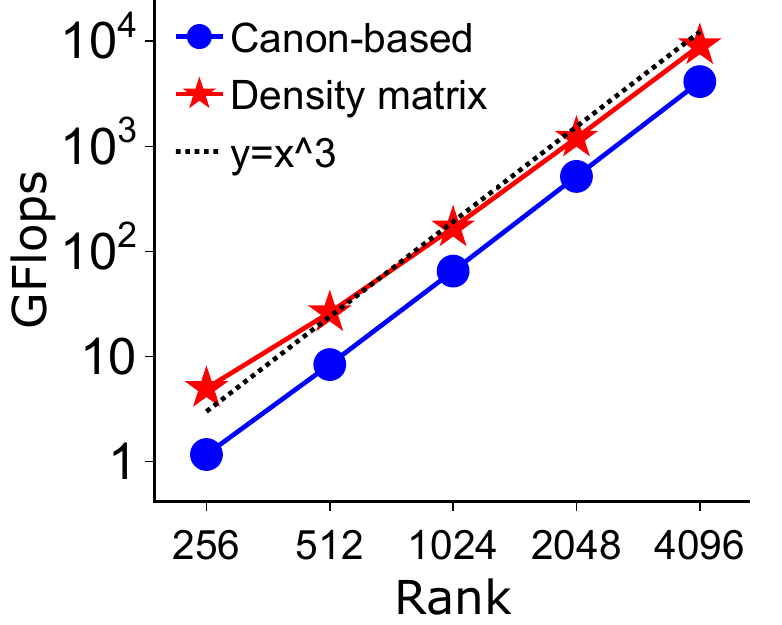}
\label{subfig:bench_mps2}}
\subfloat[BBT, $\chi=50$]{
\includegraphics[width=.242\textwidth, keepaspectratio]
{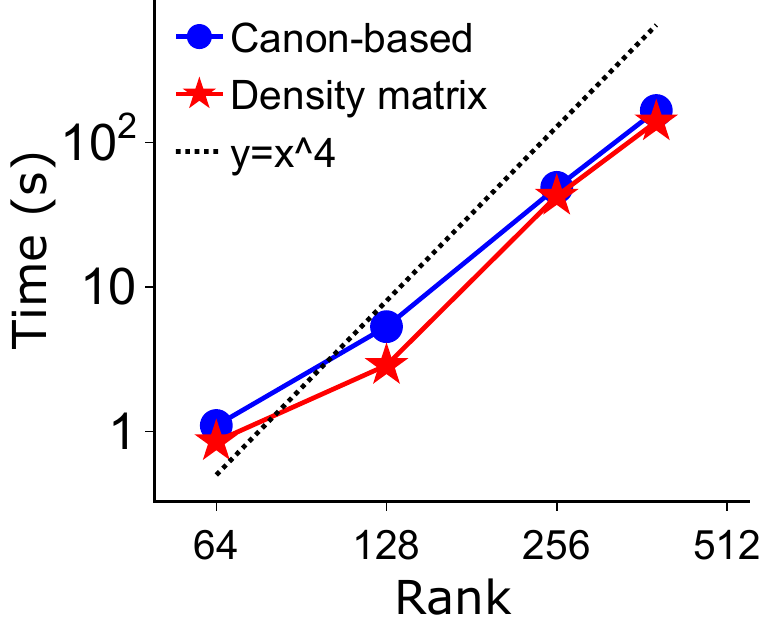}\label{subfig:bench_btree1}}
\subfloat[BBT, $\chi=50$]{\includegraphics[width=.242\textwidth, keepaspectratio]{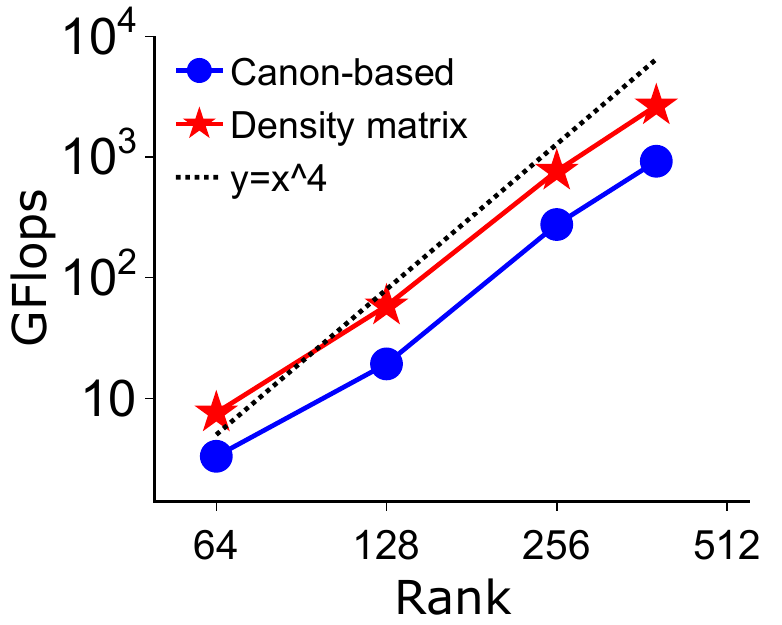}
\label{subfig:bench_btree2}}

\caption{Performance comparison between the density matrix algorithm and the canonicalization-based algorithm in truncating a binary tree tensor network.
In (a)(b), the input networks are MPSs with different ranks. In (c)(d), the inputs are balanced binary tree (BBT) tensor networks with different ranks.
The number of uncontracted modes is fixed to be 30 for all input tensor networks.
}
\label{fig:bench_truncation}
\end{figure}

We conduct an efficiency comparison between the density matrix algorithm and the canonicalization based algorithm to approximate an input tensor network into a binary tree tensor network.
Our evaluation covers scenarios where the input tensor network structure matches the output structure, as well as cases where the input network has a general non-tree structure. In both instances, the density matrix algorithm has equal or superior asymptotic cost compared to the canonicalization-based algorithm.

In \cref{fig:bench_truncation}, we conduct a performance comparison of truncating both MPSs and balanced binary tree tensor networks. Let $R$ denote the rank of the input MPS and the balanced binary tree, the analytical asymptotic cost for truncating an MPS is $\Theta(R^3)$, whereas for truncating a balanced binary tree is $\Theta(R^4)$.
As depicted in the results, the scaling behavior of both algorithms aligns with the analytical predictions. Despite the density matrix algorithm incurring a constant overhead in terms of GFlops, we observe that it exhibits slightly faster performance. This advantage can be attributed to the fact that the majority of the density matrix algorithm's execution time is spent on tensor contractions, which are practically faster compared to matrix factorizations, even though both operations have a similar computational complexity.
This property makes the density matrix algorithm more favorable on GPUs due to their ability to efficiently run tensor contractions in parallel.

 \begin{figure}[!ht]
\centering
\subfloat[]{
\includegraphics[width=.43\textwidth, keepaspectratio]
{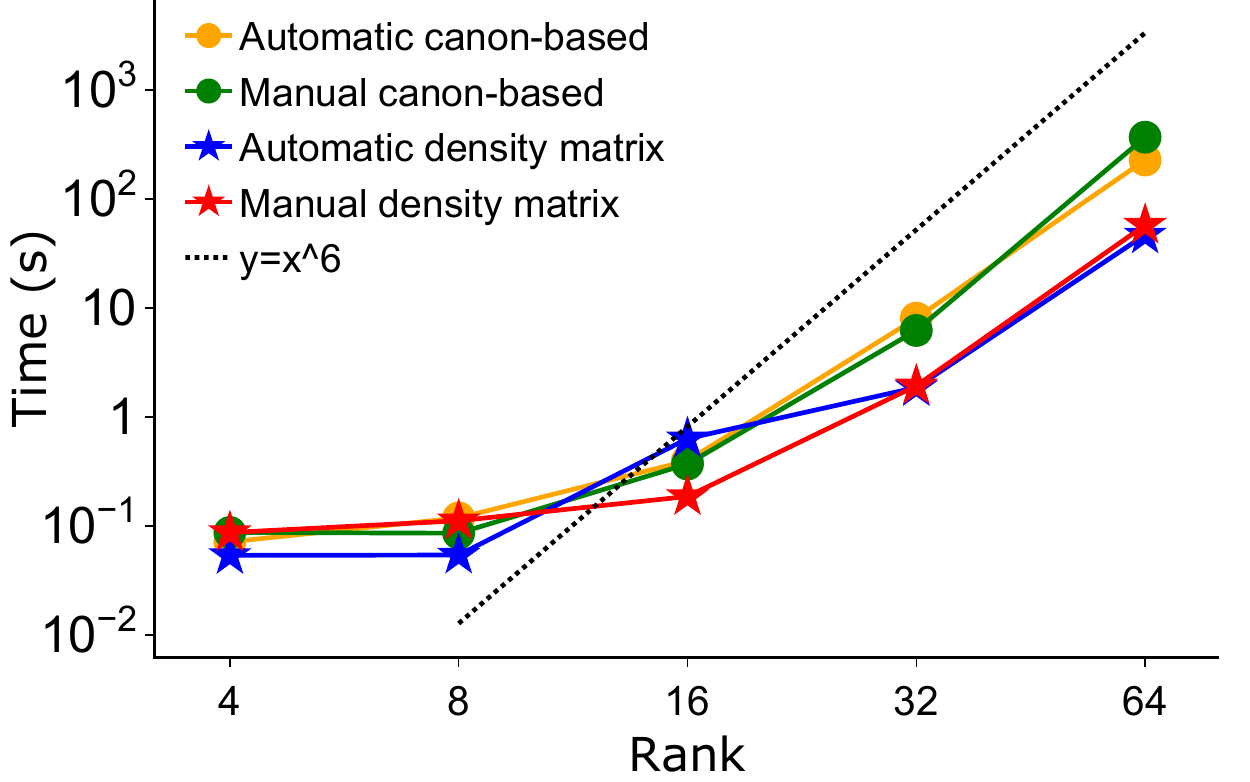}\label{subfig:bench_mpomps_1}}
\subfloat[]{\includegraphics[width=.43\textwidth, keepaspectratio]{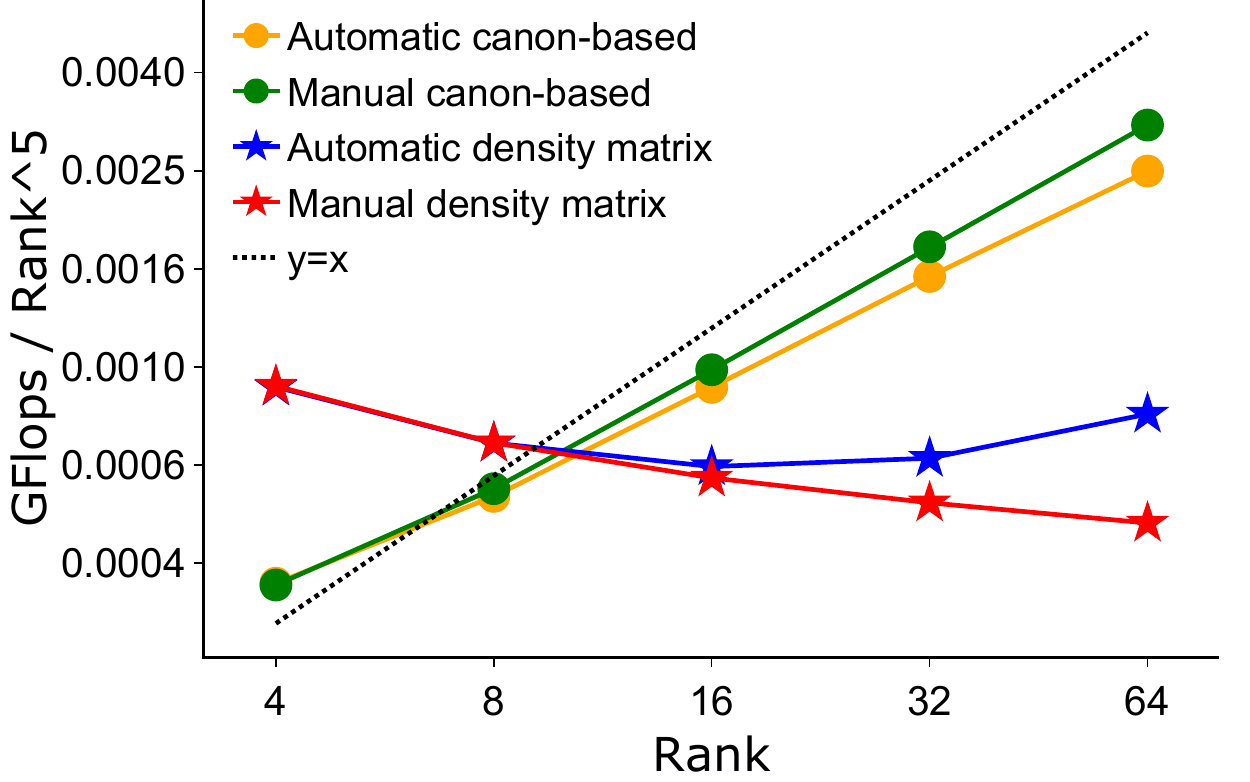}
\label{subfig:bench_mpomps_2}}

\caption{Performance comparison between the density matrix algorithm and the canonicalization-based algorithm in approximating the MPO-MPS multiplication into a low-rank MPS. The order of the input MPS and MPO  is fixed to be 40.
Both the input MPS and MPO have the same rank $\chi$, and the output MPS rank is also upper-bounded by $\chi$. The manual algorithms are those reviewed in \cref{subsec:2d} that use a manually-determined memoization strategy.
}
\label{fig:bench_mpo_mps}
\end{figure}

In \cref{fig:bench_mpo_mps}, we compare the performance of truncating the multiplication of an MPS and an MPO.
Our experiments encompass both the canonicalization-based algorithm and the density matrix algorithm, alongside the reference algorithms reviewed in \cref{subsec:2d}. In the reference algorithms, the memoization strategy is determined and implemented manually rather than automatically.
For the canonicalization-based algorithm, the asymptotic cost is $\Theta(R^6)$, where $R$ represents the input rank of both MPS and MPO. On the contrary, the density matrix algorithm exhibits an asymptotic cost of $\Theta(R^5)$.
As shown in \cref{subfig:bench_mpomps_2}, the scaling behavior of both algorithms aligns with our analysis. The density matrix algorithm outperforms the canonicalization-based algorithm and has a remarkable 4.9X execution time speedup when the input rank is 64. Furthermore, our algorithm, equipped with the automatically-chosen memoization strategy, performs similarly to the reference algorithms, thereby confirming the efficacy of our approach.

 \begin{figure}[!ht]
\centering
\subfloat[MPS, $\chi=250$]{
\includegraphics[width=.242\textwidth, keepaspectratio]
{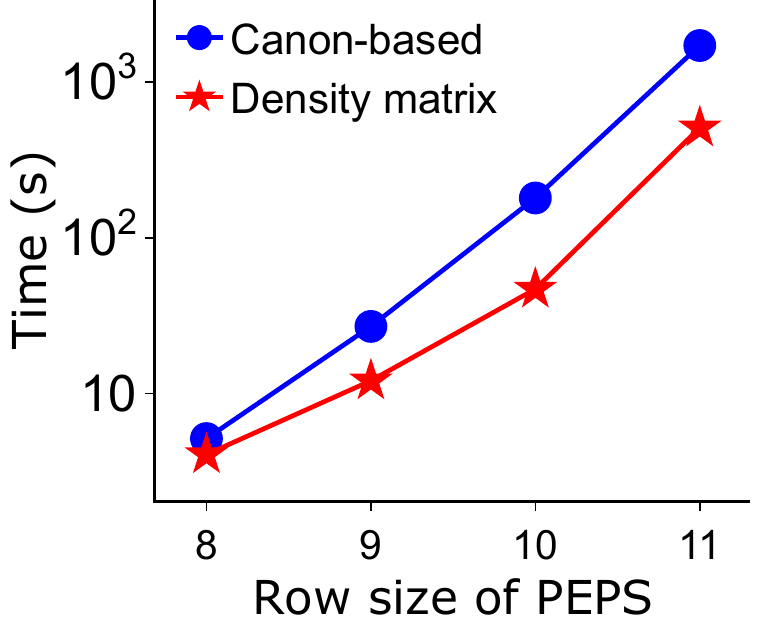}\label{subfig:bench_peps_mps1}}
\subfloat[MPS, $\chi=250$]{\includegraphics[width=.242\textwidth, keepaspectratio]{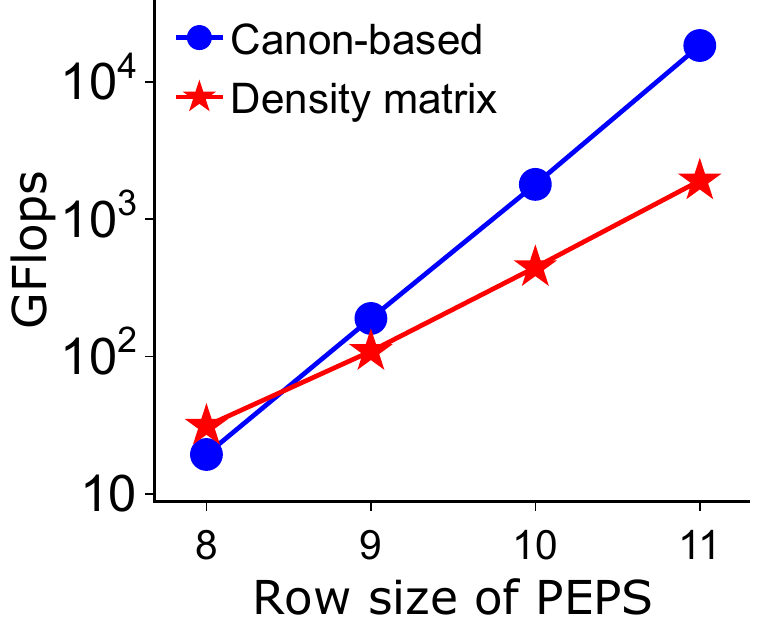}
\label{subfig:bench_peps_mps2}}
\subfloat[Comb, $\chi=50$]{
\includegraphics[width=.242\textwidth, keepaspectratio]
{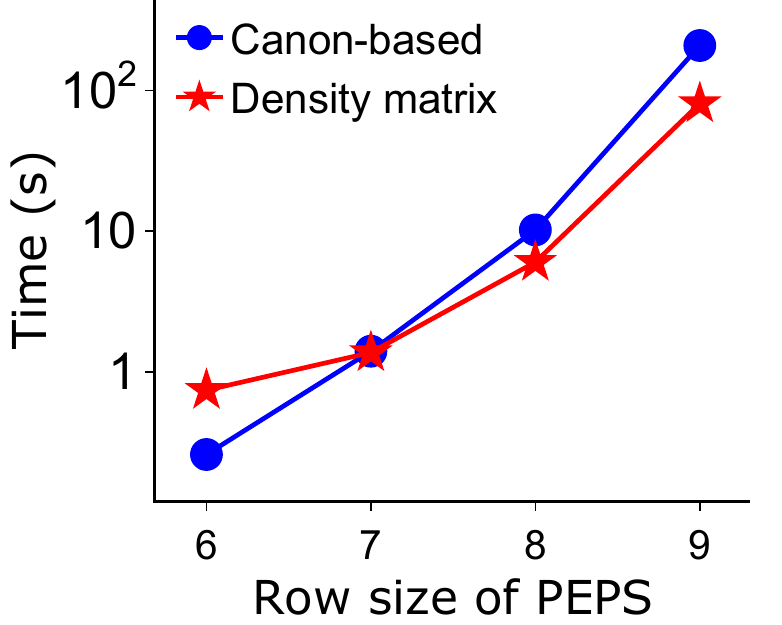}\label{subfig:bench_peps_comb1}}
\subfloat[Comb, $\chi=50$]{\includegraphics[width=.242\textwidth, keepaspectratio]{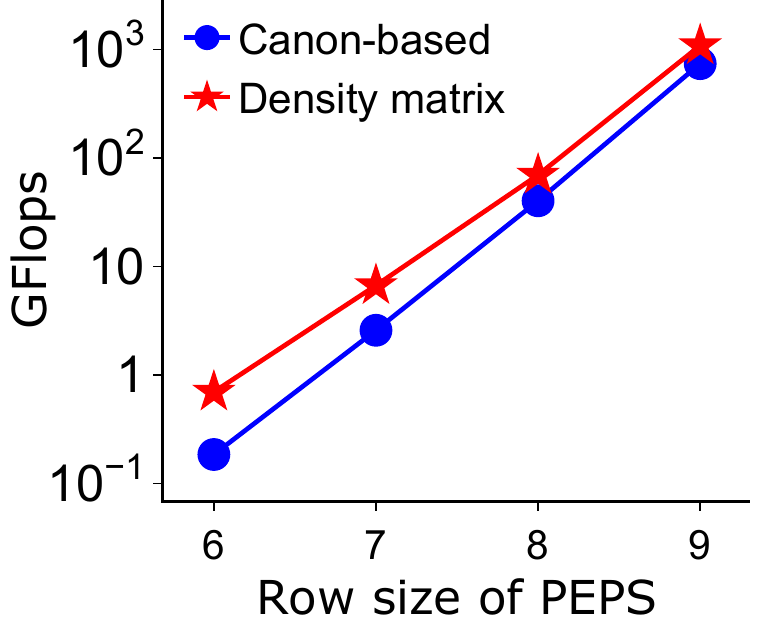}
\label{subfig:bench_peps_comb2}}

\caption{Performance comparison between the density matrix algorithm and the canonicalization-based algorithm in approximating a PEPS with rank 2 into a binary tree tensor network.
The column size of the PEPS equals the row size.
In (a)(b), the embedding tree structure is an MPS, and the MPS site ordering is chosen based on the row- or column-wise traversal of the 2D coordinates of the PEPS tensors. In (c)(d), the embedding tree structure is a comb, and each edge subset in the comb is a row of the PEPS.
}
\label{fig:bench_peps}
\end{figure}

In \cref{fig:bench_peps}, we compare the performance of approximating PEPS on square grids into MPS and comb binary tree structures. Both structures are defined in \cref{subsec:embed_tree}.
As can be seen, the density matrix algorithm outperforms the canonicalization-based algorithm when the number of rows and columns of the PEPS is large. The inefficiency of the canonicalization-based algorithm is due to the fact that there exists some partition embedded in one vertex of the MPS/comb whose contraction yields a large-sized tensor. The density matrix algorithm avoids the explicit formation of such tensors and thus is more efficient. In later sections we will discuss the relative merits of using MPS or comb tree structures for intermediate networks.

\subsection{Benchmark of the \texttt{partitioned\_contract} algorithm}\label{subsec:exp_par_contract}

 \begin{figure}[!ht]
\centering
\subfloat[2D lattice, rank=16]{
\includegraphics[width=.3\textwidth, keepaspectratio]
{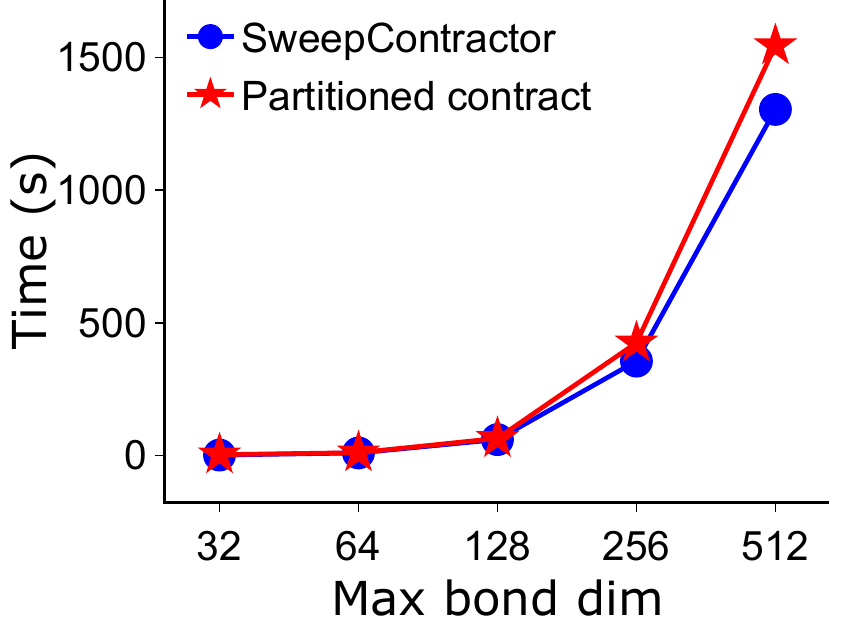}\label{subfig:sweep1}}
\subfloat[3D lattice, rank=4]{\includegraphics[width=.3\textwidth, keepaspectratio]{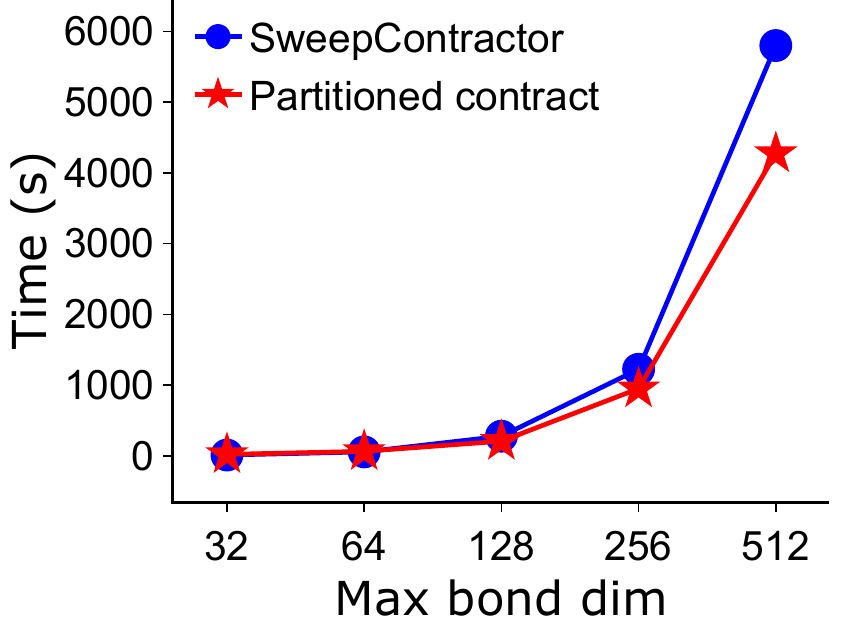}
\label{subfig:sweep2}}
\subfloat[Random regular graph, rank=8]{
\includegraphics[width=.3\textwidth, keepaspectratio]
{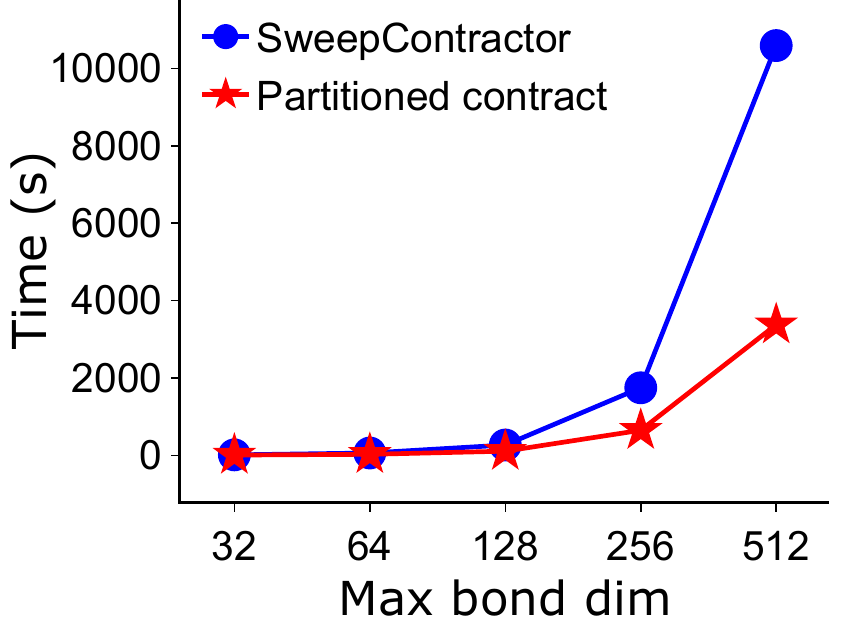}\label{subfig:sweep3}}

\caption{Performance comparison between \texttt{partitioned\_contract} and SweepContractor under the same contraction path. The swap batch size is set to 1 for all experiments in \texttt{partitioned\_contract}. In (a), the row and the column size of the 2D lattice is 8. In (b), each mode in the 3D lattice has a size of 5. In (c), the random regular graph has 100 vertices, each with a degree of 3.
}
\label{fig:compare_w_sweep}
\end{figure}

\paragraph{Impact of the embedding tree on contraction efficiency}

In this section we present results that justify our embedding tree selection algorithm in \texttt{partitioned\_contract}. In \cref{fig:compare_w_sweep}, we compare \texttt{partitioned\_contract} with SweepContractor for tensor networks defined on three different structures using tensor networks with random tensor elements that was introduced in \cref{subsec:imple}. For all the experiments, \texttt{partitioned\_contract} uses the MPS ansatz, and
both algorithms use the same maximally-unbalanced contraction tree where each partition only contains one tensor.
Consequently, the only distinction between the two algorithms lies in the usage of different embedding trees for each contraction between an MPS and a tensor.

As can be seen, both algorithms have a similar performance when contracting a 2D grid, while \texttt{partitioned\_contract} significantly outperforms SweepContractor for the other two graph structures.
This difference in performance arises from the fact that different embedding trees result in varying numbers of adjacent swaps of MPS modes. For tensor networks defined on 3D lattice and random regular graphs, our algorithm generates embedding trees that lead to substantially fewer adjacent swaps.
 Note that the \texttt{partitioned\_contract} algorithm achieves higher approximation accuracy on these two graphs, as fewer swaps imply reduced truncations, contributing to improved accuracy in the results.

 \begin{figure}[!ht]
\centering

\subfloat[$5 \times 5 \times 5$, $\beta=0.3$]{
\includegraphics[width=.32\textwidth, keepaspectratio]
{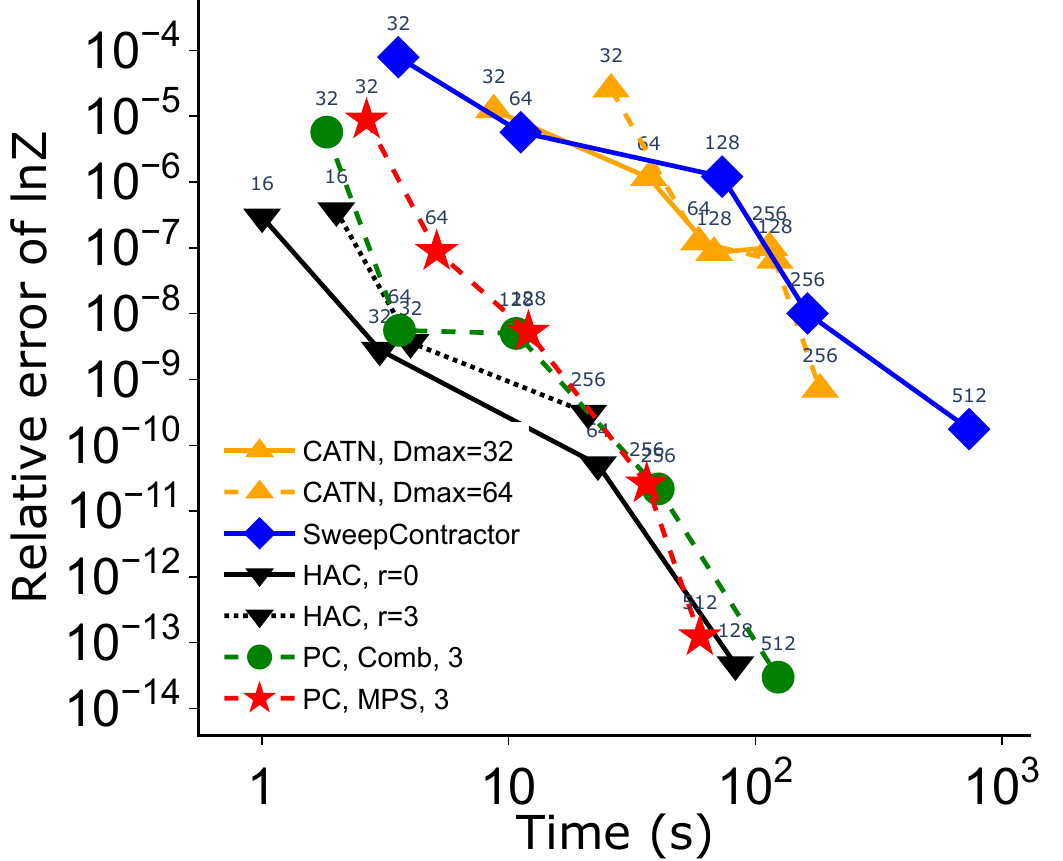}\label{subfig:3d_ising_1}}
\subfloat[$6 \times 6 \times 6$, $\beta=0.3$]{
\includegraphics[width=.32\textwidth, keepaspectratio]
{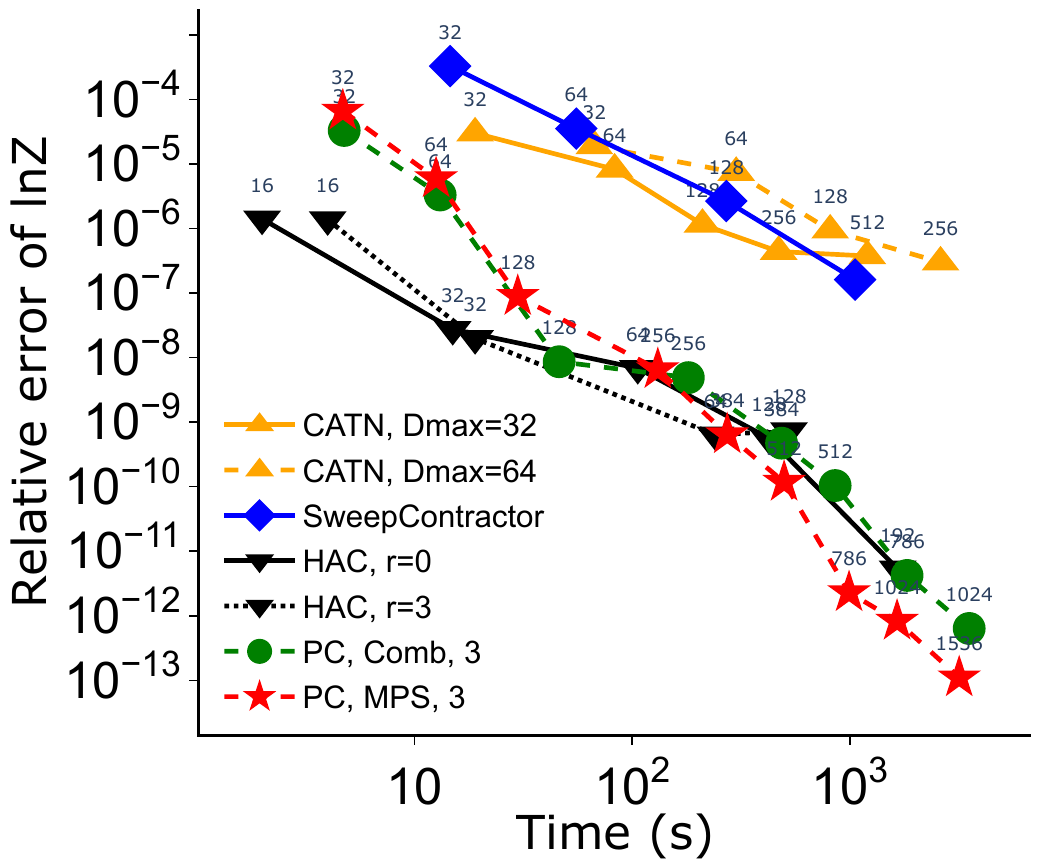}\label{subfig:3d_ising_4}}
\subfloat[$28 \times 28$, $\beta=0.44$]{
\includegraphics[width=.32\textwidth, keepaspectratio]
{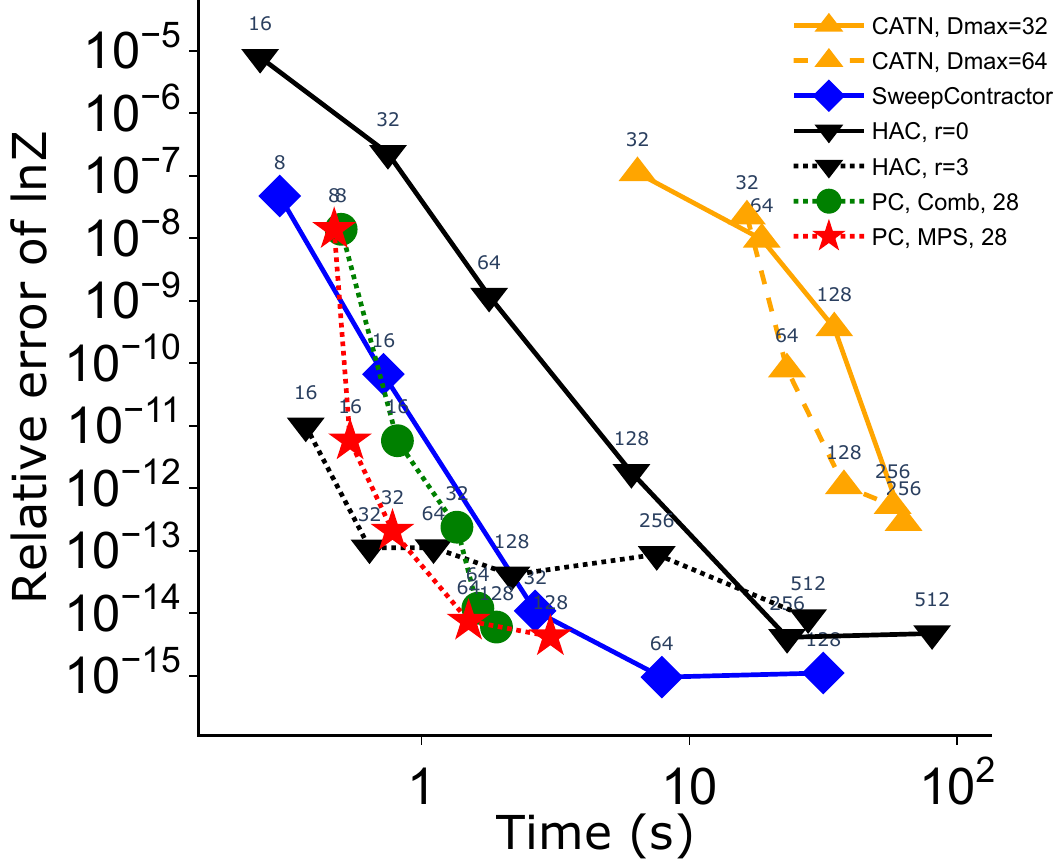}\label{subfig:2d_ising_1}}

\subfloat[$5 \times 5 \times 5$, $\beta=0.3$]{\includegraphics[width=.24\textwidth, keepaspectratio]{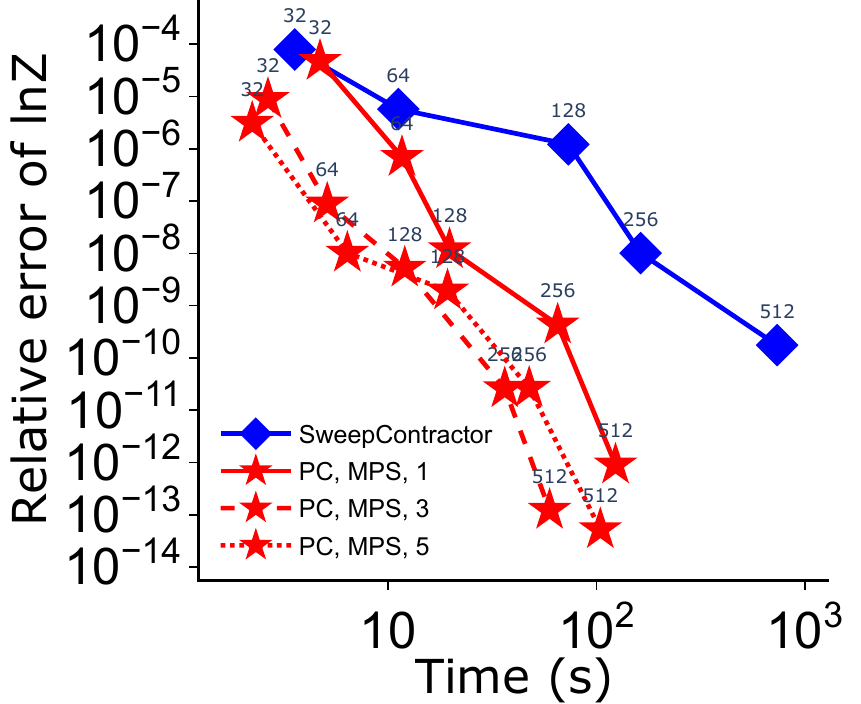}
\label{subfig:3d_ising_2}}
\subfloat[$5 \times 5 \times 5$, $\beta=0.3$]{
\includegraphics[width=.24\textwidth, keepaspectratio]
{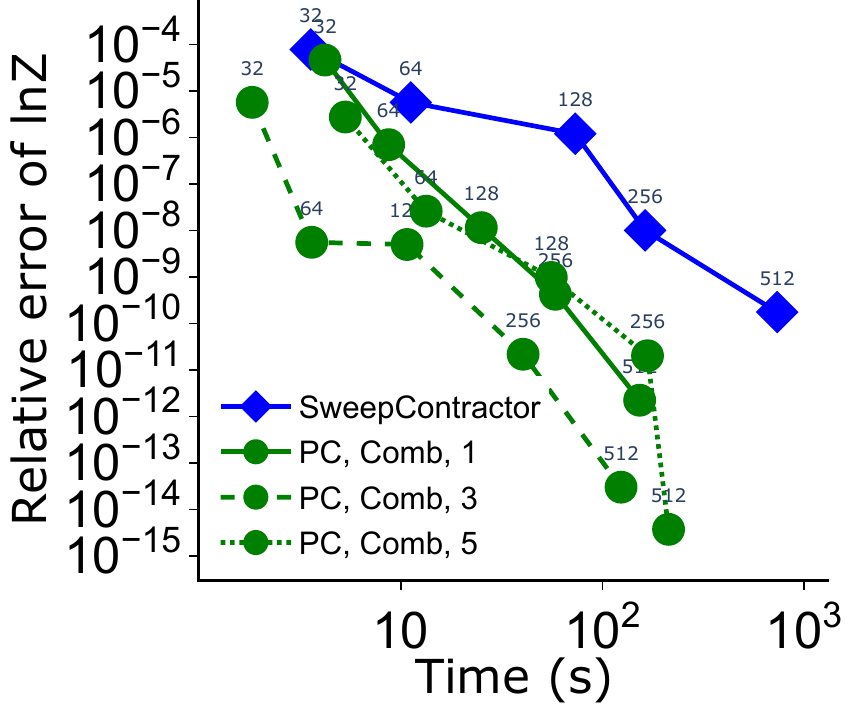}\label{subfig:3d_ising_3}}
\subfloat[$5\times 5 \times 5$, $\beta=0.3$]{\includegraphics[width=.24\textwidth, keepaspectratio]{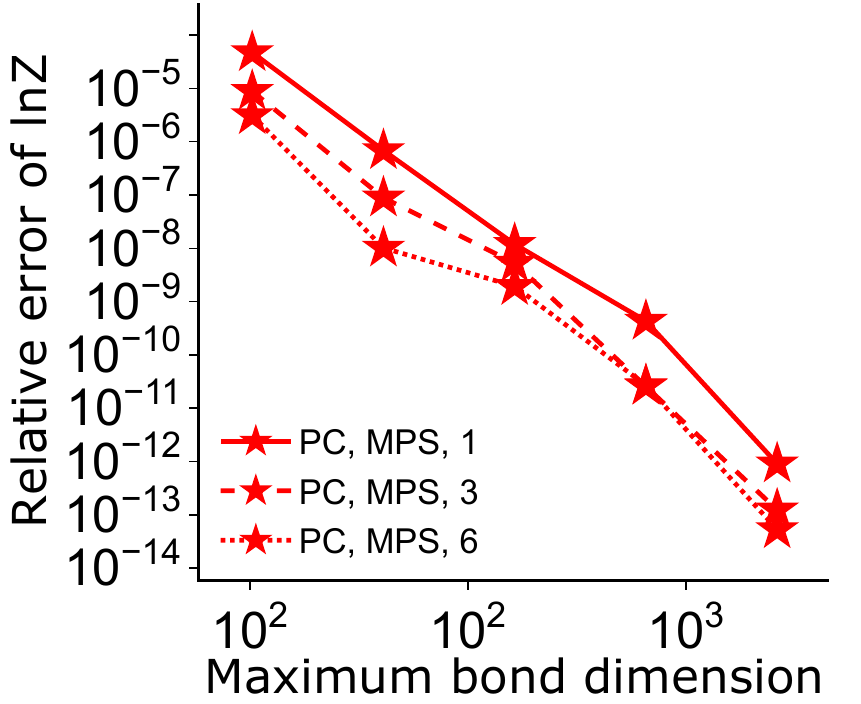}
\label{subfig:bond2}}
\subfloat[$5\times 5 \times 5$, $\beta=0.3$]{
\includegraphics[width=.24\textwidth, keepaspectratio]
{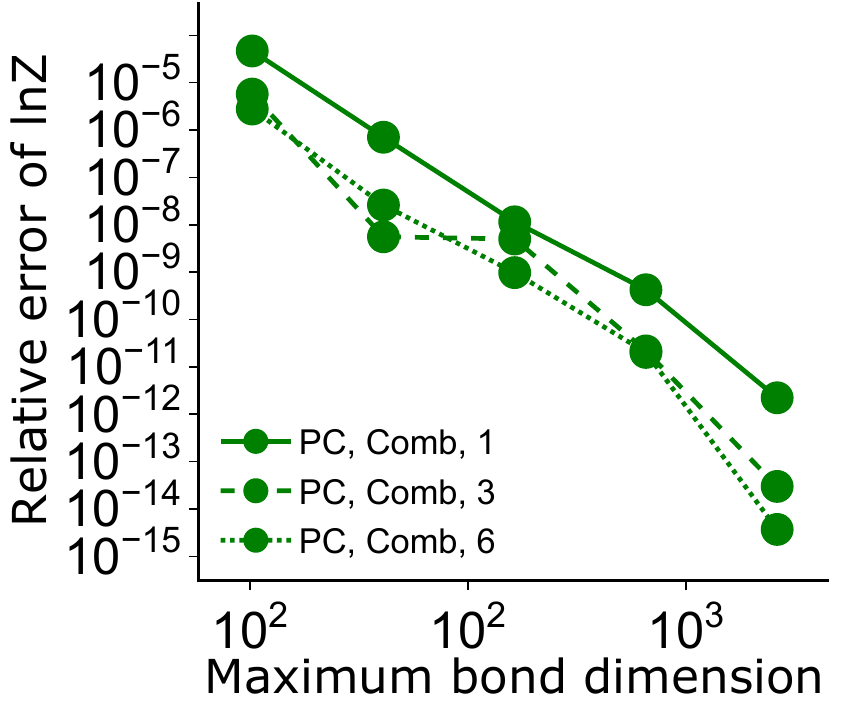}\label{subfig:bond1}}

\subfloat[$6 \times 6 \times 6$, $\beta=0.3$]{\includegraphics[width=.24\textwidth, keepaspectratio]{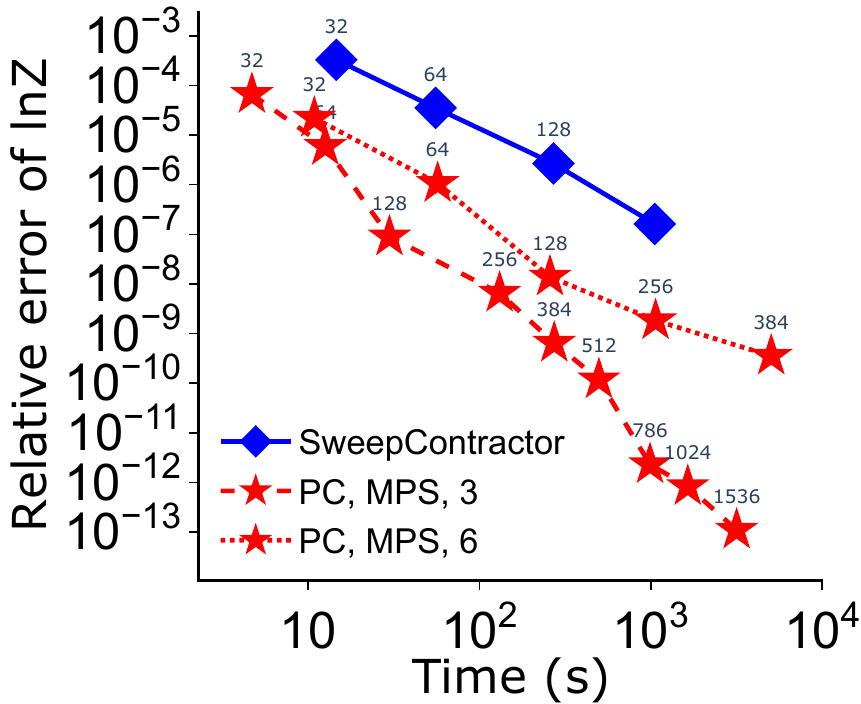}
\label{subfig:3d_ising_5}}
\subfloat[$6 \times 6 \times 6$, $\beta=0.3$]{
\includegraphics[width=.24\textwidth, keepaspectratio]
{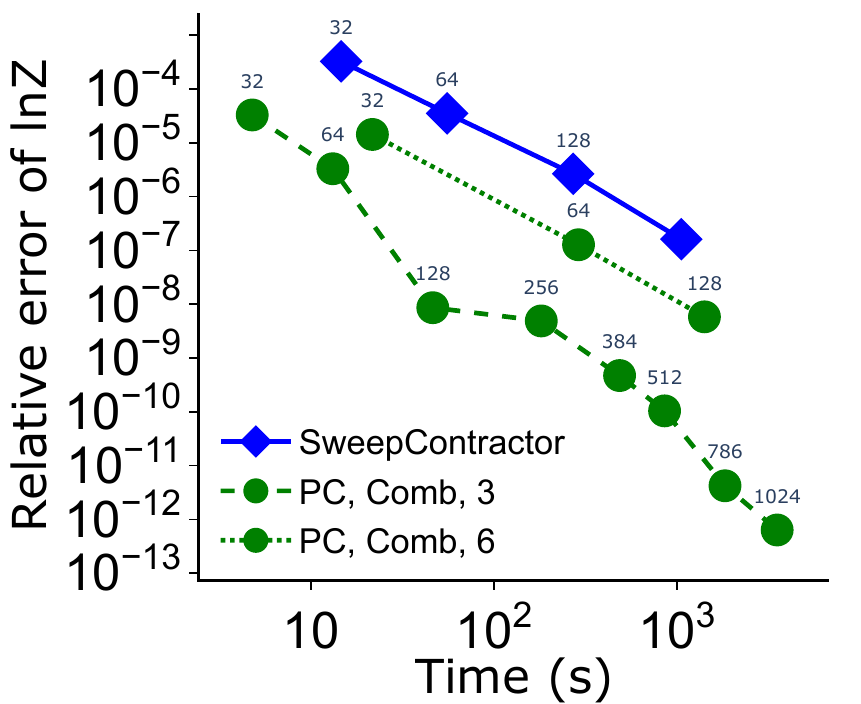}\label{subfig:3d_ising_6}}
\subfloat[$6\times 6 \times 6$, $\beta=0.3$]{\includegraphics[width=.24\textwidth, keepaspectratio]{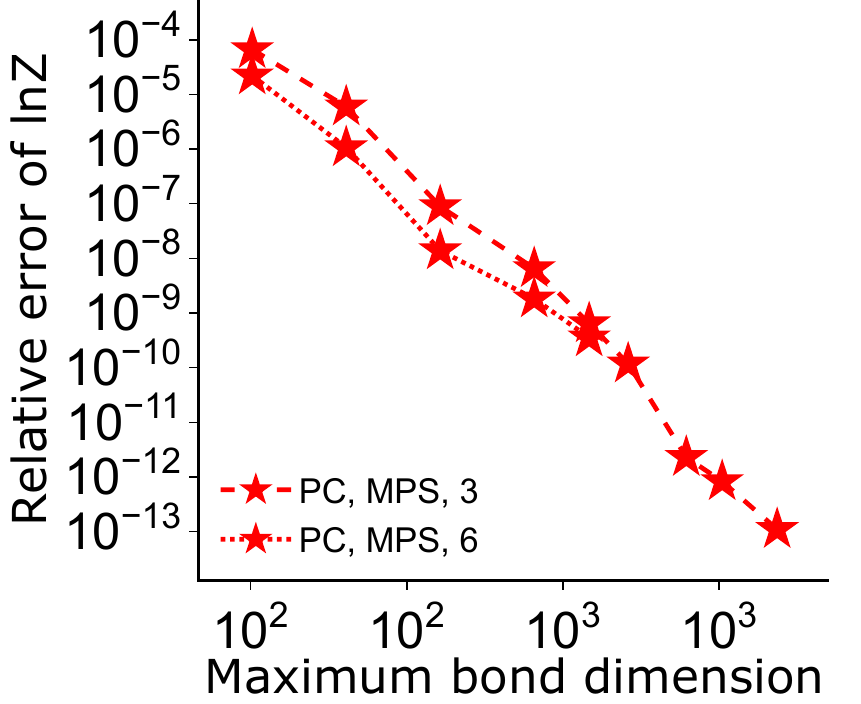}
\label{subfig:bond_666_2}}
\subfloat[$6\times 6 \times 6$, $\beta=0.3$]{
\includegraphics[width=.24\textwidth, keepaspectratio]
{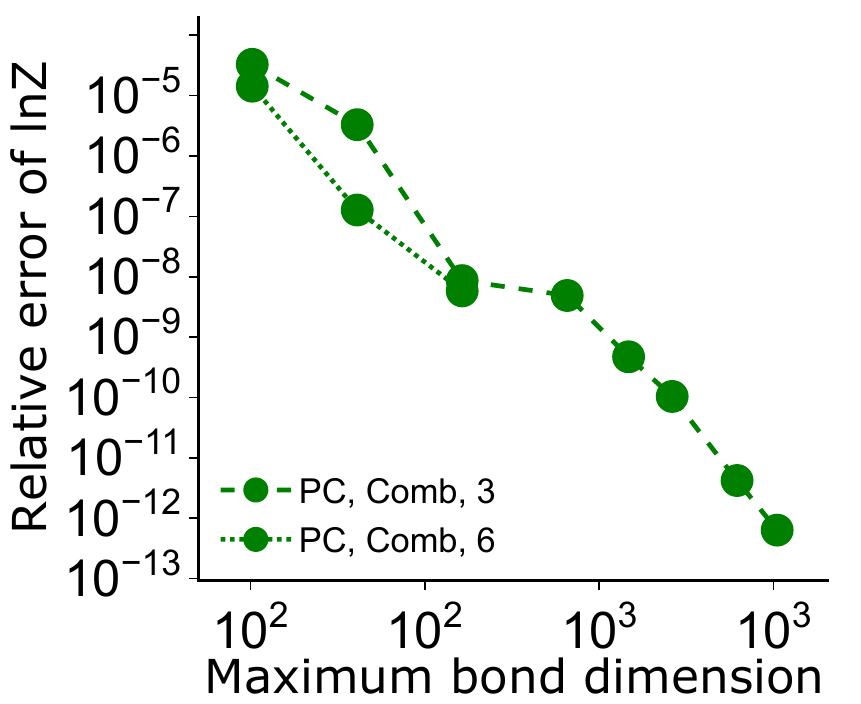}\label{subfig:bond_666}}

\subfloat[$28 \times 28$, $\beta=0.44$]{\includegraphics[width=.24\textwidth, keepaspectratio]{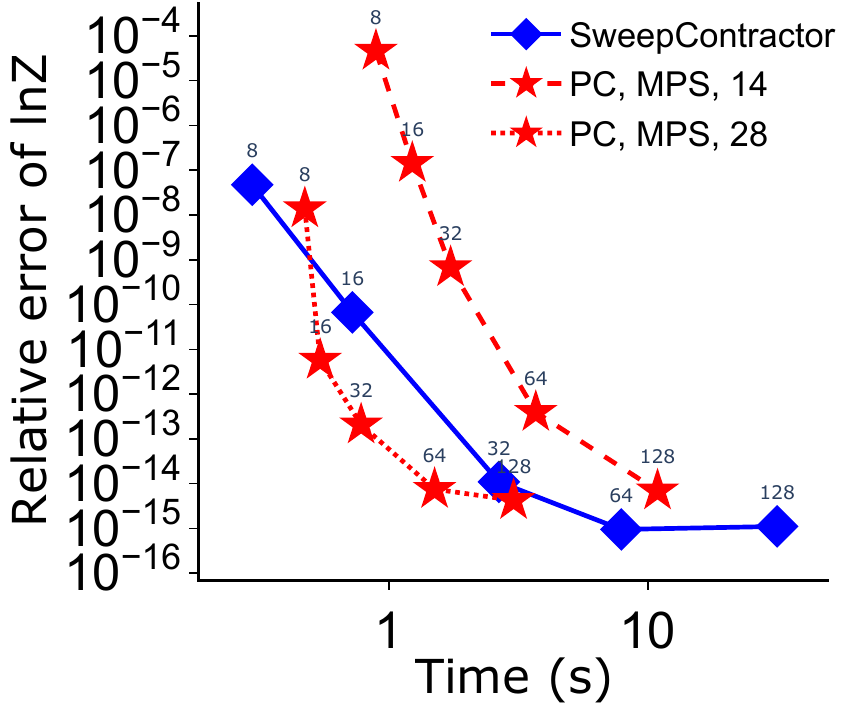}
\label{subfig:2d_ising_2}}
\subfloat[$28 \times 28$, $\beta=0.44$]{
\includegraphics[width=.24\textwidth, keepaspectratio]
{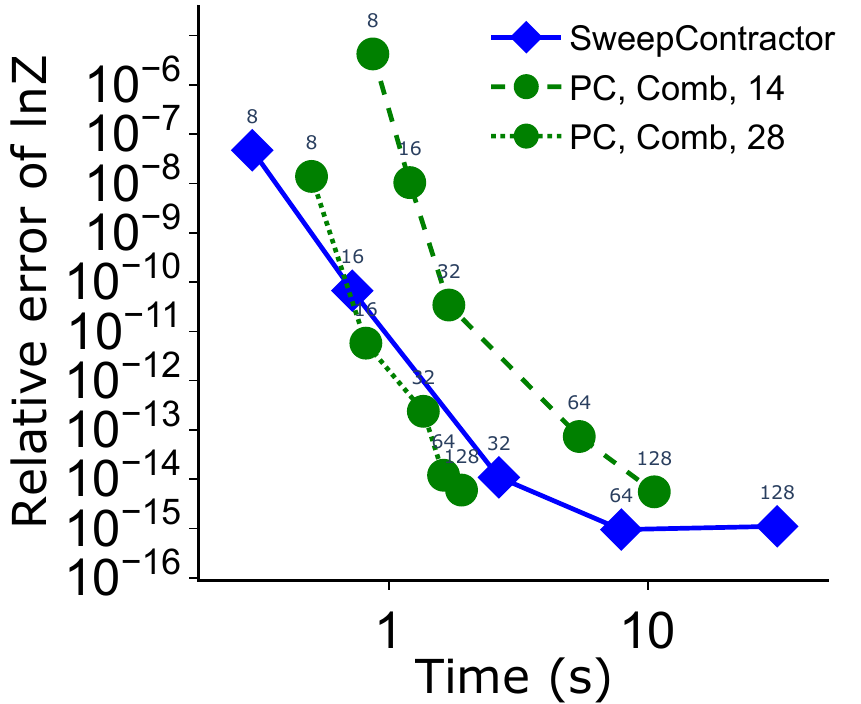}\label{subfig:2d_ising_3}}
\subfloat[$28 \times 28$, $\beta=0.44$]{\includegraphics[width=.24\textwidth, keepaspectratio]{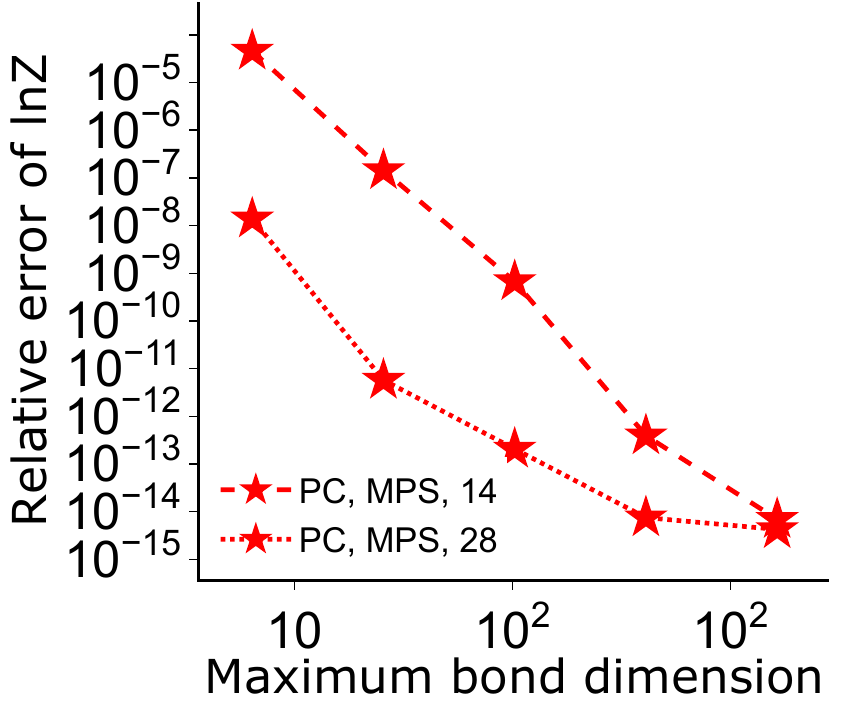}
\label{subfig:bond_2d_2}}
\subfloat[$28 \times 28$, $\beta=0.44$]{
\includegraphics[width=.24\textwidth, keepaspectratio]
{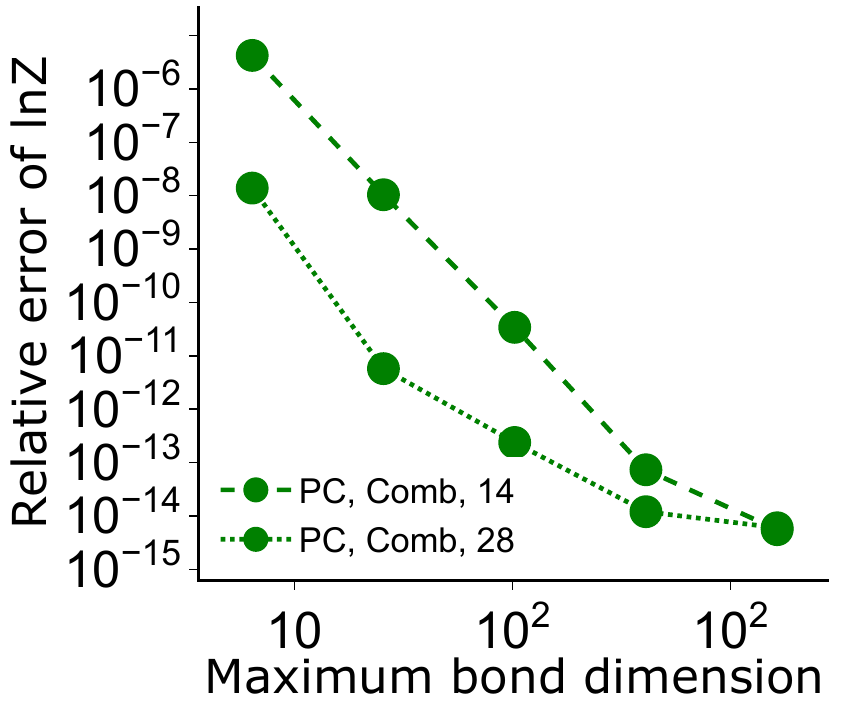}\label{subfig:bond_2d_1}}

\caption{Performance comparison between \texttt{partitioned\_contract}, SweepContractor~\cite{chubb2021general}, CATN~\cite{pan2020contracting}, and hyperoptimized approximate contraction (HAC)~\cite{gray2022hyper} in contracting lattices based on the Ising model. The swap batch size is fixed to be 32 for all experiments. In the legends, ``PC" denotes \texttt{partitioned\_contract}, MPS/Comb denotes the embedding tree ansatz, and the values (1, 3, 5) denote the size of each partition.
In (a)-(f), the number shown on top of each point is the maximum bond size $\chi$.
In CATN, ``Dmax" is an additional input parameter of the algorithm that controls the size of the MPS uncontracted modes. In HAC, $r$ denotes the distance of the spanning trees used in canonicalization.}
\label{fig:3d}
\end{figure}

 \begin{figure}[!ht]
\centering

\subfloat[]{
\includegraphics[width=.32\textwidth, keepaspectratio]
{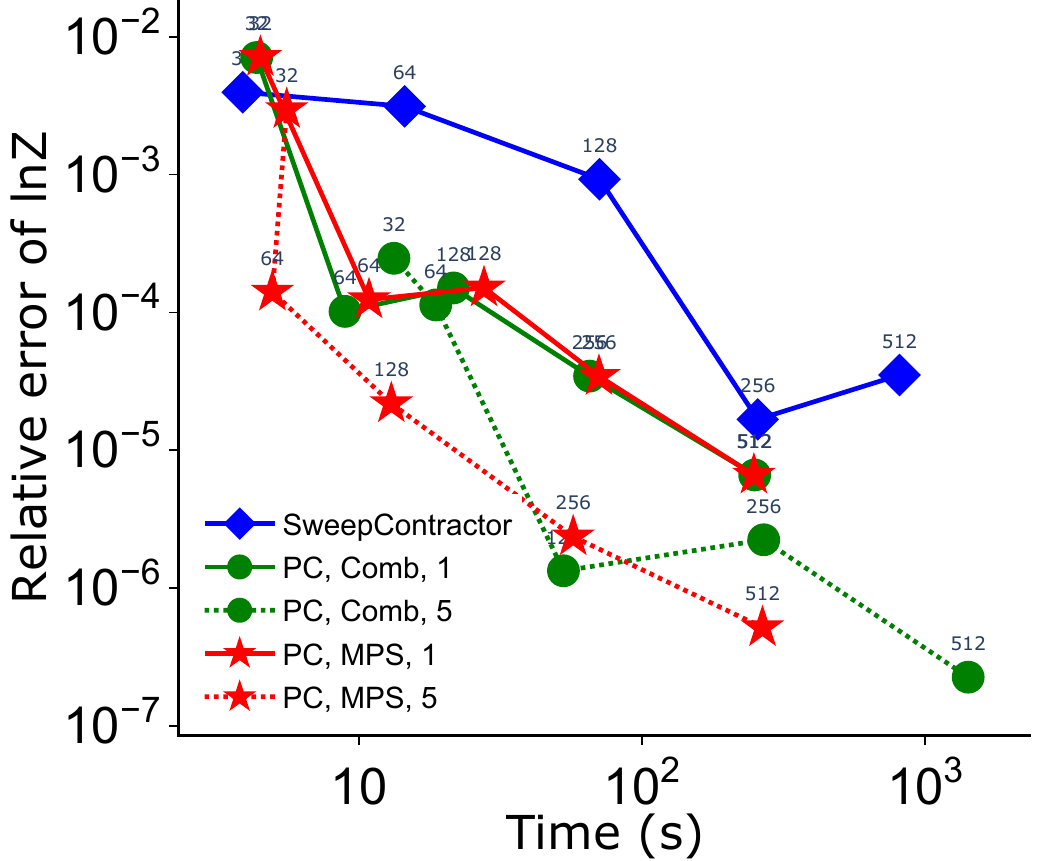}\label{subfig:3d_random_1}}
\subfloat[]{\includegraphics[width=.32\textwidth, keepaspectratio]{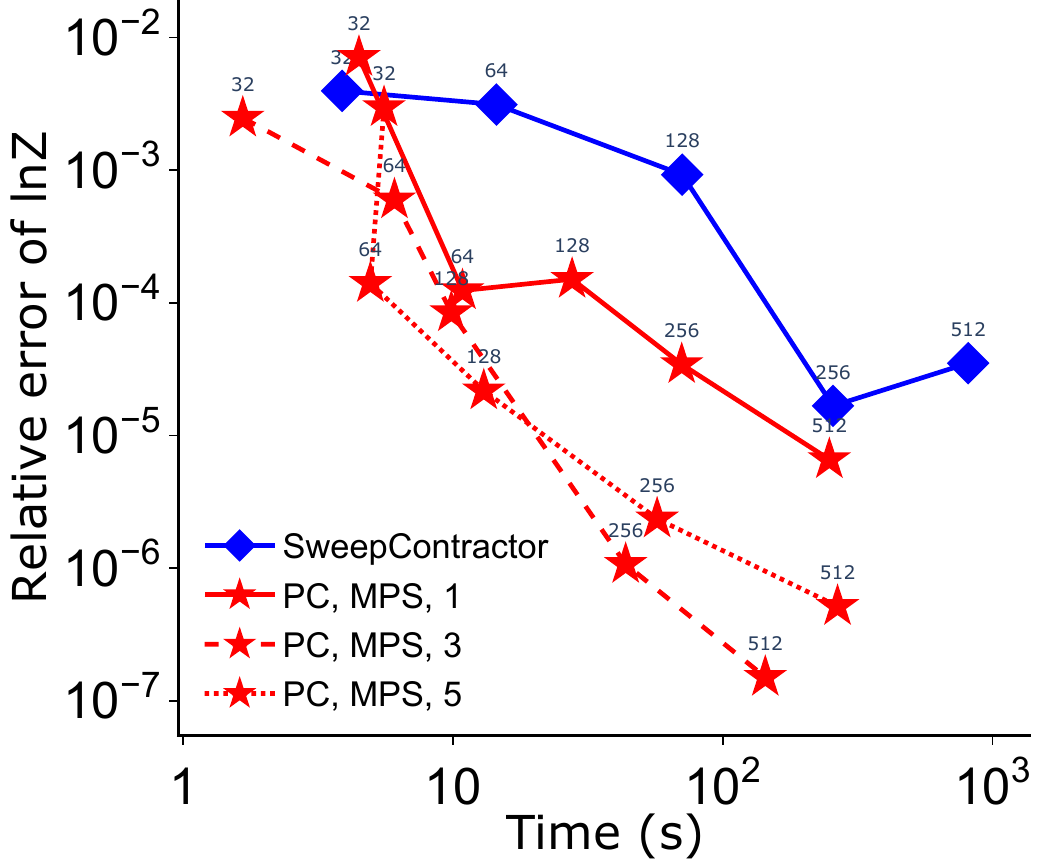}
\label{subfig:3d_random_2}}
\subfloat[]{
\includegraphics[width=.32\textwidth, keepaspectratio]
{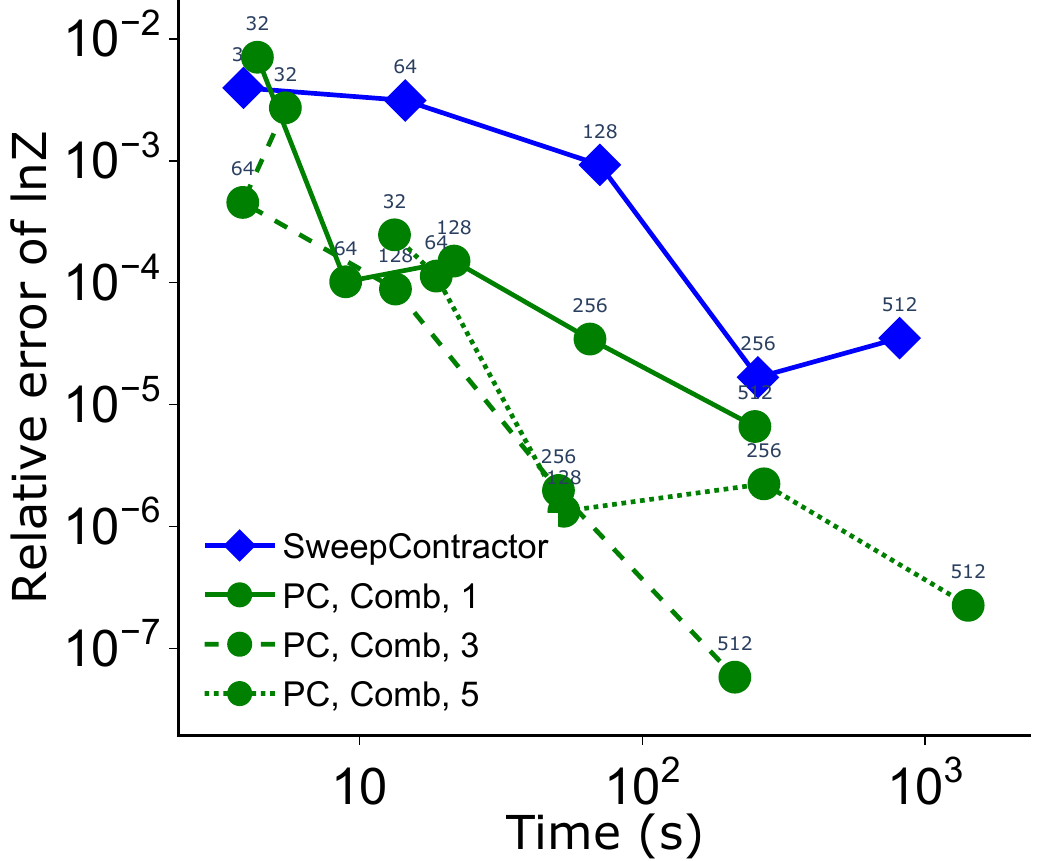}\label{subfig:3d_random_3}}

\subfloat[]{
\includegraphics[width=.24\textwidth, keepaspectratio]
{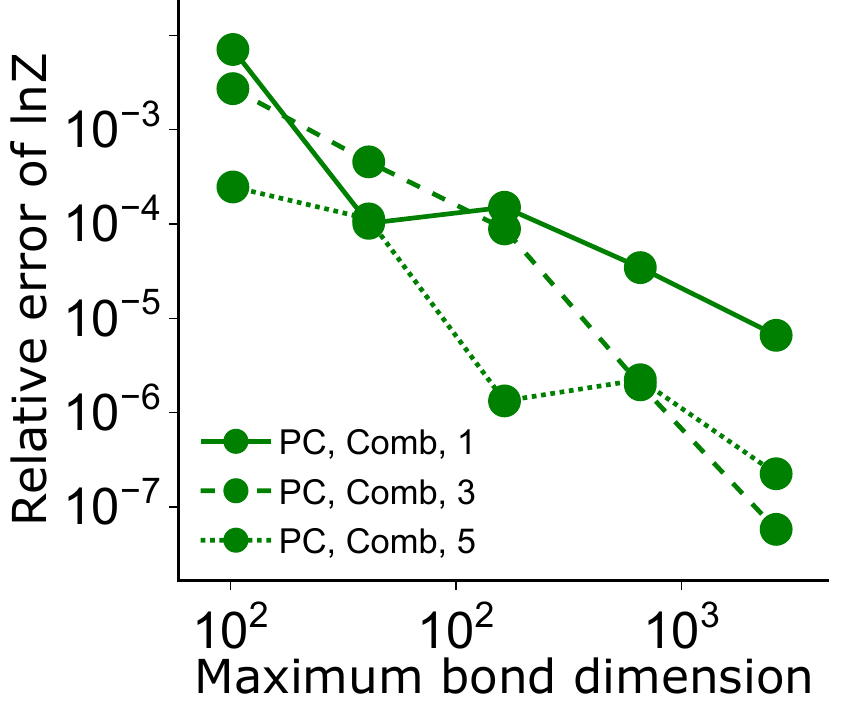}\label{subfig:bond3}}
\subfloat[]{
\includegraphics[width=.24\textwidth, keepaspectratio]
{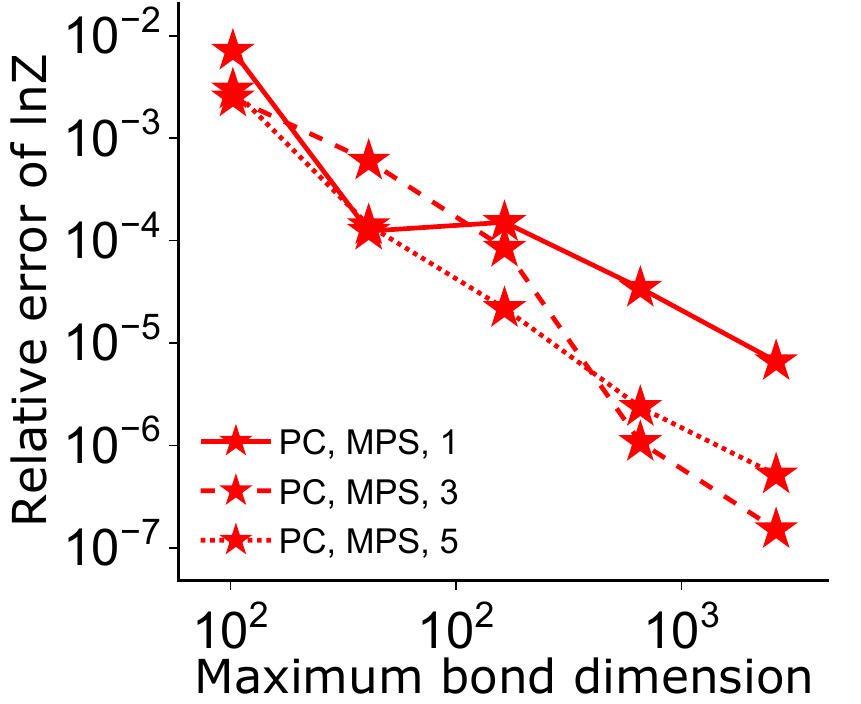}\label{subfig:bond4}}

\caption{Performance comparison between \texttt{partitioned\_contract} and SweepContractor~\cite{chubb2021general} in contracting $5 \times 5 \times 5$ random tensor networks defined on 3D lattices with $\alpha=-0.4$. The swap batch size is fixed to be 32 for all experiments. In the legends, ``PC" denotes \texttt{partitioned\_contract}, MPS/Comb denotes the embedding tree ansatz, and the values (1, 3, 5) denote the size of each partition.
In (a)-(c), the number shown on top of each point is the maximum bond size $\chi$.
}
\label{fig:3d_random}
\end{figure}

 \begin{figure}[!ht]
\centering
\subfloat[Ising Model, $\beta=0.65$]{
\includegraphics[width=.32\textwidth, keepaspectratio]
{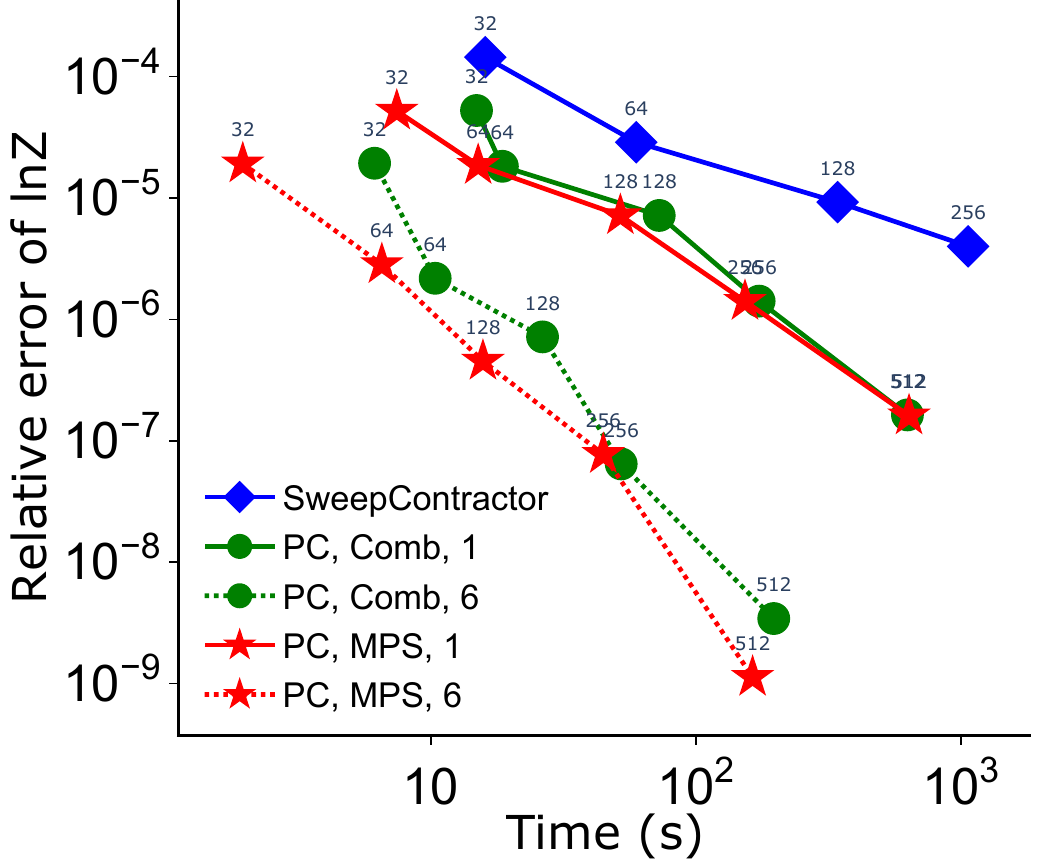}\label{subfig:randreg_ising_1}}
\subfloat[Ising Model, $\beta=0.65$]{\includegraphics[width=.32\textwidth, keepaspectratio]{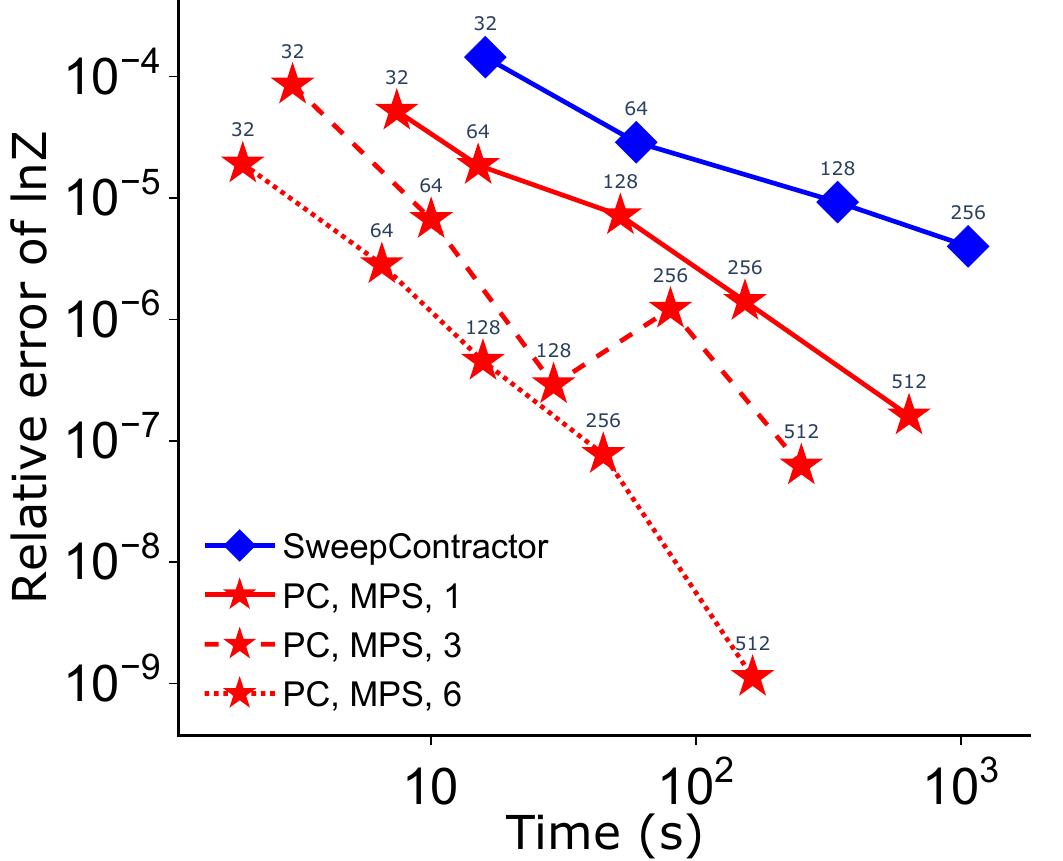}
\label{subfig:randreg_ising_2}}
\subfloat[Ising Model, $\beta=0.65$]{
\includegraphics[width=.32\textwidth, keepaspectratio]
{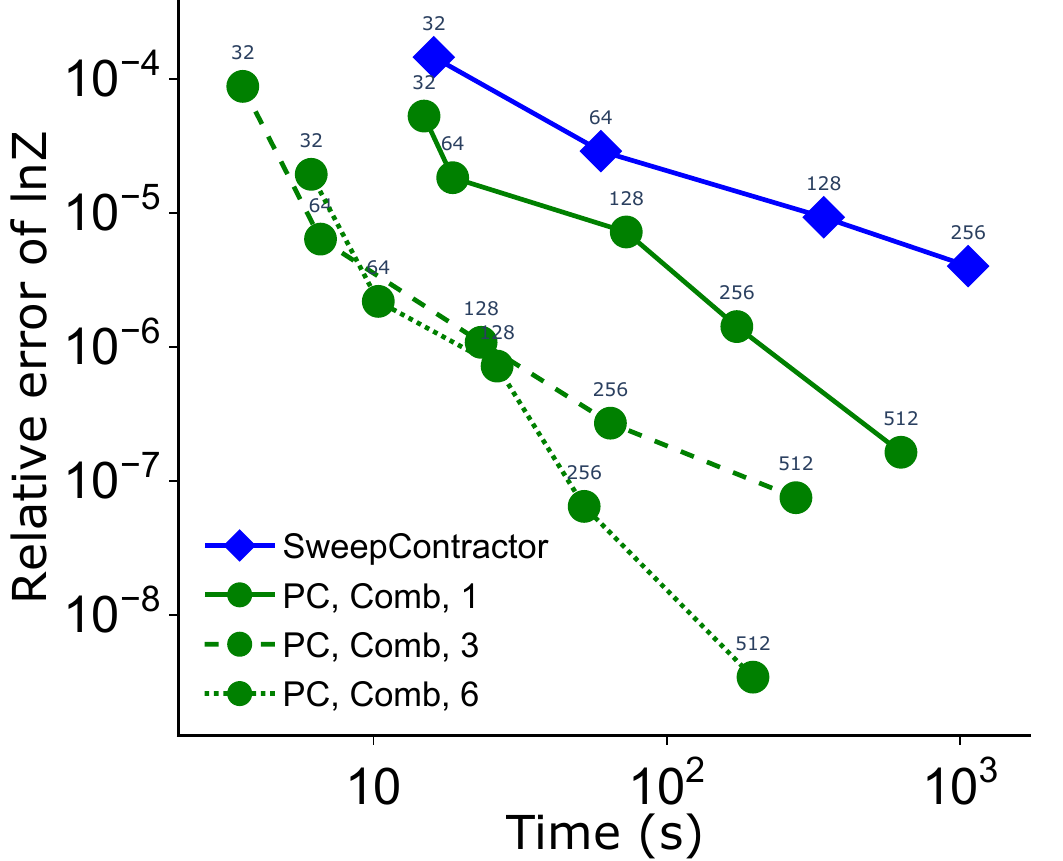}\label{subfig:randreg_ising_3}}

\subfloat[Random Model, $\alpha=-0.2$]{
\includegraphics[width=.32\textwidth, keepaspectratio]
{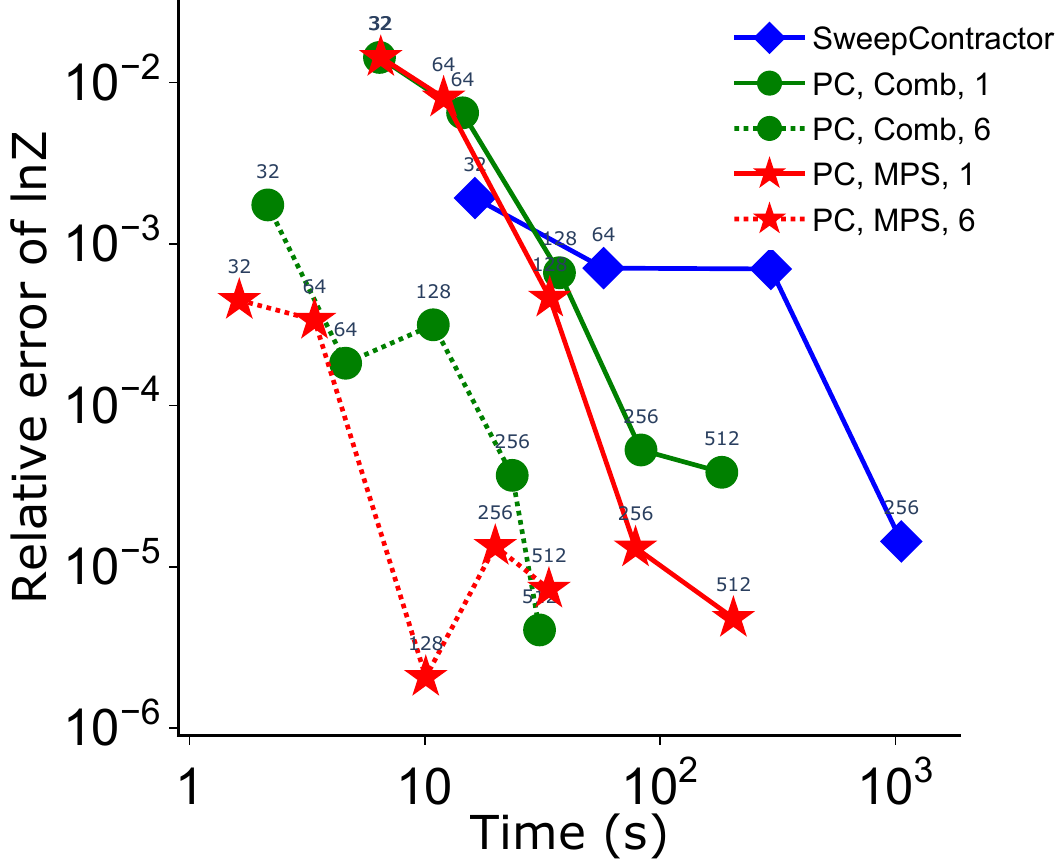}\label{subfig:randreg_random_1}}
\subfloat[Random Model, $\alpha=-0.2$]{\includegraphics[width=.32\textwidth, keepaspectratio]{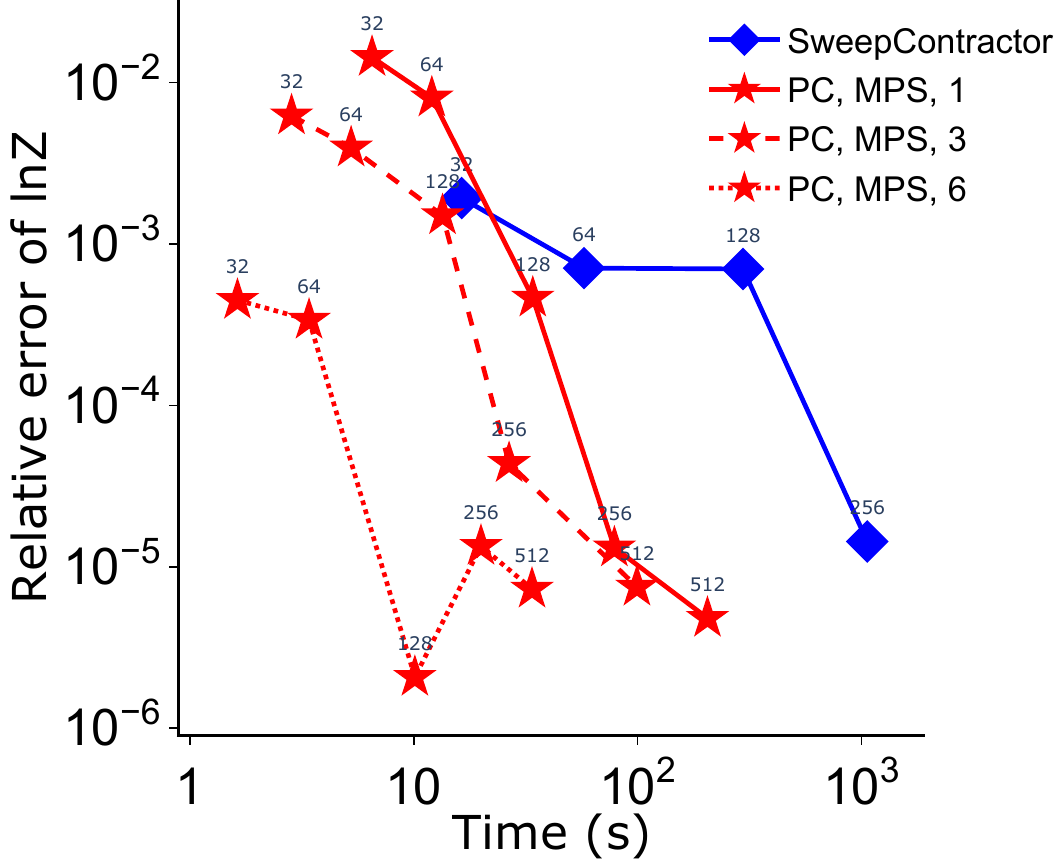}
\label{subfig:randreg_random_2}}
\subfloat[Random Model, $\alpha=-0.2$]{
\includegraphics[width=.32\textwidth, keepaspectratio]
{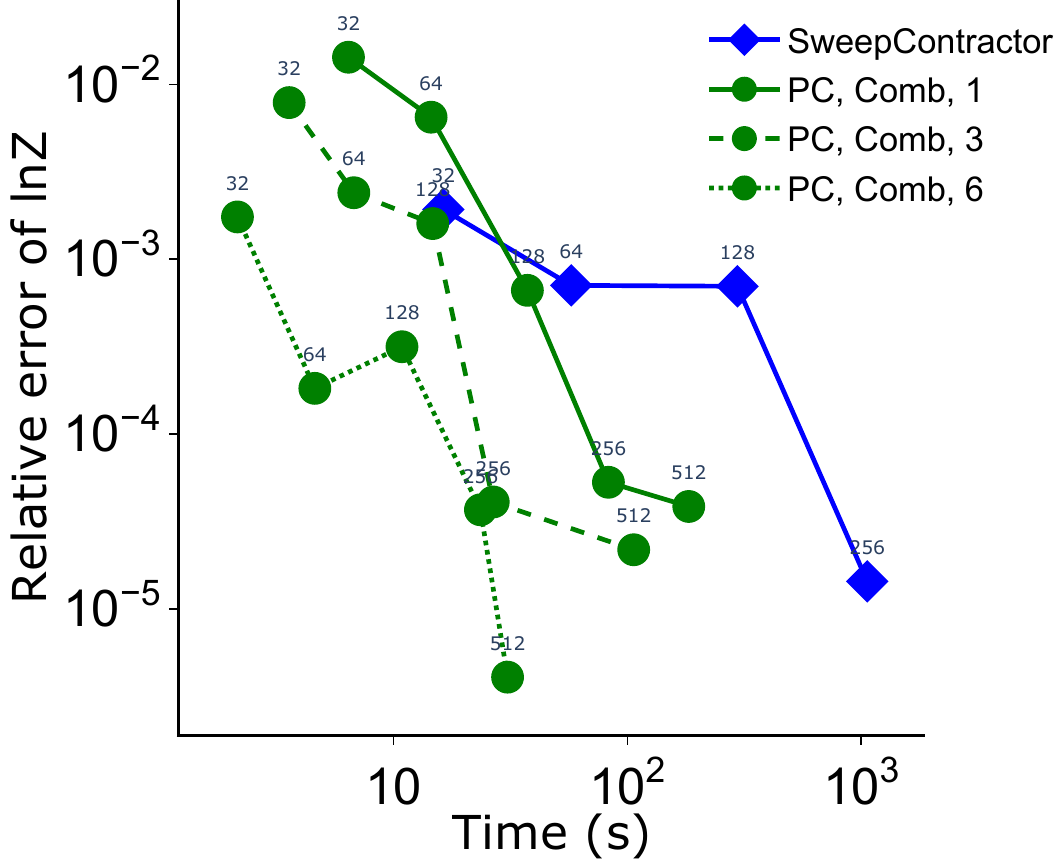}\label{subfig:randreg_random_3}}

\subfloat[Ising, $\beta=0.3$]{
\includegraphics[width=.24\textwidth, keepaspectratio]
{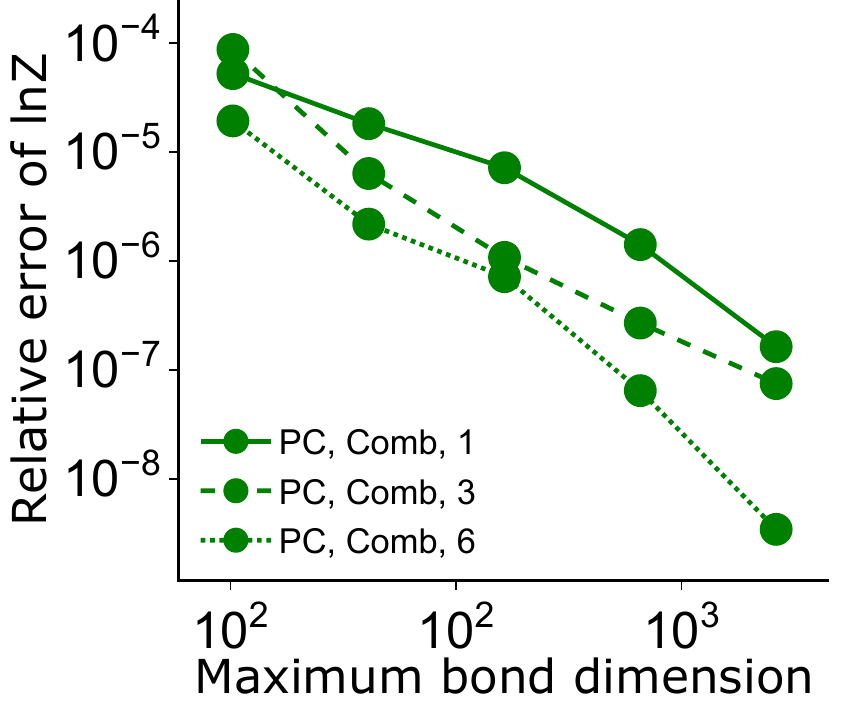}\label{subfig:bond5}}
\subfloat[Ising, $\beta=0.3$]{\includegraphics[width=.24\textwidth, keepaspectratio]{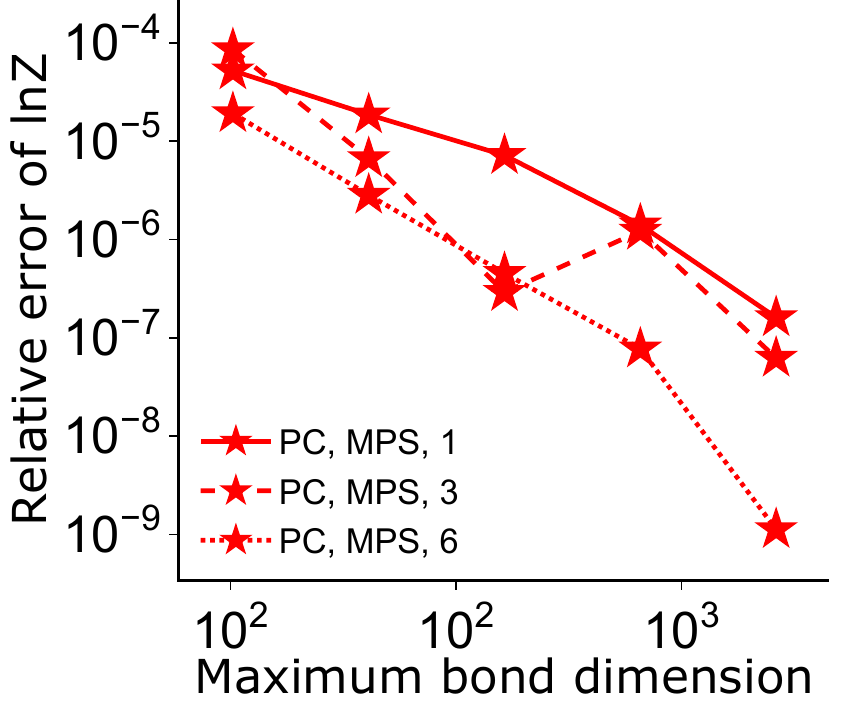}
\label{subfig:bond6}}
\subfloat[Random, $\alpha=-0.4$]{
\includegraphics[width=.24\textwidth, keepaspectratio]
{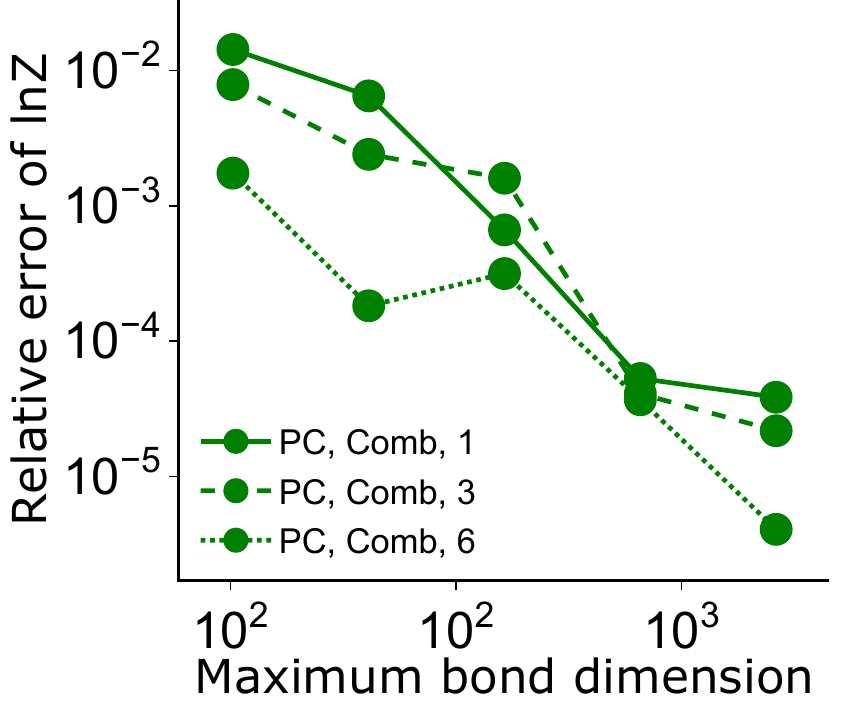}\label{subfig:bond7}}
\subfloat[Random, $\alpha=-0.4$]{
\includegraphics[width=.24\textwidth, keepaspectratio]
{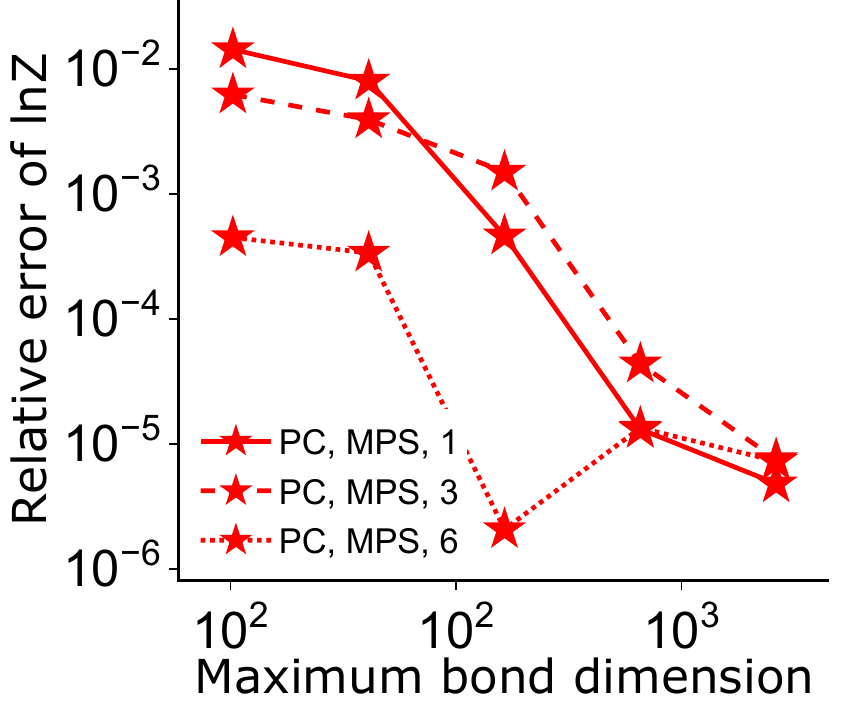}\label{subfig:bond8}}

\caption{Performance comparison between \texttt{partitioned\_contract} and SweepContractor~\cite{chubb2021general} in contracting random regular graphs with degree 3 and 220 vertices.
The swap batch size is fixed to be 32 for all experiments. In the legends, ``PC" denotes \texttt{partitioned\_contract}, MPS/Comb denotes the embedding tree ansatz, and the number (1, 3, 6) denotes the size of each partition.
In (a)-(f),
the number shown on top of each point is the maximum contracted bond size $\chi$.
}
\label{fig:randreg}
\end{figure}

\paragraph{Impact of the environment size on contraction accuracy and efficiency}\label{subsec:exp_environ}

We explore the impact of the environment size on the accuracy and efficiency of contracting tensor networks defined on 2D and 3D lattices as well as random regular graphs, and the results are shown in \cref{fig:3d}, \cref{fig:3d_random}, and \cref{fig:randreg}.

In both lattices and random regular graphs, we employ the maximally-unbalanced partial contraction path for the contraction process. This path initiates from one partition and progressively combines the previously-contracted section with a new partition following a linear sequence of the partitions.
For 3D lattices, each partition represents either a portion of a fiber or an entire fiber of the lattice. The contraction path is determined through a row- or column-wise traversal of the 2D array resulting from partitioning the 3D lattice into fibers.
Regarding random regular graphs, we draw inspiration from \cite{ibrahim2022constructing} to construct the contraction path using a linear ordering of vertices. We achieve this by first employing recursive bisection to generate the linear ordering of all the vertices. Then, we sequentially include a partition consisting of a specified number of tensors into the contraction path, following the order of traversal in the vertex ordering.

When contracting tensor networks defiend on lattices, the results presented in \cref{fig:3d,fig:3d_random} reveal that employing a larger parition size (3 or 5 for the $5 \times 5 \times 5$ 3D grid and 14 for the $28 \times 28$ 2D grid) leads to both faster and more accurate contractions when compared to the base condition where each partition contains only one tensor.
The improved efficiency arises from using larger partitions, which reduces the number of executions of the density matrix algorithm, offsetting any overhead from using larger environments. 
Regarding accuracy, we can see from \cref{subfig:bond2,subfig:bond1,subfig:bond_666_2,subfig:bond_666,subfig:bond_2d_2,subfig:bond_2d_1} that under the same maximum contracted bond size, utilizing a larger partition size generally yields lower relative errors for both MPS and comb structures. This observation validates the efficacy of the environment in enhancing accuracy.
See the next section for a comparison of the MPS and comb structures.

Regarding random regular graphs, the results displayed in \cref{subfig:randreg_ising_2,subfig:randreg_ising_3,subfig:randreg_random_2,subfig:randreg_random_3,subfig:bond5,subfig:bond6,subfig:bond7,subfig:bond8} indicate that using a partition size of 6 results in the best combination of efficiency and accuracy.
To summarize, employing a larger partition leads to a larger environment size, generally reducing the contraction error under the same rank. However, when it comes to efficiency, the optimal partition size depends on the specific problem. Factors such as the number of executions of the density matrix algorithm to be performed and the cost of forming the density matrix under different environment sizes need to be taken into consideration to determine the most suitable partition size.

\paragraph{Comparison between the MPS and the comb structure}

In this section we discuss the relative merits of using the MPS and comb tree ansatzes. When contracting lattices, the results in \cref{subfig:3d_ising_1,subfig:3d_ising_4,subfig:2d_ising_1,subfig:3d_random_1} demonstrate that both MPS and comb structures exhibit similar performance when the maximum bond size is small. However, as the maximum bond size increases, using the comb ansatz becomes slower in comparison to MPS.
On the other hand, when contracting random regular graphs, the results in \cref{subfig:randreg_ising_1,subfig:randreg_random_1} reveal that both structures display similar levels of accuracy and efficiency.
In summary, both MPS and comb binary tree structures perform similarly in terms of accuracy. However, the efficiency of the comb structure may lag behind MPS, particularly when dealing with a large rank. This disparity in performance is attributed to the presence of large tensors with size $\chi^3$ in the comb ansatz.

\paragraph{Comparison among \texttt{partitioned\_contract} and the baselines}

 \begin{figure}[!ht]
\centering
\subfloat[Tensor diagram of $|{\psi}\rangle$]{
\includegraphics[width=.25\textwidth, keepaspectratio]
{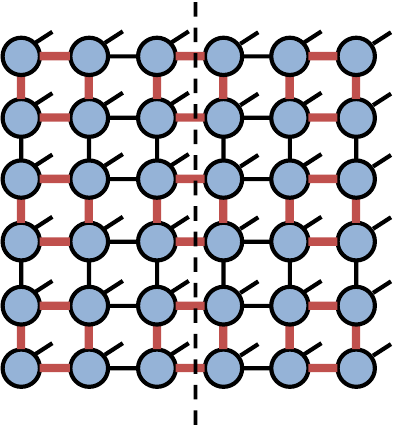}\label{subfig:circuit_1}}
\hspace{15mm}
\subfloat[Benchmark results]{\includegraphics[width=.43\textwidth, keepaspectratio]{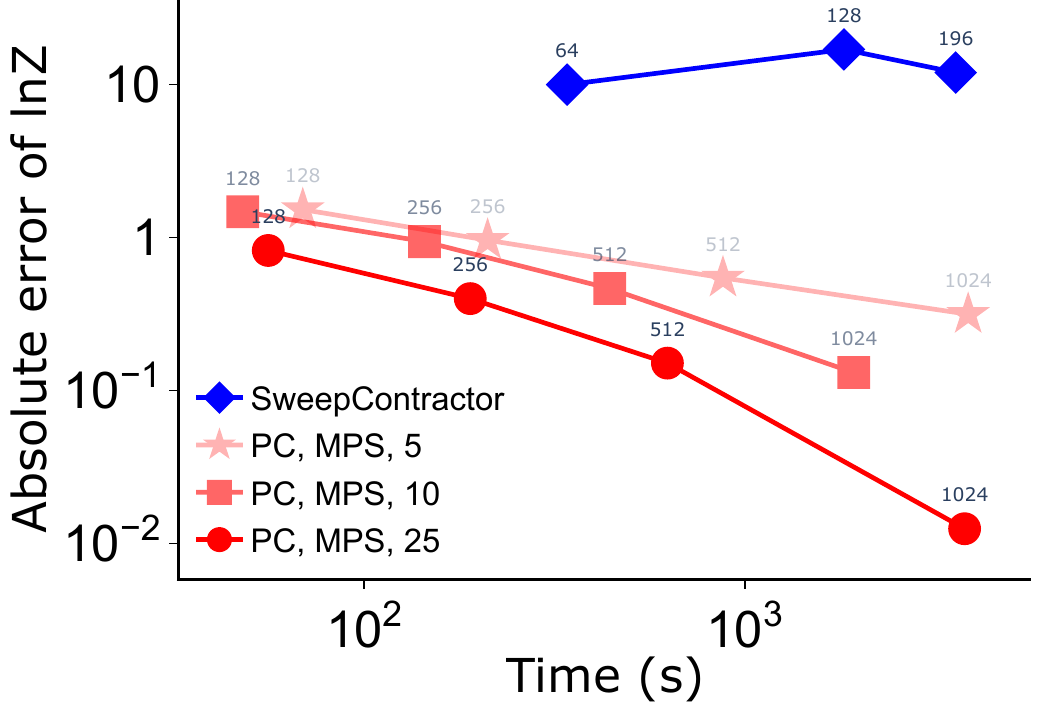}
\label{subfig:circuit_2}}

\caption{An illustration and benchmark results for approximately contracting a 2D random quantum circuit tensor network with 6 layers of gates. (a) The tensor diagram visualization of the PEPS $|{\psi}\rangle$ that is the output of the random quantum circuit simulation with 6 layers of gates. Each black edge denotes a mode with size 2, and each thick red edge denotes a mode with size 4. The dashed line denotes a graph cut with the cut mode size being $2^{12}$. (b) Performance comparison between \texttt{partitioned\_contract} and SweepContractor~\cite{chubb2021general} in contracting the tensor network that represents $\langle\psi|\psi\rangle$.
The swap batch size is fixed to be 8 for all experiments. In the legends, ``PC" denotes \texttt{partitioned\_contract}, and the number (5, 10, 25) denotes the size of each partition.
The number shown on top of each point is the maximum contracted bond size $\chi$.
}
\label{fig:circuit}
\end{figure}

Here we compare our proposed \texttt{partitioned\_contract} algorithm, along with the CATN algorithm~\cite{pan2020contracting}, SweepContractor~\cite{chubb2021general}, and hyper-optimizated approximate contraction (HAC)~\cite{gray2022hyper}, on contracting tensor networks defined on lattices and random regular graphs.

As demonstrated in \cref{subfig:3d_ising_1,subfig:3d_ising_4,subfig:2d_ising_1,subfig:3d_random_1,subfig:randreg_ising_1}, our algorithm consistently outperforms CATN in terms of efficiency across all levels of relative error. When compared to SweepContractor, our algorithm shows superior efficiency for 3D lattices and random graphs, while both methods perform similarly on 2D grids. This is expected, as SweepContractor is specifically optimized for planar graphs.
Notably, when contracting a 3D lattice tensor network based on the Ising model, \texttt{partitioned\_contract} achieves a 9.2X speed-up compared to both CATN and SweepContractor when reaching a relative error of less than $10^{-9}$. Similarly, when contracting a tensor network with a random regular graph structure based on the Ising model, \texttt{partitioned\_contract} achieves a 52.4X speed-up compared to SweepContractor when achieving a relative error of less than $10^{-5}$.

As shown from the results in  \cref{subfig:3d_ising_1,subfig:3d_ising_4,subfig:2d_ising_1}, HAC is a more efficient baseline compared to CATN and SweepContractor in contracting lattices based on the Ising model.
On both $5 \times 5 \times 5$ and $28 \times 28$ grids, both HAC and \texttt{partitioned\_contract} can reach pretty high accuracy with similar efficiency. On the $6 \times 6 \times 6$ grid, the results indicate that \texttt{partitioned\_contract} can reach higher accuracy and is more efficient when the relative error is smaller than $10^{-10}$. This can be attributed to HAC generating high-order intermediate tensors, which slows it down and exceeds memory capacity when tested with a rank of 256. In contrast, \texttt{partitioned\_contract} performs better due to its use of the MPS ansatz, making it more memory-efficient and effective at higher ranks.


In terms of memory usage, let $s$ denote the size of dimensions in the input graph, $\chi \geq s$ the maximum bond dimension, and $D_{\text{max}} \geq s$ an additional parameter controlling the size of uncontracted MPS modes in CATN. The largest intermediate tensor generated in \texttt{partitioned\_contract} has a size of $s\chi^2$ for the MPS ansatz and $\chi^3$ for the comb ansatz. 
For CATN, the largest tensor size is $D_{\text{max}}r^2$.
SweedContractor and \texttt{partitioned\_contract} with the MPS ansatz share a similar memory footprint of $s\chi^2$.
For HAC, the largest tensor size can be $\chi^d$ where $d$ is the maximum order of the intermediate tensors. To summarize, \texttt{partitioned\_contract} with the MPS ansatz has memory usage that is better than other algorithms as a function of bond dimension.

We also benchmark our algorithm against SweepContractor~\cite{chubb2021general} in contracting the random quantum circuit tensor network outlined in \cref{subsec:imple}.
The output state of the circuit can be represented as a PEPS that is visualized in \cref{subfig:circuit_1}. When accurately transforming $|{\psi}\rangle$ into an MPS, the maximum bond size will be at least $2^{12}=4096$, as inferred from the graph cut in \cref{subfig:circuit_1}. As illustrated in \cref{subfig:circuit_2}, our method manages to approximate the tensor network contraction with a bond size of $2^{10}=1024$, achieving an absolute error of $10^{-2}$. In contrast, SweepContractor fails to converge within the same amount of time.
These significant speed improvements clearly demonstrate the efficiency of our approach over the compared algorithms.

\section{Conclusion}
\label{sec:conclu}

In this work we introduced an efficient algorithm called \texttt{partitioned\_contract} to contract tensor networks with arbitrary structures. 
The algorithm has the flexibility to incorporate a large portion
of the environment when performing low-rank approximations, and includes a cost-efficient density matrix algorithm for approximating a general tensor network into a tree structure whose computational cost is asymptotically upper-bounded by that
of the standard algorithm that uses canonicalization.
Experimental results indicate that the proposed technique outperforms previously proposed approximate tensor network contraction algorithms for multiple problems 
in terms of both accuracy and efficiency.

We emphasize several potential future directions that need further
exploration and investigation.
Firstly, the partitioned\_contract algorithm assumes that both a partitioning of the input tensor network and a contraction path over these partitions are provided. There remains an opportunity to explore efficient methods for finding optimal partitionings and contraction paths for \texttt{partitioned\_contract}, which could further improve its performance.
Additionally, there is scope for investigating how the canonicalization-based algorithm for tree approximation can be accelerated. One possibility is to leverage tensor network sketching techniques~\cite{ma2022cost, ma2021fast, rakhshan2020tensorized, ahle2020oblivious} to speed up randomized SVD~\cite{halko2011finding}, which may enhance the efficiency of the tree approximation process. Another is to use variational or fitting algorithms for approximately contracting network partitions, which can have better scaling than density matrix and canonicalization-based algorithms at the expense of potentially requiring multiple iterations to converge~\cite{verstraete2004renormalization,stoudenmire2010minimally}.
Finally, integrating the proposed algorithm into automatic differentiation libraries \cite{ma2020autohoot, gray2018quimb, liao2019differentiable} could be highly beneficial. This integration would enable the algorithm to be used in gradient-based optimization algorithms for tensor networks, thereby expanding its utility in various optimization tasks.

\section*{Acknowledgments}

M.F. and M.S. are grateful for ongoing support through the Flatiron Institute, a division of the Simons Foundation. L.M. acknowledges support through the Flatiron Institute internship program through which this work was initiated. The work of L.M. and E.S. was also supported by the US National Science Foundation (NSF) grant \#1942995 and the Department of Energy (DOE)
Advanced Scientific Computing Research program via award
DE-SC0023483.

\bibliographystyle{plainnat}
\bibliography{main}

\appendix

\section{Notations}

\begin{table}[!ht]
  \begin{center}
    \renewcommand{\arraystretch}{1.6}
    {
    \begin{tabular}{l|p{10cm}}
\hline
      Notations  & Meanings \\ \hline
      $G =(V,E)$  & The input tensor network graph \\ \hline
      $\Set{V}=\{V_1,\ldots, V_N\}$ & A partitioning of $G$  \\ \hline
$T^{(\Set{V})}$  & A contraction tree over $\Set{V}$  \\ \hline
$G[U]$  & A subgraph of $G$ that only contains vertices $U$  \\ \hline

$\mathcal{E}=\left\{E(V_i,V_j): i,j\in \{1,\ldots,N\}\right\}$ & The set where each element in $\mathcal{E}$ is an edge subset connecting two different partitions \\ \hline
  $(U_s,W_s)$ & 
     The contraction $s$ where the first tensor is the contraction output of vertices $U_s$ and the second tensor is the contraction output of vertices $W_s$
         \\ \hline
  $\mathcal{E}_s=\left\{E'\cap E(U_s\cup W_s): E'\in \Set{E}\right\}$ & 
      The subset of $\mathcal{E}$ that is adjacent to the tensor network represented by $U_s\cup W_s$
         \\ \hline
    $\sigma^{(\Set{E}_s)}$ & A linear ordering of the elements in $\Set{E}_s$ \\ \hline
    $T=\pth\left(T^{(\mathcal{V})}, V_s\right)$ & The sub-contraction path over $V_s$\\ \hline
   $T^{(\mathcal{E}_s)}$ & Constraint tree detailed in \cref{subsec:adjacency} \\ \hline
    \end{tabular}
    }
    \renewcommand{\arraystretch}{1}
  \end{center}
\caption{Notation used throughout the paper.
}
\label{tab:notations}
\end{table}

\section{Additional background}

\subsection{The swap-based algorithm to reorder MPS modes}
\label{subsec:swap}
In the MPS-based automated tensor network contraction algorithms including CATN and SweepContractor,
an important step is to reorder the sites in an MPS. The reordering changes the adjacency relation in the MPS, and is used so that subsequent contractions can be performed with lower cost. The reordering is commonly performed via a sequence of adjacent site swappings. For a given MPS whose sites are denoted as a set $S$ and its input ordering is denoted as an injective mapping $\sigma: S \to \{1,\ldots, |S|\}$, changing it to a different ordering $\tau$ requires at least $\ktdist{\sigma,\tau}$ number of swaps, where $d_{\text{KT}}$ denotes the Kendall-Tau distance defined in \cref{defn:ktdist}.

\begin{definition}\label{defn:ktdist}
    Let $\sigma,\tau$ be two orderings over $S$.
The  Kendall Tau distance between $\sigma,\tau$ is the number of pairs that are ordered
differently in $\sigma,\tau$, and is also
the number of pairwise adjacent transpositions needed to transform
$\sigma $ into $\tau$ (or vise versa), 
\begin{equation}
\ktdist{\sigma,\tau} = 
\sum_{(c,c')\in S} 
\Big|\sigma(c,c')- \tau(c,c')\Big|,
\end{equation}
where $\sigma(c,c') := \mathbbm{1}\Big(\sigma(c) < \sigma(c')\Big)$ indicates if $c$ is ahead of $c'$ in $\sigma$. 
\end{definition}

We illustrate the standard algorithm to swap adjacent MPS sites via a contraction and a low-rank approximation in \cref{fig:swap}. The algorithm first contracts two sites into a single tensor and subsequently performs a low-rank approximation to split the tensor into two parts.
When the uncontracted modes have sizes $x$ and $y$, and the MPS ranks are $a, c$, and $b$, the contraction step has an asymptotic cost of $\Theta(abcxy)$, resulting in a tensor with a size of $abxy$. Without truncation, the output rank of the low-rank approximation operation would be the minimum among $ay, bx, cxy$. In practice, it is common to set an upper bound $\gamma$ for the MPS ranks, which limits the asymptotic cost of the approximation operation to $O(abxy\min(ay,bx,cxy,\gamma))$ when using the cost model in \cref{subsec:cost_model}.
To reduce the truncation error, canonicalization is commonly performed on the MPS to orthogonalize all other sites.

 \begin{figure}[!ht]
\centering
\includegraphics[width=1.\textwidth, keepaspectratio]{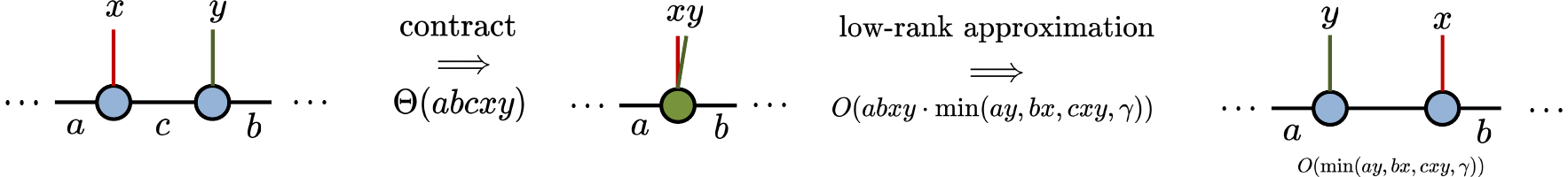}

\caption{
Illustration of the swap operation and the asymptotic computational cost. 
}
\label{fig:swap}
\end{figure}

\subsection{Background on embedding an source graph into a target graph}
\label{subsec:recur_bisection}
Our proposed algorithm uses heuristics from the graph embedding problem.
A graph embedding of a source graph $G_s=(V_s,E_s)$ into a target graph $G_t=(V_t,E_t)$ is a map from vertices of the input graph onto vertices of the output graph, $\phi: V_s\to V_t$, and each edge connecting $u,v$ of $G_s$ is mapped onto a path connecting $\phi(u),\phi(v)$ of $G_t$. 
For each edge $e\in E_t$, we let $\texttt{congestion}(e)$ denote the  number of times $e$ is used as a corresponding path of some edge in $G_s$.
We look at the problem of finding the graph embedding that minimizes the congestion~\cite{hruska2008tree,bienstock1990embedding,matsubayashi2015separator,bezrukov2000congestion,manuel2009exact}. This metric is used since when embedding a tensor network into another graph, low congestion implies that the embedded tensor network has low ranks as well as low memory usage.

For the case where $G_t$ is a line graph and $\phi$ is an injective mapping, finding $\phi$ that minimizes the congestion is the widely-discussed linear ordering problem. When the objective is to minimize $\max_{e\in E_t}\texttt{congestion}(e)$, the problem has been called the minimum cut linear arrangement problem, and the congestion is also called cutwidth in the previous work~\cite{thilikos2005cutwidth}. When the objective is to minimize $\sum_{e\in E_t}\texttt{congestion}(e)$, the problem has been called the minimum linear arrangement problem~\cite{harper1964optimal, diaz2002survey}, and multiple approximation algorithms with bounded complexity have been proposed~\cite{even2000divide,rao2005new,feige2007improved,charikar2010,devanur2006integrality}.

Recursive bisection is a simple yet effective divide-and-conquer heuristic widely adopted in both  linear ordering problems~\cite{dhulipala2016compressing,hansen1989approximation} and balanced graph partitioning~\cite{schlag2023high,simon1997good}.
For the linear ordering problem, the algorithm proceeds via first applying an approximate $1/3$-balanced cut to separate $V_s$ into two parts $S$ and $V_s\setminus S$, then placing all vertices of $S$ before all vertices not in $S$, and then recursing on both $S$ and $V_t\setminus S$. 
Let $n$ denote the number of vertices in the graph,
it is known that if
one has a $\gamma$-approximation algorithm for minimum $1/3$-balanced cut, then both the minimum cut linear arrangement and the minimum linear arrangement problem admit an approximation of $O(
\gamma \log n)$~\cite{hansen1989approximation,vazirani2001approximation}.
The approximation factor for the $1/3$-balanced cut is improved from $\gamma = O(\log n)$~\cite{leighton1989approximate} to $\gamma = O(\sqrt{\log n})$~\cite{arora2009expander}, making the approximation factor of the recursive bisection $O(\log^{1.5} n)$.

In \cref{sec:dm,sec:binary_tree}, we use recursive bisection as a heuristic for other embedding problems where $\phi$ is not necessarily injective, and $G_t$ is a general binary tree rather than a line graph.

\section{Embedding tree definitions}\label{sec:appendix_tree_def}

We formally define the embedding tree with an MPS and a comb structure in \cref{def:embed_mps} and \cref{def:embed_comb}. Both definitions are based on the MPS tree, which is defined in \cref{def:unbalanced}.

\begin{definition}[MPS tree]\label{def:unbalanced}
Consider a set $S$ with a linear ordering $\sigma^S$.
Let $x_i\in S$ denote the element with $\sigma^S(x_i) = i$.
The MPS tree defined on $\sigma^S$ is a full binary tree with the elements of $S$ serving as the tree's leaf nodes.
The MPS tree contains $|S|-1$ non-leaf nodes, where the first non-leaf node is connected to $x_1$ and $x_2$, and the $i$th non-leaf node for $i\in \{2,\ldots, |S|-1\}$ is connected to the $i-1$th non-leaf node and $x_{i+1}$. An example is shown in \cref{contract_subfig:mps_tree}.
\end{definition}

\begin{definition}[Embedding tree with an MPS structure]\label{def:embed_mps}
Consider orderings $\sigma^{(\mathcal{E}_s)}$
and $\sigma^{(E')}$ for $E'\in \mathcal{E}_s$. Let $n_s=\left|\mathcal{E}_s\right|$, and let $E_i$ denote the edgeset with $\sigma^{(\mathcal{E}_s)}(E_i) = i$.
The MPS embedding tree based on $\sigma^{(\mathcal{E}_s)}$, $\{\sigma^{(E')}, E'\in \mathcal{E}_s\}$ is the MPS tree defined on the ordering $\sigma^{(E_1)}\oplus \cdots \oplus \sigma^{(E_{n_s})}$, where we use $\sigma^{S_1}\oplus\sigma^{S_2}$ to denote the concatenation of two orderings $\sigma^{S_1}$ and $\sigma^{S_2}$, so that each $x\in S_1$ is mapped to $\sigma^{S_1}(x)$ and each $x\in S_2$ is mapped to $\sigma^{S_2}(x) + |S_1|$.
\end{definition}

\begin{definition}[Embedding tree with a comb structure]\label{def:embed_comb}
    Consider orderings $\sigma^{(\mathcal{E}_s)}$
and $\sigma^{(E')}$ for $E'\in \mathcal{E}_s$. Let $n_s=\left|\mathcal{E}_s\right|$, and let $E_i$ denote the edgeset with $\sigma^{(\mathcal{E}_s)}(E_i) = i$.
Let $T_i$ denote the MPS tree on top of $\sigma^{(E_i)}$ and let $r_i$ denote the root node of $T_i$. The comb embedding tree based on $\sigma^{(\mathcal{E}_s)}$, $\{\sigma^{(E')}, E'\in \mathcal{E}_s\}$ contains all $T_i$ for $i\in \{1,\ldots,n_s\}$ and another MPS tree $\hat{T}$ used to connect all $T_i$. The MPS tree $\hat{T}$ connects all $r_i$ and is defined on top of the ordering $\hat{\sigma}:\{r_1,\ldots, r_{n_s}\}\to \{1,\ldots, n_s\}$, where $\hat{\sigma}(r_i) = i$.
\end{definition}

\section{Determination of the constraint tree based on the contraction path}\label{subsec:adjacency}

The constraint tree $T^{(\mathcal{E}_s)}$ is constructed based on the sub-contraction path $T$. 
The tree is constructed bottom-up by connecting subsets of edges involved in the contraction path.
This construction is based on the assumption that ordering edges to make earlier rather than later contractions efficient is more important.

Specifically,
we let $U_1,\ldots, U_n$ be the $n$ partitions contracted with $V_s$ in order in the path $T$, let $\mathcal{E}$ be the edge partitions defined in Line~\ref{line:edge_partitions} of \cref{alg:ovw}, and let $\mathcal{E}(U_i) = \{\bar{E}\cap E(U_i): \bar{E}\in \mathcal{E}\}$ be the subset of $\mathcal{E}$ incident on $U_i$.
For each contraction with $U_i$, we use $\hat{\mathcal{E}}_i$ to denote the subset of $\mathcal{E}_s$ that we want to be connected in $T^{(\mathcal{E}_s)}$ based on the contraction.
In particular, $\hat{\mathcal{E}}_1 = \left(\mathcal{E}_s\cap \mathcal{E}(U_1) \right)$ contains all contracted edges $E(V_s,U_1)$. 
For each $i\in \{2, \ldots, n\}$, we want $\left(\mathcal{E}_s\cap \mathcal{E}(U_i) \right)$ along with some $\hat{\mathcal{E}}_j,j<i$ to be adjacent.
Formally speaking, for each $i\in \{1, \ldots, n\}$,
we define 
\begin{equation}
\hat{\mathcal{E}}_i = \left(\mathcal{E}_s\cap \mathcal{E}(U_i) \right) \bigcup_{j\in S_i} \hat{\mathcal{E}}_j,
\end{equation}
where $S_i\subseteq \{1,\ldots, i - 1\}$ is a subset of indices ahead of $i$ such that for each $j\in S_i$, $U_j$ is adjacent to $U_i$.
In \cref{fig:adjacency_tree_example}, we use an example to illustrate the constraint tree construction algorithm, and each $\hat{\mathcal{E}}_i$ is also shown in the figure.

 \begin{figure}[htb]
\centering
\includegraphics[width=.98\textwidth, keepaspectratio]{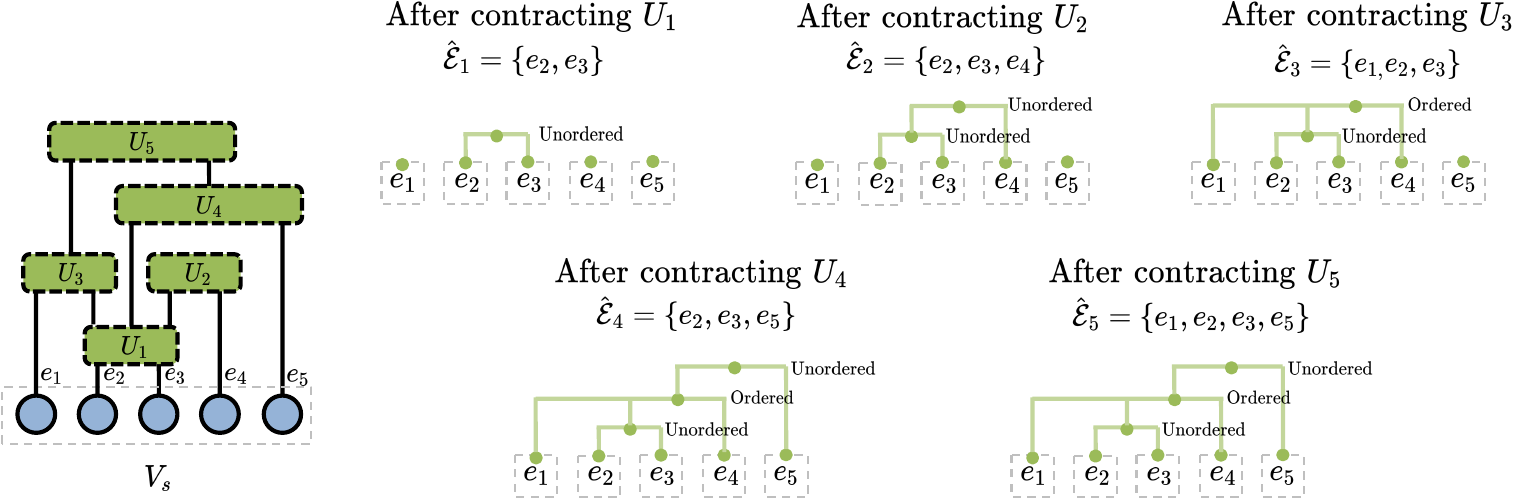}
\caption{
Illustration of the algorithm to construct the constraint tree. The constraint tree is built on top of the uncontracted edgesets of $V_s$, $\Set{E}_s=\{\{e_1\}, \{e_2\}, \{e_3\}, \{e_4\}, \{e_5\}\}$. The partitions $U_1,\ldots,U_5$ are contracted with $V_s$ in order. For the $i$th contraction, we show the value of $\hat{\Set{E}}_i$ and show the constraint tree after that contraction step.
}
\label{fig:adjacency_tree_example}
\end{figure}

In the algorithm, $T^{(\mathcal{E}_s)}$ is initialized to be a disconnected graph with vertices $\mathcal{E}_s$. 
For the $i$th contraction that contracts $U_i$, the algorithm updates the $T^{(\mathcal{E}_s)}$ so that the leaves $\hat{\mathcal{E}}_i$ will be connected.
The rules are as follows.
\begin{enumerate}
    \item If  $\hat{\mathcal{E}}_i$ are already connected in $T^{(\mathcal{E}_s)}$, we just keep the constraint tree unchanged. For example, in \cref{fig:adjacency_tree_example} the constraint tree is unchanged after we consider the fifth contraction, since $\hat{\mathcal{E}}_5$ is already connected. 
    \item If  $\hat{\mathcal{E}}_i$ is the union of multiple connected leaf subsets, then a vertex is added to $T^{(\mathcal{E}_s)}$ whose children are the root vertices of these connected leaf subsets. In addition, this new vertex is labeled as ``unordered". In \cref{fig:adjacency_tree_example}, the constraint trees after both the first and the second contraction belong to this case.
    \item If  $\hat{\mathcal{E}}_i$ is a subset of the union of multiple connected leaves subsets $\bar{\mathcal{E}}$, then there are cases where $\hat{\mathcal{E}}_i$ cannot be adjacent in the tree. For this case, a vertex is added to $T^{(\mathcal{E}_s)}$ whose children are the root vertices of $\bar{\mathcal{E}}$ and the vertex is labeled as ``unordered". In \cref{fig:adjacency_tree_example}, the constraint trees after the fourth contraction belongs to this case.
    For the other cases, we can reorder the constraint tree and label some vertices as ``ordered" to add the adjacency constraints.
    In \cref{fig:adjacency_tree_example}, the constraint trees after the third contraction belongs to this case.
\end{enumerate}

\section{Optimality of the edge set ordering algorithm}\label{sec:order_optimality}

In this section, we prove that the output ordering of \cref{alg:kdt} minimizes the Kendall-Tau distance with the reference ordering under the adjacency constraint.

Let $\mathcal{P}\left(T^{(\mathcal{E}_s)}\right)$ represents the set of orderings 
of the leaves of $T^{(\mathcal{E}_s)}$
constrained by $T^{(\mathcal{E}_s)}$. Each ordering in this set must adhere to all the adjacency relations specified by $T^{(\mathcal{E}_s)}$. In \cref{thm:ordering_under_constraint}, we establish that the output ordering produced by \cref{alg:min_kdt} aims to minimize the Kendall-Tau distance, as defined in \cref{defn:ktdist}, between itself and the reference ordering $\tau$,
\begin{equation}\label{eq:min_ktdist}
    \sigma^{(\mathcal{E}_s)} = \arg\min_{\sigma\in \mathcal{P}\left(T^{(\mathcal{E}_s)}\right)}
\ktdist{\sigma, \tau}.
\end{equation}
Before the presentation of \cref{thm:ordering_under_constraint}, we first present \cref{lem:ktdist_1} that is used in the proof of the theorem. The lemma can be easily proved based on the definition of Kendall-Tau distance in \cref{defn:ktdist}.

\begin{lemma}\label{lem:ktdist_1}
    Consider an ordering $\tau^{(C)}$ over a set  $C = C_1 \cup C_2$, and let $\tau^{(C_1)}$, $\tau^{(C_2)}$ denote the restrictions of the ordering $\tau^{(C)}$ to the subset
    $C_1$, $C_2$, respectively.
    Consider another two orderings $\sigma^{(C_1)}$, $\sigma^{(C_2)}$ over $C_1$, $C_2$, respectively. Then, we have
    \begin{equation}
    \ktdist{\tau^{(C)}, \sigma^{(C_1)}\oplus\sigma^{(C_2)}}
    = \ktdist{\tau^{(C)}, \tau^{(C_1)}\oplus\tau^{(C_2)}}  +
    \ktdist{\tau^{(C_1)}, \sigma^{(C_1)}} 
    + 
    \ktdist{\tau^{(C_2)}, \sigma^{(C_2)}},
    \end{equation}
    where $\tau^{(C_1)}\oplus\tau^{(C_2)}$ denotes the concatenation of $\tau^{(C_1)}$, $\tau^{(C_2)}$.
\end{lemma}

\begin{theorem}\label{thm:ordering_under_constraint}
Given a reference ordering $\tau$ and a guide tree $T^{(\mathcal{E}_s)}$,
the output ordering of \cref{alg:min_kdt} is an optimal solution of the optimization problem, $\min_{\sigma\in \mathcal{P}\left(T^{(\mathcal{E}_s)}\right)}\ktdist{\sigma, \tau}$.
\end{theorem}
\begin{proof}
For each vertex $v$ in the constraint tree $T^{(\mathcal{E}_s)}$, we let $\texttt{subtree}\left(v, T^{(\mathcal{E}_s)}\right)$ denote the subtree in the constraint tree where the root vertex is $v$. 
In addition, as is defined in Line~\ref{line:tauv} of \cref{alg:min_kdt}, we use $\tau_v$ to denote the restriction of the ordering  $\tau$ to the subset represented by the leaves of $\texttt{subtree}\left(v, T^{(\mathcal{E}_s)}\right)$.
Below we prove that for each $v\in T^{(\mathcal{E}_s)}$, 
\begin{equation}\label{eq:ordering_under_constraint_2}f(v) = \arg\min_{\sigma\in \mathcal{P}\left(\texttt{subtree}\left(v, T^{(\mathcal{E}_s)}\right)\right)}\ktdist{\sigma,\tau_v}, 
\end{equation} 
where $f(v)$ is defined in Line \ref{line:min_kdt} of \cref{alg:min_kdt}. 
Since we output $f\left(r\right)$ with $r = \texttt{root}\left(T^{(\mathcal{E}_s)}\right)$, and 
$\texttt{subtree}\left(r, T^{(\mathcal{E}_s)}\right) = T^{(\mathcal{E}_s)}$, the output ordering satisfies $f(r) = \arg\min_{\sigma\in \mathcal{P}\left(T^{(\mathcal{E}_s)}\right)}\ktdist{\sigma, \tau}$. This finishes the proof.

For the base cases where $v$ is one of the leaf vertices, \eqref{eq:ordering_under_constraint_2} holds since the set to be ordered only contains one element thus the ordering is unique. 

Now consider the case where $v$ is a non-leaf vertex. In the analysis we assume 
$v$ has two children, $u_1$ and $u_2$. Note that the analysis can be easily generalized to the case of more than 2 children for both ``unordered" and ``ordered" labels. 

Assume \eqref{eq:ordering_under_constraint_2} holds for its children, $u_1$ and $u_2$.
Consider the case where 
\begin{equation}\label{eq:ordering_under_constraint_less}
    \ktdist{f(u_1)\oplus f(u_2), \tau_v} < \ktdist{f(u_2)\oplus f(u_1), \tau_v},
\end{equation}
so that Line \ref{line:min_kdt} sets $f(v)$ as $f(u_1)\oplus f(u_2)$. We then have 
\begin{equation}\label{eq:ordering_under_constraint_3}
\begin{split}
    \ktdist{f(v), \tau_v} &= 
    \ktdist{f(u_1)\oplus f(u_2), \tau_v} \\
&\overset{\cref{lem:ktdist_1}}{=} \ktdist{\tau_v, \tau_{u_1} \oplus \tau_{u_2}} + 
 \ktdist{f(u_1), \tau_{u_1}} + 
 \ktdist{f(u_2), \tau_{u_2}}. 
 \end{split}
\end{equation}
The first term in \eqref{eq:ordering_under_constraint_3} reaches the minimum since \eqref{eq:ordering_under_constraint_less} holds. Moreover, the last two terms also reach the minimum since \eqref{eq:ordering_under_constraint_2} holds for $u_1$ and $u_2$. These conditions imply $f(v)$ satisfies
\eqref{eq:ordering_under_constraint_2}. Similar analysis can be applied for the case where $\ktdist{f(u_1)\oplus f(u_2), \tau_v} > \ktdist{f(u_2)\oplus f(u_1), \tau_v}$. 
This along with the based cases finish the proof.

\end{proof}

\section{Lemmas in computational cost analysis}\label{sec:cost_analysis_appendix}

This section provides lemmas for \cref{thm:upper_bound}, which shows that the asymptotic cost of the density matrix algorithm is upper-bounded by
that of the canonicalization-based algorithm.
We demonstrate in \cref{lem:same_cost_tree_approximation} that when the input tensor network has a tree structure, both the density matrix algorithm and the canonicalization-based algorithm exhibit the same asymptotic cost for truncating the bond sizes in the tree tensor network.

The \cref{lem:same_cost} and \cref{lem:step_dm} below are used to prove \cref{lem:same_cost_tree_approximation}.

\begin{lemma}\label{lem:same_cost}
    Consider a tensor network with a tree structure $T=(V_T,E_T)$.
    Assuming that changing a tree tensor network into the canonical form will not change any bond size of the network.
    For two adjacent vertices $z,v$, forming $\texttt{canonical\_form}_{T}(v, z)$ has the same asymptotic cost as forming $\texttt{density\_matrix}_{T}(v, z)$.
\end{lemma}
\begin{proof}
    For each edge set $E'\subseteq E_T$, we let $s(E') = \exp(w(E'))$ denote the bond size of $E'$.
    We also let $\mat{M}_v$ denote the tensor at each vertex $v\in V_T$.

    For the pair of adjacent vertices $v, z$, assume that
$\texttt{canonical\_form}_{T}(u, v)$ already exist for all $u\in N(v)\setminus \{z\}$. Let $\mat{R}_u$ denote the non-orthogonal core of $\texttt{canonical\_form}_{T}(u, v)$.
To construct the form $\texttt{canonical\_form}_{T}(v, z)$,
we first contract $\mat{M}_v$ with $\mat{R}_u$ for each $u\in N(v)\setminus \{z\}$, which yields a cost of $\Theta\left(\sum_{u\in N(v)\setminus \{z\}}s(E_T(v))s(E_T(u,v))\right)$, and then use a QR decomposition to orthogonalize the tensor at $v$, which yields a cost of $\Theta\left(s(E_T(v))s(E_T(v, z))\right)$.
These steps make the overall cost
\begin{equation}
\Theta\left(\sum_{u\in N(v)}s(E_T(v))s(E_T(u,v))\right).
\end{equation}

We now consider the computation of  $\texttt{density\_matrix}_{T}(v, z)$ under the assumption that for all $u\in N(v)\setminus \{z\}$,
$\mat{L}_u=\texttt{density\_matrix}_{T}(u, v)$ already exist. Below we consider the three different cases,
\begin{itemize}
    \item when $N(v)\setminus \{z\} = \emptyset$, the computation involves the contraction $\mat{M}_v\mat{M}_v^T$,
    \item when $N(v)\setminus \{z\} = \{u\}$, the computation involves the contraction $(\mat{M}_v\mat{L}_u)\mat{M}_v^T$,
    \item when $N(v)\setminus \{z\} = \{u_1,u_2\}$, the computation involves the contraction $\mat{M}_v(\mat{L}_{u_1}\otimes \mat{L}_{u_2})\mat{M}_v^T$, which can be efficiently computed by performing the contractions $\mat{M}_v$ with $\mat{L}_{u_1}$ and $\mat{M}_v$ with $\mat{L}_{u_2}$ first, and then contracting the outputs.
\end{itemize}
For all the cases above, the overall cost  is $\Theta\left(\sum_{u\in N(v)}s(E_T(v))s(E_T(u,v))\right)$, which equals the cost of the canonical form.
Since both $\texttt{canonical\_form}_{T}(v, z)$ and $\texttt{density\_matrix}_{T}(v, z)$ have the same recursive relation, computing $\texttt{canonical\_form}_{T}(u, v)$ has the same cost as  that of
 the $\texttt{density\_matrix}_{T}(u, v)$ for $u\in N(v)\setminus \{z\}$. This finishes the proof.

\end{proof}

\begin{lemma}\label{lem:step_dm}
Consider a tensor network with a tree structure $T=(V_T,E_T)$, where each $z\in V_T$ represents a tensor $\mat{M}_z$. Let $v\in V_T$ be a leaf vertex that represents $\mat{M}_v\in \R^{a_v\times b_v}$, where
 $a_v$ denotes the size of the uncontracted modes and $b_v$ denotes the size of the contracted modes incident on $v$, and let $u= \texttt{parent}(T,v)$. Given that $\texttt{density\_matrix}_T(u,v)$ has been computed, computing the orthogonal matrix $\mat{U}_v$ (Line~\ref{line:end_dm} or \ref{line:end_qr_svd} of \cref{alg:dm}) has a cost of $\Theta\left(a_vb_v^2\right)$.
\end{lemma}
\begin{proof}
For the case where $a_v=O\left(b_v\right)$, the algorithm first computes $\mat{L}_v = \texttt{density\_matrix}_T(v)$ with a cost of $\Theta\left(a_vb_v^2 + a_v^2b_v\right)$, and then
computes $\mat{U}_v$ via a low-rank factorization on $\mat{L}_v\in \R^{a_v\times a_v}$ with the maximum rank being $r = O\left(a_v\right)$, which costs $\Theta\left(a_v^2r\right)$. The overall cost is
$
\Theta\left(a_vb_v^2 + a_v^2b_v + a_v^2r\right) = \Theta\left(a_vb_v^2\right).
$

For the case where $a_v = \Omega\left(b_v\right)$, the algorithm first performs a QR decomposition of $\mat{M}_v$ into $\mat{U}_v\in \R^{a_v\times b_v},\mat{R}_v\in \R^{b_v\times b_v}$ with a cost of $\Theta\left(a_vb_v^2\right)$, then computes the leading singular vectors of $\mat{R}_v\mat{L}_u$ that is denoted $\hat{\mat{U}}_v\in \R^{b_v\times r}$, which costs $\Theta\left(b_v^3\right)$. Finally, $\mat{U}_v$ is updated as the product $\mat{U}_v\hat{\mat{U}}_v$ with a cost of $\Theta(a_vb_vr)$. Overall the cost is $\Theta\left(a_vb_v^2+b_v^3+a_vb_vr\right)=\Theta\left(a_vb_v^2\right)$. This finishes the proof.
\end{proof}

 \begin{figure}[!ht]
\centering
\includegraphics[width=.9\textwidth, keepaspectratio]{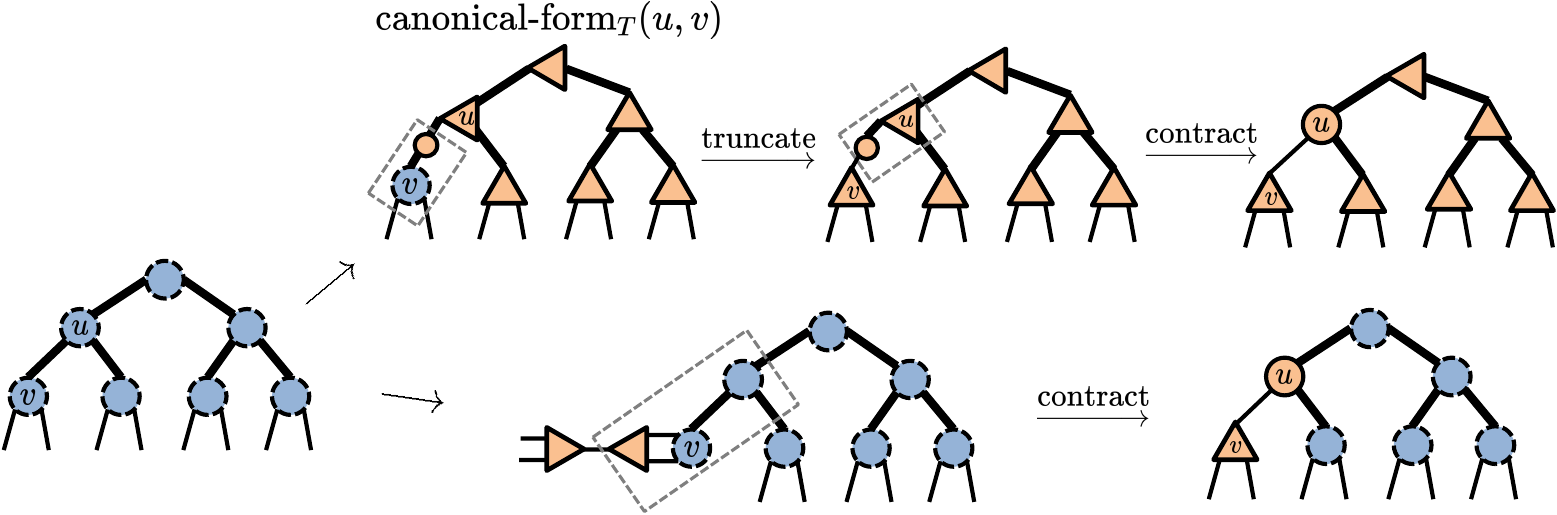}

\caption{
Illustration of the difference between the canonicalization-based algorithm and the density matrix algorithm. The upper path denotes truncating the edge $(u,v)$ using canonicalization, and the lower path uses the density matrix algorithm. In the lower path, the orthogonal matrix is calculated as the leading singular vectors/eigenvectors of the density matrix $\texttt{density\_matrix}_T(z)$.
}
\label{fig:compare}
\end{figure}

\begin{lemma}\label{lem:same_cost_tree_approximation}
Consider a given tree tensor network $T=(V_T,E_T)$.
Let $\sigma : V_T\to \{1,\ldots, |V_T|\}$ be a post-order DFS traversal of $T$ that shows the the tensor update ordering.
Assuming that changing a tree tensor network into its canonical form will not change any bond size of the network, the asymptotic cost of the density matrix algorithm (\cref{alg:dm}) for truncating the modes in $T$ is the same as that of the canonicalization-based algorithm (\cref{alg:canonize})
if both algorithms use the same update ordering $\sigma$, and the same maximum bond size $\chi$.
\end{lemma}
\begin{proof}
Consider the step to update the tensor at a given vertex $v\in V_T$.
Let $\mat{M}_v\in \R^{a_v\times b_v}$, where
 $a_v$ denote the size of the uncontracted modes and $b_v$ denote the size of the contracted modes of $\mat{M}_v$.
 Also let $r = \min\left(a_v,b_v,\chi\right)$ and $u= \texttt{parent}(T,v)$.
We break down the cost of \cref{alg:dm} and \cref{alg:canonize} into 3 parts, and show that for each of the three parts, the costs of the two algorithms are asymptotically equal.

In \cref{alg:canonize}, the steps include 1) forming $\texttt{canonical\_form}_T(u,v)$, 2)
multiplying $\mat{M}_v$ with $\mat{R}_u\in \R^{b_v\times b_v}$, the non-orthogonal core of the canonical form, and 3) performing a rank-$\chi$ approximation to get $\mat{U}_v\in \R^{a_v\times r},\hat{\mat{R}}_u\in \R^{r\times b_v}$, and 3) multiplying $\hat{\mat{R}}_u$ with $\mat{M}_u$.

In \cref{alg:dm} with each partition contracted into a tensor, the steps include 1) forming the density matrix $\texttt{density\_matrix}_T(u,v)$, 2) using $\texttt{density\_matrix}_T(u,v)$ and $\mat{M}_v$ to compute $\mat{U}_v\in \R^{a_v\times r}$  and  $\mat{M}_v = \mat{U}_v^T\mat{M}_v$, and 3) multiplying $\mat{M}_v\in \R^{r\times b_v}$ with $\mat{M}_u$.

The comparison between the two algorithms is visualized in \cref{fig:compare}.
It can be seen that
the third step of both algorithms have the same asymptotic cost. For the first step, we show in  \cref{lem:same_cost} that  both algorithms have the same asymptotic cost.
For the second step,
the canonicalization-based algorithm yields a cost of $\Theta\left(a_vb_v^2 + a_vb_vr\right) = \Theta\left(a_vb_v^2\right)$ using the cost model in \cref{subsec:cost_model}.
In addition, we show in \cref{lem:step_dm} that the cost to compute $\mat{U}_v$
in the density matrix algorithm under the assumption that each partition is contracted into a tensor is also $\Theta\left(a_vb_v^2\right)$. Since the multiplication $\mat{U}_v^T\mat{M}_v$ costs $\Theta\left(a_vb_vr\right)=O\left(a_vb_v^2\right)$, the cost equals the cost of the canonicalization-based algorithm, thus finishing the proof.
\end{proof}

\end{document}